\documentclass[a4paper]{article}
\usepackage[numbers,sort]{natbib} 

\usepackage[british]{babel}
\usepackage[T1]{fontenc}
\usepackage{textcomp}
\usepackage{floatflt}
\usepackage{url}
    
\usepackage{xcolor,amssymb,amsmath,xspace,suffix,enumerate,graphicx}
\usepackage{tikz,pgfplots}
\usetikzlibrary{positioning,arrows,automata,shapes,decorations,decorations.pathreplacing}

\usepackage{authblk}
\usepackage{amsthm}

\begin{document}

\title{Parity Game Reductions}

\author[1]{Sjoerd Cranen}
\author[2,3]{Jeroen J.A. Keiren}
\author[1]{Tim A.C. Willemse}

\affil[1]{Department of Computer Science and Mathematics\\
    Eindhoven University of Technology\\
    PO Box 513, 5600MB Eindhoven, The Netherlands\\
    \texttt{s.cranen,t.a.c.willemse@tue.nl}}
\affil[2]{Open University of the Netherlands\\
    Faculty of Management, Science \& Technology\\
    PO Box 2960, 6401 DL Heerlen, The Netherlands\\
    \texttt{Jeroen.Keiren@ou.nl}}
\affil[3]{Radboud University\\
    Institute for Computing and Information Sciences\\
    Nijmegen, The Netherlands}

\renewcommand\Authands{, and }
\renewcommand\Affilfont{\itshape\small}

\date{}

\tikzstyle{even}=[shape=diamond,draw,inner sep=0pt, minimum size=15pt, font={\scriptsize}]
\tikzstyle{odd}=[shape=rectangle,draw,inner sep=0pt, minimum size=11pt, font={\scriptsize}]


\newcommand{\SC}[1]{\todo[color=blue!70!cyan]{SC: #1}}
\newcommand{\TW}[1]{\todo[color=red!90!blue]{TW: #1}}
\newcommand{\JK}[1]{\todo[color=green!85!black]{JK: #1}}
\newcommand{\ie}{\emph{i.e.}\xspace}
\newcommand{\eg}{\textit{e.g.}\xspace}
\newcommand{\viz}{\textit{viz.}\xspace}
\newcommand{\etal}{\textit{et al.}\xspace}
\newcommand{\ibid}{\textit{ibid.}\xspace}
\newcommand{\true}{\mathsf{true}}
\newcommand{\false}{\mathsf{false}}
\newcommand{\nat}{\mathbb{N}}
\newcommand{\oftype}{\mathop{:}}
\renewcommand{\land}{\mathrel{\wedge}}
\renewcommand{\lor}{\mathrel{\vee}}
\let\existssym\exists
\let\forallsym\forall
\renewcommand{\exists}[2]{\existssym #1 \mathrel{:} #2}
\renewcommand{\forall}[2]{\forallsym #1 \mathrel{:} #2}
\newcommand{\ind}[2]{{#1}[#2]} 
\newcommand{\disjunion}{\sqcup}

\newcommand{\isdef}{\ensuremath{\stackrel{\Delta}{=}}}
\newcommand{\dom}[1]{\ensuremath{\textsf{dom}(#1)}}

\newcommand{\sosrule}[2]{\frac{\raisebox{.7ex}{\normalsize{$#1$}}}
  {\raisebox{-1.0ex}{\normalsize{$#2$}}}}
\newcommand{\project}[2]{\pi_{#1}(#2)}
\newcommand{\NP}{\mathsf{NP}}
\newcommand{\coNP}{\mathsf{co-NP}}
\newcommand{\UP}{\mathsf{UP}}
\newcommand{\coUP}{\mathsf{co-UP}}
\newcommand{\game}{\mathcal{G}}
\newcommand{\even}{\ifmmode \scalebox{0.7071}{\rotatebox{45}{$\Box$}}\else \emph{even}\xspace\fi}
\newcommand{\odd}{\ifmmode \Box \else \emph{odd}\xspace\fi}
\newcommand{\player}[1][]{\mathit{i}_{#1}} 
\newcommand{\opponent}[1]{\neg #1}
\newcommand{\priosym}{\Omega}
\newcommand{\prio}[1]{\priosym(#1)}
\newcommand{\getplayername}{\mathcal{P}}
\newcommand{\getplayer}[1]{\getplayername(#1)}
\newcommand{\priority}[1][]{\mathit{n}_{#1}} 
\newcommand{\post}[1]{\ensuremath{#1^\bullet}}
\renewcommand{\path}[1]{#1}
\newcommand{\pathvar}[1][]{p_{#1}}
\newcommand{\pathconcat}{}
\newcommand{\pathsname}{\Pi}
\newcommand{\treename}{\Psi}
\newcommand{\paths}[3][]{\pathsname_{#2}^{#1}(#3)}
\newcommand{\segpaths}[3][\strategyname]{\bar{\pathsname}_{#2}^{#1}(#3)}
\newcommand{\tree}[3]{\treename_{#1}^{#3}(#2)}
\newcommand{\strategy}[1]{\mathbb{S}_{#1}}
\newcommand{\memstrategy}[1]{\mathbb{S}_{#1}^*}
\newcommand{\allows}{\Vdash}
\newcommand{\consistent}[2]{\ensuremath{{#2}\allows{#1}}}
\newcommand{\strategyname}{\sigma}
\newcommand{\strategynamealt}{\psi}
\newcommand{\strategynamealta}{\psi'}
\newcommand{\step}{\to}
\newcommand{\nstep}{\not\to}
\newcommand{\steps}[1]{\mathop{\,\!_{#1}\smash\step}}
\newcommand{\nsteps}[1]{\mathop{\,\!_{#1}\smash\nstep}}
\newcommand{\force}{\mapsto}
\newcommand{\nforce}{\not\mapsto}
\newcommand{\forces}[3][]{\mathop{\,\!_{#2}\smash\force^{#1}_{#3}}}
\newcommand{\stutters}[1]{\forces{}{#1}}
\newcommand{\nforces}[3][]{\mathop{\,\!_{#2}\smash\nforce^{#1}_{#3}}}
\newcommand{\diverges}[2]{\forces{#1}{#2}}
\newcommand{\stutdiverges}[1]{\diverges{}{#1}}
\newcommand{\ndiverges}[2]{\nforces{#1}{#2}}

\newcommand{\C}{\mathcal{C}}
\renewcommand{\O}{\mathcal{O}}
\newcommand{\U}{\mathcal{U}}
\newcommand{\T}{\mathcal{T}}
\newcommand{\R}[1][]{\mathrel{R^{#1}}}
\renewcommand{\S}{\mathrel{S}}
\newcommand{\Rt}{\R^*} 
\newcommand{\RnaS}{\mathrel{(\R \circ \S)}} 
\newcommand{\RS}{\mathrel{(\R \cup \S)}}    
\newcommand{\RSt}{\mathrel{\RS^*}} 
\newcommand{\Q}{\mathrel{Q}}
\newcommand{\Qt}{\Q^*}
\newcommand{\inv}[1]{#1^{-1}} 
\newcommand{\delaygame}[3]{\mathsf{de-game}(#1,#2,#3)}
\newcommand{\gamevertex}[1][]{v_{#1}}
\newcommand{\obligation}[1][]{k_{#1}}
\newcommand{\gameposition}[1][]{p_{#1}}
\newcommand{\gameplayer}[1][]{\player[#1]}
\newcommand{\spoiler}{\ifmmode S\else\emph{Spoiler}\xspace\fi}
\newcommand{\duplicator}{\ifmmode D\else\emph{Duplicator}\xspace\fi}
\newcommand{\obligations}{K}
\newcommand{\configuration}[1][]{c_{#1}}
\newcommand{\configurations}{P}
\newcommand{\before}{\circ}
\newcommand{\Gstutname}{Governed stuttering bisimulation\xspace}
\newcommand{\gstutname}{governed stuttering bisimulation\xspace}
\newcommand{\eGstutname}{Governed stuttering bisimilarity\xspace}
\newcommand{\egstutname}{governed stuttering bisimilarity\xspace}

\newcommand{\isomorphic}{\cong}
\newcommand{\winner}{\sim_w}

\newlength{\simw}
\settowidth{\simw}{$\smash\sim$}
\newcommand{\dashes}[1]{%
  \rule{0.1#1}{0pt}\rule[0.05em]{0.16#1}{0.5pt}%
  \rule{0.16#1}{0pt}\rule[0.05em]{0.16#1}{0.5pt}%
  \rule{0.16#1}{0pt}\rule[0.05em]{0.16#1}{0.5pt}}
\def\gstut{\mathrel{\mathpalette\gstuti\relax}}
\def\gstuti#1#2{%
  \setbox0=\hbox{\raisebox{0.1em}{$\mathsurround=0pt#1\smash\sim$}}\usebox0\hskip-\wd0\dashes{\wd0}}
\newcommand{\stut}{\simeq}
\newcommand{\bisim}{\mathrel{\raisebox{0.15em}{\resizebox{\simw}{!}{$\leftrightarrow$}}\hspace{-\simw}{\rule{0.1\simw}{0pt}\rule{0.8\simw}{0.5pt}}}}
\newcommand{\gov}{\mathrel{\raisebox{0.15em}{\resizebox{\simw}{!}{$\leftrightarrow$}}\hspace{-\simw}{\dashes{\simw}}}}
\newcommand{\directsim}{\leq_{d}}
\newcommand{\strongdirectsim}{\leq_{sd}}
\newcommand{\strongdirectsimeq}{\equiv_{sd}}
\newcommand{\semistut}{\sim_{\mathrm{semi-st}}}
\newcommand{\delaysimc}{\leq_{de}}
\newcommand{\delaysimco}{\delaysimc^{o}}
\newcommand{\delaysimce}{\delaysimc^{e}}
\newcommand{\minsucc}[1]{\ensuremath{\min_{\directsim}(#1)}}
\newcommand{\maxsucc}[1]{\ensuremath{\max_{\directsim}(#1)}}

\newcommand{\bisimg}{\equiv_{\mathit{s}}}
\newcommand{\govg}{\equiv_{\mathit{g}}}
\newcommand{\directsimg}{\sqsubseteq_{\mathit{d}}}
\newcommand{\delaysim}{\sqsubseteq_{\mathit{de}}}
\newcommand{\delaysimo}{\delaysim^{o}}
\newcommand{\delaysime}{\delaysim^{e}}
\newcommand{\directsimeq}{\equiv_{\mathit{d}}}
\newcommand{\delaysimeq}{\equiv_{\mathit{de}}}
\newcommand{\delaysimoeq}{\delaysimeq^{o}}
\newcommand{\delaysimeeq}{\delaysimeq^{e}}
\newcommand{\stutg}{\equiv_{\mathit{st}}}
\newcommand{\gstutg}{\equiv_{\mathit{g,st}}}

\newcommand{\distname}{\mathit{dist}}
\newcommand{\dist}[4]{\distname_{#1,#2}(#3,#4)}
\newcommand{\measuresym}{\mathit{m}}
\newcommand{\measure}[1]{\measuresym(#1)}
\newcommand{\frontname}{\textit{expel}}
\newcommand{\front}[1]{\ensuremath{\frontname(#1)}}
\newcommand{\exitname}{\textit{exit}}
\newcommand{\exit}[1]{\ensuremath{\exitname(#1)}}

\newcommand{\rankname}{\mathsf{rank}}
\newcommand{\rank}[3]{(#1,#2,#3)}

\newcommand{\mimickname}{\mathsf{mimic}}
\newcommand{\mimick}[1]{\mimickname_{#1}}
\newcommand{\divr}[2]{\mathsf{div}_{#1}(#2)}
\newcommand{\entry}[2]{\mathsf{reach}_{#1}(#2)}

\newcommand{\targetclass}[2]{\mathsf{targetclass}_{#1}(#2)}
\newcommand{\target}[2]{\tau_{#1} (#2)}
\newcommand{\vertexorder}{\ensuremath{\sqsubset}}
\newcommand{\targetorder}[1]{\ensuremath{\prec_{#1}}}
\newcommand{\vertexordermin}{\sqcap}
\newcommand{\targetordermin}[1]{\curlywedge_{#1}}

\newcommand{\prioorder}{\prec}
\newcommand{\prioordereq}{\preccurlyeq}
\newcommand{\playerorder}{\lessdot}
\newcommand{\playermin}{\min_{\playerorder}}
\newcommand{\fmistorder}{\lll}
\newcommand{\fmistmin}{\min_{\fmistorder}}
\newcommand{\update}[1][]{\gamma^{#1}}

\newcommand{\partition}[2]{{#1}_{/#2}}
\newcommand{\class}[2]{[#1]_{#2}}
\newcommand{\partitioneq}[1]{\mathrel{#1}}

\newcommand{\modcompatible}[2]{#1 \equiv_{2}#2}
\newcommand{\nmodcompatible}[2]{#1 \not\equiv_{2} #2}
\newcommand{\attrsym}{\ensuremath{\mathit{Attr}}}
\newcommand{\attr}[3][]{\ensuremath{\attrsym^{#1}_{#2}(#3)}}
\newcommand{\battr}[4][]{\ensuremath{{_{#2}}\!\attrsym^{#1}_{#3}(#4)}}
\newcommand{\bottomsym}{\ensuremath{\mathit{Bottom}}}
\newcommand{\bbottom}[3]{\ensuremath{{_{#1}}\!\bottomsym_{#2}(#3)}}
\newcommand{\restrict}[2]{#1 \upharpoonright#2}
\newcommand{\leavesym}{\ensuremath{\mathit{Leave}}\xspace}
\newcommand{\leave}[4][]{\ensuremath{\leavesym^{#1}_{#2}(#3, #4)}}
\newcommand{\possym}{\ensuremath{\mathit{pos}}\xspace}
\newcommand{\pos}[3]{\ensuremath{\possym_{#1}(#2, #3)}}
\newcommand{\becomes}{\ensuremath{\leftarrow}}
\newcommand{\todolist}{\ensuremath{\mathit{todo}}\xspace}
\newcommand{\head}[1]{\ensuremath{\mathit{head}(#1)}}
\newcommand{\pop}[1]{\ensuremath{#1.\mathit{pop}()}}
\newcommand{\append}[2]{\ensuremath{#1.\mathit{append}(#2)}}
\newcommand{\remove}[2]{\ensuremath{#1.\mathit{remove}(#2)}}
\newcommand{\lhead}[1]{\ensuremath{\mathit{head}(#1)}}
\newcommand{\stable}{\ensuremath{\mathit{stable}}\xspace}
\newcommand{\foundsplitter}{\ensuremath{\mathit{foundsplitter}}\xspace}
\newcommand{\inert}[1]{\ensuremath{#1.\mathit{incoming}}}
\newcommand{\noninert}[1]{\ensuremath{#1.\mathit{incoming}}}
\newcommand{\inerttononinert}{\ensuremath{\mathit{inert\_becomes\_non\_inert}}\xspace}
\newcommand{\trysplit}[2]{\ensuremath{\mathit{TrySplit}(#1, #2)}}

\theoremstyle{plain}
\newtheorem{theorem}{Theorem}
\newtheorem{lemma}{Lemma}
\newtheorem{proposition}{Proposition}
\newtheorem{corollary}{Corollary}

\theoremstyle{definition}
\newtheorem{definition}{Definition}
\newtheorem{example}{Example}

\def\precdot{{\ensuremath{<\!\!\!\cdot}}}

\maketitle

\begin{abstract}
Parity games play a central role in model checking and satisfiability checking. Solving parity games is computationally expensive, among others due to the size of the games, which, for model checking problems, can easily contain $10^9$ vertices or beyond. Equivalence relations can be used to reduce the size of a parity game, thereby potentially alleviating part of the computational burden. We reconsider (governed) bisimulation and (governed) stuttering bisimulation,  and we give detailed proofs that these relations are equivalences, have unique quotients and they approximate the winning regions of parity games. Furthermore, we present game-based characterisations of these relations. Using these characterisations our equivalences are compared to relations for parity games that can be found in the literature, such as direct simulation equivalence and delayed simulation equivalence. To complete the overview we develop coinductive characterisations of direct- and delayed simulation equivalence and we establish a lattice of equivalences for parity games. 
\end{abstract}

\section{Introduction}

We study preorders and equivalences defined on \emph{parity games}. Such games are
turn-based graph games between two players taking turns pushing a
token along the vertices of a finitely coloured graph. These players,
called \even and \odd, strive to optimise the parity of the dominating
colour occurring infinitely often in a play.
Parity games appear in the core of various foundational results
such as Rabin's proof of the decidability of a monadic second-order
theory.  Solving parity games is a computationally expensive but
key step in many model checking algorithms~\cite{EJS:01,SS:98,Sti:99} and synthesis
and supervisory control algorithms~\cite{AVW:03,AW:08,FL:13}. 

Parity game solving enjoys a special status among combinatorial
optimisation problems, being one of the rare problems in the intersection
of the UP and coUP classes \cite{Jur:98} that is not known to be in P.
Despite the continued research effort directed to it,
resulting in numerous algorithms for solving parity games, see, \eg,~\cite{McN:93,Zie:98,VJ:00,PV:01,BSV:03,BV:05,Jur:00,JPZ:06,SS:98,Sch:07},
no polynomial time algorithm has yet been found.  

Orthogonally to
the algorithmic improvements, heuristics and static analyses have been devised that may
speed up solving, or fully solve parity games that occur in practice~\cite{FL:09,HKP:13,HKP:16}. Such
heuristics work particularly well for verification problems, which
give rise to games with only few different priorities.  
In a similar vein, heuristics based on the intimate ties between temporal
logics and bisimulation relations are often exploited to speed-up
model checking. First minimising a state space by computing the
equivalence quotient and only then analysing this quotient can be
an effective strategy, see \eg~\cite{KKZJ:07}. 

Given the close connection between parity game solving and model checking,
a  promising heuristic
in this setting is to minimise a game prior to solving
it.  Of course, this requires that the winning regions of the
original game can be recovered cheaply from the winning regions of the 
minimised game. Moreover,
minimisation makes sense only for equivalence relations
that strike a favourable balance between their power to compress
the game graph and the computational complexity of quotienting with
respect to the equivalence relation.  Indeed, in~\cite{CKW:11,KW:09}
we showed that quotienting using standard strong bisimilarity and
stuttering equivalence allow to solve parity games that could not
be solved otherwise. Despite the immense reductions that can be
obtained, the results were mixed and, apart from a number of cases
that become solvable, there was on average no clear gain from
using such relations.  It should be noted that the stuttering
equivalence experiments in \cite{CKW:11,CKW:12} were conducted using the Groote-Vaandrager
algorithm~\cite{GV:90} which runs in $\mathcal{O}(m n)$,
where $m$ is the number of edges and $n$ is the number of states.
A recent improvement on this algorithm, described in~\cite{GrW:16},
may very well mean the scale tips in favour of using stuttering
equivalence minimisation prior to solving a parity game, as experiments using this 
$\mathcal{O}(m \log{n})$ algorithm have shown speed-ups of
several orders of magnitude compared to the $\mathcal{O}(m n)$
algorithm. 

Similar observations can be made for \emph{governed
bisimilarity}~\cite{Kei:13} (also known as \emph{idempotence-identifying
bisimilarity} in~\cite{KW:09} and \emph{governed stuttering
bisimilarity}~\cite{CKW:12}, which weaken strong bisimilarity and
stuttering equivalence, respectively, by taking the potentials of
players into account. Quotienting for the latter relations relies
on the claim that the relations are equivalences.

As a side-note, simulation and bisimulation relations, tailored to
parity games, may lead to insights into the core of the parity game
solving problem. Indeed, in \eg~\cite{Jan:05}, Janin relies on
different types of simulations to provide uniform proofs when showing
the existence of winning strategies; at the same time he suggests
simulation relations may ultimately be used to solve games efficiently.

\paragraph{Contributions.}
In this paper,
we revisit the notions of (governed) bisimilarity and (governed) stuttering
bisimilarity for parity games from \cite{KW:09,Kei:13,CKW:11,CKW:12}.
We give formal proofs that they are indeed
equivalence relations and, equally important, that they approximate
the winning regions of a parity game, substantiating our claims in
the aforementioned papers. Showing that the relations are indeed
equivalence relations is technically rather involved, and slight
oversights are easily made, see~\eg~\cite{Bas:96}, and the added
complexity of working in a setting with two players complicates
matters significantly.

We furthermore study how our equivalence relations are related to
two other notions that have been studied in the context of parity
games, \viz{} \emph{direct simulation} and \emph{delayed
simulation}~\cite{FW:06} and the latter's \emph{even} and
\emph{odd}-biased versions.  A complicating factor is the fact that
these relations have only \emph{game-based} definitions, whereas
our equivalences are defined \emph{coinductively}. We mend this by
providing alternative coinductive definitions for direct simulation
and delayed simulation, inspired by~\cite{Nam:97},
and we show that these coincide with their
game-based definitions. Likewise, we give game-based definitions
for our coinductively defined relations, drawing inspiration
from~\cite{BKZ:07,YFHHT:14}, thereby offering a more
operational view on our relations.  

Finally, we show that, contrary
to (even- and odd-biased) delayed simulation equivalence, direct
simulation equivalence, governed bisimilarity and governed stuttering
bisimilarity have unique quotients.

\paragraph{Related work.}
In logic, bisimulation has
been used to characterise the subfamily of first-order logic that
is definable in modal logic~\cite{vBe:84}, and which fragment of
monadic second-order logic is captured by the modal $\mu$-calculus.
Bisimulation and simulation-like relations, called \emph{consistent
correlations}~\cite{Wil:10} and \emph{consistent consequence}~\cite{GW:12}
for PBESs, a fixpoint-logic based framework which is closely related
to parity games, were imperative to prove the soundness of the
syntax-based static analysis techniques described
in~\cite{OWW:09,OW:10,KWW:14,CGWW:15}. Various simulation relations have been
used successfully for minimising B\"uchi automata, see \eg~\cite{Cle:11,MC:13,EWS:05}.

In the context of process theory, there is an abundance of different
simulation and bisimulation relations, allowing to reason about the
powers of different types of observers of a system's behaviour,
see~\cite{vGla:90,vGla:93}.  Coinductive definitions of weak behavioural equivalences
such as stuttering equivalence (which is, essentially, 
the same as \emph{branching bisimulation} for labelled transition systems) are
commonplace, see~\cite{vGla:93} for a comprehensive overview.
Typically, these definitions rely on the transitive closure of the
transition relation. As argued by Namjoshi \cite{Nam:97}, local
reasoning typically leads to simpler arguments. He therefore
introduced well-founded bisimulation, a notion equivalent to
stuttering bisimulation which solely relies on local reasoning by
introducing a well-founded order into the relation. Still, at its
basis, well-founded bisimulation only serves to show the reachability
of some pair of related vertices. In our coinductive characterisation
of the delayed simulation of~\cite{FW:06}, we use Namjoshi's ideas. However, we need to
factor in that in delayed simulation each step on one side must be
matched by exactly one step on the simulating side.

There are only a few documented attempts that provide game-based definitions for
weak behavioural equivalences. Yin \etal describe branching bisimulation games for normed
process algebra \cite{YFHHT:14}. A game-based characterisation of
divergence-blind stuttering bisimulation was provided by Bulychev
\etal \cite{BKZ:07}. Neither of these definitions is easily
extended to the setting of governed stuttering bisimulation for
parity games. In particular, the latter definition is only sound
for transition systems that are free of divergences and requires a separate
preprocessing step to deal with these.
For governed stuttering bisimulation, it is unclear how the parity game
should be preprocessed, so instead we incorporate divergence into the game-based definition
as a first-class citizen.


%
%
%
%
%
%

\paragraph{Structure of the paper.} Parity games 
are introduced in Section~\ref{sec:parity_games}.
In Section~\ref{sec:notation}, we introduce notation that facilitates
us to define preorders and equivalences on parity games and we state
several basic results concerning this notation. A technical overview
of the relations studied in the remainder of the paper, and how
these are related is presented in Section~\ref{sec:lattice}. In
Section~\ref{sec:preorders} we study direct simulation, delayed
simulation and its biased versions and in Section~\ref{sec:bisimulations}
we study governed bisimulation and governed stuttering bisimulation.
Quotienting, for all involved equivalences that admit unique
quotients, is discussed in Section~\ref{sec:quotients} and in
Section~\ref{sec:comparison} we return to, and substantiate, the
overview we presented in Section~\ref{sec:lattice}. We wrap up with
conclusions and an outlook for future work in Section~\ref{sec:conclusions}.


\section{Parity Games}\label{sec:parity_games}

A parity game is a two-player graph game, played by two players \even and \odd (denoted $\even$ and $\odd$) on a total directed graph in which the vertices are partitioned into two sets, one for each player, and in which a natural priority is assigned to every vertex. The game is played by placing a token on some initial vertex, and if the token is on a vertex owned by player \even, then she moves the token to a successor of the current vertex (likewise for vertices owned by \odd).
The game is formally defined as follows.
\begin{definition}[Parity game]
  A parity game is a directed graph $(V, \to, \priosym, \getplayername)$,
  where
  \begin{itemize}
    \item $V$ is a finite set of vertices,
    \item ${\to} \subseteq V \times V$ is a total edge relation (\ie, for each 
    $v \in V$ there is at least one $w \in V$ such that $(v,w) \in {\to}$),
    \item $\priosym \oftype V \to \nat$ is a priority function that assigns
    priorities to vertices,
    \item $\getplayername \oftype V \to \{\even,\odd\}$ is a function assigning 
    vertices to players.
  \end{itemize}
\end{definition}
Instead of $(v,w) \in \to$ we typically write $v \to w$, and we write
$\post{v}$ for the set $\{w \in V ~|~ v \to w\}$.
If $\player$ is a player, then $\opponent{\player}$ denotes the
opponent of $\player$, \ie, $\opponent{\even}=\odd$ and $\opponent{\odd}=\even$.
The function $\getplayername$ induces a partitioning of $V$ into a set
of vertices $V_{\even}$ owned by player \even and a set of vertices $V_{\odd}$
owned by player \odd; we use $\getplayername$ and $V_{\even}, V_{\odd}$
interchangeably.
The \emph{reward order} on natural numbers is defined such that $n \prioordereq m$ if $n$ is even and $m$ is odd; or $n$ and $m$ are even and $n \leq m$, or $n$ and $m$ are odd and $m \leq n$. Note that $n \prioorder m$ means that $n$ is \emph{better than} $m$ for player \even. Notions $\min$ and $\max$ are always used with respect to the standard ordering on natural numbers.
Finally, we remark that the assumption that the edge relation is total only
serves to simplify the theory described in this paper. All results
can be generalised to deal with the situation in which one of the
players is unable to move.

\paragraph{Paths.}
A sequence of vertices $v_0 \ldots v_n$ for which
$v_m \to v_{m+1}$ for all $m < n$ is a \emph{path}. The concatenation 
$\pathvar[1] \pathconcat \pathvar[2]$ of paths $\pathvar[1]$ and $\pathvar[2]$ 
is again a path, provided there is
a step from the last vertex in $\pathvar[1]$ to the first vertex in 
$\pathvar[2]$. Infinite paths are
defined in a similar manner. We use $\ind{\pathvar}{j}$ to denote the 
$j^\textrm{th}$ vertex in a path $\pathvar$, counting from $0$. 
The set of paths of length $n$ 
starting in $v$ is defined inductively for $n \geq 1$ as follows:
\begin{align*}
\paths[1]{}{v} & \isdef \{ \path{v} \} \\
\paths[n+1]{}{v} & \isdef \{ \pathvar \pathconcat \path{ u } 
\mid \pathvar \in \paths[n]{}{v} \land \ind{\pathvar}{n} \to u \}
\end{align*}
The set of infinite paths starting in $v$ is denoted 
$\paths[\omega]{}{v}$, and the set of
both finite and infinite paths starting in $v$ is defined as follows:
\begin{equation*}
\paths{}{v} \isdef \paths[\omega]{}{v} \cup \bigcup_{n \in \nat} \paths[n]{}{v}
\end{equation*}

\paragraph{Plays and their winners.}
A game starts by placing a token on some vertex $v \in V$.  Players
move the token indefinitely according to the following simple rule:
if the token is on some vertex $v$, player $\getplayer{v}$ moves
the token to some vertex $w$ such that $v \to w$. The result is an
infinite path $\pathvar$ in the game graph; we refer to this infinite
path as a \emph{play}.  The \emph{parity} of the lowest priority
that occurs infinitely often on $\pathvar$ defines the \emph{winner}
of the play. If this priority is even, then player $\even$ wins,
otherwise player $\odd$ wins.

\paragraph{Strategies.}
A \emph{strategy} for player $\player$ is a partial function $\strategyname\oftype
V^{*\!\!}\to V$, that is defined only for paths ending in a vertex owned by
player $\player$ and determines the next vertex to be played onto. The set
of strategies for player $\player$ in a game $\game$ is denoted
$\memstrategy{\game,\player}$, or simply $\memstrategy{\player}$ if 
$\game$ is clear from the context.  If a strategy
yields the same vertex for every pair of paths that end in the same
vertex, then the strategy is said to be \emph{memoryless}. The set
of memoryless strategies for player $\player$ in a game $\game$ is
denoted $\strategy{\game,\player}$, abbreviated to $\strategy{\player}$
when $\game$ is clear from the context. A memoryless strategy
is usually given as a partial function $\strategyname\oftype V\to V\!$.

A strategy $\strategyname\in\memstrategy{\player}$ \emph{allows} a path $\pathvar$ of length $n$, denoted $\strategyname 
\allows \pathvar$, if and only if for all $j < n-1$ it is the case that if
$\strategyname$ is defined for $\path{\ind{\pathvar}{0} \ldots 
\ind{\pathvar}{j}}$, then $\ind{\pathvar}{j+1} = 
\strategyname(\path{\ind{\pathvar}{0} \ldots 
  \ind{\pathvar}{j}})$. The definition of consistency is 
extended to infinite paths in the obvious manner. We generalise the
definition of $\pathsname$ to paths allowed by a strategy $\strategyname$;
formally, we define:
\begin{equation*}
\paths[n]{\strategyname}{v} \isdef \{ \pathvar \in \paths[n]{}{v} \mid 
\strategyname \allows \pathvar \}
\end{equation*}
The definition for infinite paths is generalised in the same way and
denoted $\paths[\omega]{\strategyname}{v}$.
By $\paths{\strategyname}{v}$ we denote $\paths[\omega]{\strategyname}{v} \cup
\bigcup_{n \in \nat}
\paths[n]{\strategyname}{v}$, \ie, the set of all finite and infinite paths
starting in $v$ and allowed by $\strategyname$.

A strategy $\strategyname\in\memstrategy{\player}$ is said
to be a \emph{winning strategy} from a vertex $v$ if and only if $\player$ is the 
winner of every path allowed by $\strategyname$. A vertex is won by player $\player$ if $\player$
has a winning strategy from that vertex.

\paragraph{Solving parity games.}
It is well-known that parity games are determined, \ie that each vertex
in a game is won by exactly one player, and if a winning strategy for a
player exists from a vertex, then also a memoryless strategy exists.
This is summarised in the following theorem.

\begin{theorem}[Memoryless determinacy \cite{EJ:91}]
For every parity game there is a unique partition $(W_{\even}, W_{\odd})$
such that winning strategies $\strategyname_{\even} \in \memstrategy{\even}$
from $W_{\even}$ and $\strategyname_{\odd} \in \memstrategy{\odd}$ from
$W_{\odd}$ exist. Furthermore, if $\strategyname_{\player} \in \memstrategy{\player}$
is winning from $W_{\player}$ a memoryless strategy $\strategynamealt_{\player} \in
\strategy{\player}$ winning from $W_{\player}$ exists.
\end{theorem}
The problem of \emph{solving} a parity game is defined as the problem
of computing the winning partition $(W_{\even}, W_{\odd})$ of a
parity game.

\section{Notation}
\label{sec:notation}

In the remainder of this paper we frequently need to reason about
the concept of a player being able to force play towards a set of
vertices.  We introduce notation that facilitates such reasoning
and we provide some lemmata that express basic properties of parity
games in terms of this extended notation.
Throughout this section, we fix a parity game
$(V, \to, \priosym, \getplayername)$. Furthermore, we let
$T, U \subseteq V$ be subsets of vertices in the game.\medskip

Given a memoryless strategy $\strategyname$, we introduce a single-step relation
$\steps{\strategyname} \subseteq \to$ that contains only those edges allowed by $\strategyname$:%
\[
  \steps{\strategyname} \isdef \{ (v, u) \mid (v, u) \in \to \text{ and if } \strategyname(v) \text{ is defined } \strategyname(v) = u \}
\]
In line with $v \to u$, we write $v \steps{\strategyname} u$ if $(v, u) \in \steps{\sigma}$.
%
Abstracting from the specific strategy, we write $v \steps{\player} u$ iff
player $\player$ has a memoryless strategy $\strategyname$ such that
$v \steps{\strategyname} u$.

We introduce special notation to express which parts of the graph can be reached from a certain node. We use $v \stutters{U} T$ to denote that there is a finite path $\path{v_0 \ldots v_n}$, for some $n$, such that $v = v_0$, $v_n \in T$  and for all $j < n$, $v_j \in U$. Conversely, $v \stutdiverges{U}$ denotes the existence of an infinite path $\path{v_0\ v_1 \ldots}$ for which $v = v_0$ and for all $j$, $v_j \in U$.

We extend this notation to restrict this reachability analysis to plays that can be enforced by a specific player. We say that strategy $\strategyname$ forces the play from $v$ to $T$ via $U$, denoted $v \forces{\strategyname}{U} T$, if and only if for all plays $p$ starting in $v$ such that $\strategyname \allows p$, there exists an $n$ such that $\ind{p}{n} \in T$ and $\ind{p}{j} \in U$ for all $j<n$. Note that, in particular, $v \forces{\strategyname}{U} T$ if $v \in T$. Similarly, strategy $\strategyname$ forces the play to diverge in $U$ from $v$, denoted $v \diverges{\strategyname}{U}$, if and only if for all such plays $p$, $\ind{p}{j}\in U$ for all $j$.

Finally, if we are not interested in a particular strategy, but only in the \emph{existence} of a strategy for a player $\player$ via which certain parts of the graph are reachable from $v$, we replace $\strategyname$ by $\player$ in our notation to denote an existential quantification over memoryless strategies:
\begin{align*}
v\forces{\player}{U}T &\isdef \exists{\strategyname\in\strategy{\player}}{v \forces{\strategyname}{U} T} &
v\diverges{\player}{U} &\isdef \exists{\strategyname\in\strategy{\player}}{v \diverges{\strategyname}{U}}
\end{align*}
The lemma below shows that rather than using memoryless strategies, one may, if needed, use arbitrary strategies when reasoning about $v \forces{\player}{U} T$.
\begin{lemma}\label{lem:mem_vs_memless}
$\exists{\strategyname\in\strategy{\player}}{v \forces{\strategyname}{U} T}$ iff
$\exists{\strategyname\in\memstrategy{\player}}{v \forces{\strategyname}{U} T}$.
\end{lemma}
\begin{proof}
Observe that the implication from left to right holds by definition. So assume that
for some $\strategyname\in\memstrategy{\player}$, we have $v \forces{\strategyname}{U} T$.
Note that $v \forces{\strategyname}{U} T$ iff $v \forces{\strategyname}{U \setminus T} T$.
The truth value of the latter predicate does not depend on priorities of the vertices and only depends on the edges that originate in $U \setminus T$. Therefore, the truth value of this predicate will not change if we apply the following transformation to our graph:
\begin{itemize}
\item for all $u \in T$, replace all outgoing edges by a single edge $u \to u$.
\item set the priorities for all $u \in T$ to even iff $\player = \even$ and the priorities
of all other vertices to odd iff $\player = \even$.

\end{itemize}
Since $v \forces{\strategyname}{U \setminus T} T$, vertex $v$  is won by $\player$ in the resulting graph. As
parity games are memoryless determined, $\player$ must have a memoryless strategy
to move from $U \setminus T$ to $T$ in the resulting graph. Hence there is some $\strategyname' \in \strategy{\player}$
such that $v \forces{\strategyname'}{U \setminus T} T$ in the resulting graph,
but then also $v \forces{\strategyname'}{U \setminus T} T$ in the original graph, and
hence also the required $v \forces{\strategyname'}{U} T$.\qedhere
\end{proof}

The complement of these relations is denoted by a slashed version of the corresponding arrow, \eg, $\neg v \forces{\player}{U} T$ can be written $v \nforces{\player}{U} T$. We extend the transition relation of the parity game to sets and to sets of sets in the usual way, \ie, if $T$ is a set of vertices, and $\U$ is a set of vertex sets, then
\begin{align*}
v \to T &\isdef \exists{u\in T}{v \to u} &
v \to \U &\isdef v \to \bigcup\U 
\end{align*}
All other arrow notation is extended in the same way; if a set of sets $\U$ is given as a parameter, it is interpreted as the union of $\U$.\medskip

The notation $v \forces{\player}{U} T$ is closely related to the notion of \emph{attractor sets}~\cite{McN:93}. To formalise this correspondence, we generalise the attractor set definition along the lines of the generalisation used for the computation of the \emph{Until} in the alternating-time temporal logic ATL~\cite{AHK:02}.
\begin{definition}[Attractor set]
\label{def:attr}
We define $\battr{U}{\player}{T}$ as
$\battr[\omega]{U}{\player}{T}$ where:
\[
\begin{array}{lcl}
\battr[0]{U}{\player}{T} & \isdef & T \\
\battr[n+1]{U}{\player}{T} &\isdef & \battr[n]{U}{\player}{T}\\
& \cup & \{ v \in U \mid \getplayer{v} = \player \land \exists{v' \in \post{v}}{v' \in
\battr[n]{U}{\player}{T}} \} \\
& \cup & \{ v \in U \mid \getplayer{v} \neq i \land
\forall{v' \in \post{v}}{v' \in \battr[n]{U}{\player}{T}} \} \\
\end{array}
\]
\end{definition}
The attractor set as defined in~\cite{McN:93} is obtained for $U = V$.
In essence, the attractor set $\battr{U}{\player}{T}$ captures the subset of $U \mathop{\cup} T$ 
from which $\player$ can force the game to $T \subseteq V$, by staying within 
$U$ until $T$ is reached. Note that $\attrsym{}$ is a monotone operator; \ie 
for $T \subseteq T'$ we have
$\battr{U}{\player}{T} \subseteq \battr{U}{\player}{T'}$. The correspondence between our `forcing' arrow
notation and the (generalised) attractor is given by the following lemma.
\begin{lemma}
  \label{lem:attractor_vs_forces}
  Let $U,T \subseteq V$. Then $v \forces{\player}{U} T$ iff $v \in \battr{U}{\player}{T}$.
\end{lemma}
\begin{proof} We first introduce some additional notation.
Let $v \forces[n]{\player}{U} T$ denote that there is a 
$\strategyname \in \strategy{\player}$ such that for all 
$\pathvar \in \paths[n+1]{\strategyname}{v}$,
there is some $m \le n$ such that $\ind{\pathvar}{m} \in T$,
and $\ind{\pathvar}{j} \in U$ for all $j < m$. Note that 
$v \forces[\omega]{\player}{U} T$ iff $v \forces{\player}{U} T$.
We can now use induction to prove $u \forces[n]{\player}{U} T$ iff 
$u \in \battr[n]{U}{\player}{T}$ for all $u$. The required property then follows.

The base case, $n = 0$, follows instantly. For $n = m+1$, our induction
hypothesis yields $u \forces[m]{\player}{U} T$ iff $u \in \battr[m]{U}{\player}{T}$ for
all $u$.
We distinguish two cases: $\getplayer{u} = \player$ and $\getplayer{u} \not=\player$.

Suppose $\getplayer{u} = \player$, and assume $u \in \battr[n+1]{U}{\player}{T}$.
Then, by definition,
$u \in \battr[n]{U}{\player}{T}$, or, since $\getplayer{u} = \player$, 
$u \in \{ v \in U \mid \exists{v' \in \post{v}}{v' \in \battr[n]{U}{\player}{T}} \}$.
But then this is equivalent to
\[
\tag{*}
\label{form_star}
u \in \battr[n]{U}{\player}{T} \text{ or }
v' \in \battr[n]{U}{\player}{T} \text{ for some $v' \in \post{u}$}.
\]
By our induction
hypothesis, \eqref{form_star} is equivalent to
$u \forces[n]{\player}{U} T$ or
$v' \forces[n]{\player}{U} T$ for some $v' \in \post{u}$. But since $\player$ then
has a strategy to play $u \to v'$, this is again equivalent to
\[
\tag{$\dagger$}
\label{form_dagger}
u \forces[n]{\player}{U} T \text{ or }
u \forces[n+1]{\player}{U} T
\]
Since $u \forces[n]{\player}{U} T$ implies $u \forces[n+1]{\player}{U}
T$ and  $u \forces[n+1]{\player}{U} T$ implies $u \forces[n]{\player}{U}
T$ or $u \forces[n+1]{\player}{U} T$, we find that~\eqref{form_dagger}
is equivalent to the desired $u \forces[n+1]{\player}{U} T$, which finishes
the proof for the case $\getplayer{u} = \player$. The case for $\getplayer{u}
\not= \player$ uses a similar line of reasoning; the only difference is
that we do not need to identify a strategy for player $\player$.\qedhere
\end{proof}

We are now ready to formalise some intuitions using our notation.
One of the most basic properties we expect to hold is that a player can force 
the play towards some given set of vertices, or otherwise her opponent can force 
the play to the complement of that set. 

In the following lemmas, let $v\in V$, $U,T,T'\subseteq V$ and $\player$ a player.

\begin{lemma}\label{lem:some-player-forces-simple} 
$v \forces{\player}{U} T \lor v \forces{\opponent{\player}}{U} V \setminus T.$
\end{lemma}
\begin{proof}
We prove the equivalent $v \nforces{\player}{U} T \implies v \forces{\opponent{\player}}{U}
V \setminus T.$ Assume that $v \nforces{\player}{U} T$. Observe that it follows directly that
$v \not \in T$, and hence $v \in V \setminus T$. Therefore, we immediately find that $v \forces{\opponent{\player}}{U}
V \setminus T.$\qedhere
\end{proof}

In a similar train of thought, we expect that if from a single vertex, 
each player can force play towards some target set, then the players' target 
sets must overlap.
\begin{lemma}\label{lem:force-contradiction1}
$
\begin{array}[t]{l}
v \forces{\player}{U} T \land v \forces{\opponent{\player}}{U} T' \implies \\
\qquad \exists{u\in T,u'\in T'}{u = u' \lor u \in U \lor u' \in U}.
\end{array}$
\end{lemma}
\begin{proof} 
Assume $v \forces{\player}{U} T \land v \forces{\opponent{\player}}{U} T'$. 
Then there must be strategies $\strategyname\in\strategy{\player}$ and 
$\strategyname'\in\strategy{\opponent{\player}}$ such that $v 
\forces{\strategyname}{U} T \land v \forces{\strategyname'}{U} 
T'$.
Let $\strategyname$ and $\strategyname'$ be such, and consider a play 
$\pathvar$ such that $\strategyname \allows \pathvar$ and $\strategyname' 
\allows \pathvar$. For this play, there must be $m$ and $n$ such that 
$\ind{\pathvar}{m} \in T \land \forall{j<m} \ind{\pathvar}{j} \in U$, and 
$\ind{\pathvar}{n} \in T' \land \forall{j<n} \ind{\pathvar}{j} \in U$. If 
$m=n$, then this witnesses $\exists{u\in T,u'\in T'} u=u'$. If $m<n$, then 
$\ind{\pathvar}{m} \in T \land \ind{\pathvar}{m} \in U$, and if $n<m$, then 
$\ind{\pathvar}{n} \in T' \land \ind{\pathvar}{n} \in U$. \qedhere\end{proof}

The above lemmata reason about players being able to reach sets of vertices. The 
following lemma is essentially about avoiding sets of vertices: it states that
if one player can force divergence within a set, then this is the same as saying that the 
opponent cannot force the play outside this set.

\begin{lemma}
  \label{lem:force-vs-div}
  \label{lem:force-contradiction2}
  $
  v \diverges{\player}{U}
  \iff
  v \nforces{\opponent{\player}}{U} V \setminus U
  $
\end{lemma}
\begin{proof}
Again, note that the truth values of $v \diverges{\player}{U}$ and $v \nforces{\opponent{\player}}{U} V \setminus U$ only depend on edges that originate in $U$, and that these truth values do not depend on priorities at all. Therefore, the truth value of these predicates will not change if we apply the following transformations to our graph:
\begin{itemize} 
\item For all $u \in V \setminus U$, replace all outgoing edges by a single edge $u \to u$.
\item Make the priorities of all vertices in $U$ such that they are even iff $\player=\even$, and the priorities of all other vertices odd iff $\player=\even$.
\end{itemize}
In the resulting graph, player $\player$ wins if and only if $v \diverges{\player}{U}$, and player $\opponent{\player}$ wins if and only if $v \forces{\opponent{\player}}{U} V \setminus U$. Since parity games are determined, \ie $v$ can only be won by one player, the desired result follows. \qedhere\end{proof}%

Next, we formalise the idea that if a player can force the play to a
first set of vertices, and from there he can force the play to a second set of
vertices, then he must be able to force the play to that second set. 
\begin{lemma}\label{lem:glue-strategies}
  $
  (v \forces{\player}{U} T \land \forall{u \in T}{u\forces{\player}{U} T'})
    \implies v \forces{\player}{U} T'
  $
\end{lemma}
\begin{proof}
By Lemma~\ref{lem:attractor_vs_forces}, $\forall{u \in T}{u \forces{\player}{U} T'}$ implies $T \subseteq \battr{U}{\player}{T'}$.
By monotonicity of $\attrsym{}$ we have $\battr{U}{\player}{T} \subseteq \battr{U}{\player}{\battr{U}{\player}{T'}}$.
Since $\battr{U}{\player}{\battr{U}{\player}{T'}} = \battr{U}{\player}{T'}$, we thus have $\battr{U}{\player}{T} \subseteq \battr{U}{\player}{T'}$.
By the same token, from $v \forces{\player}{U} T$ we find $v \in \battr{U}{\player}{T}$. 
Combined, we find $v \in \battr{U}{\player}{T'}$ which, by Lemma~\ref{lem:attractor_vs_forces}, yields the desired $v \forces{\player}{U} T'$. \qedhere\end{proof}
Finally, we state two results that relate a player's capabilities to reach a set of vertices
to the capabilities of the vertices that are able to leave a set of vertices in a single step.
\begin{lemma}\label{lem:exits}
Let $S = \{ u \in U ~|~ \post{u} \cap T \neq \emptyset \}$ and $v \notin T$. 
Then $v \forces{\player}{U} T$ implies
$V_{\player} \cap S \neq \emptyset$ or
for some $u \in S$,  $\post{u} \subseteq T$.
\end{lemma}
\begin{proof} Assume $v \forces{\player}{U} T$ and suppose $V_{\player}
\cap S = \emptyset$. Assume that for all $u \in
S$, not $\post{u} \subseteq T$.  Then $\battr{U}{\player}{T} = T$
follows from $S \subseteq V_{\opponent{\player}}$.  Since $v \notin
T$ this implies $v \notin \battr{U}{\player}{T}$.  By
Lemma~\ref{lem:attractor_vs_forces} we then have $v \nforces{\player}{U} T$.
Contradiction. So there is some $u \in S$ such that $\post{u} \subseteq T$.\qedhere
\end{proof}

\begin{lemma}\label{lem:exits_restrict}
For $v \in U$, $v \forces{\player}{U} T$ implies $v \forces{\player}{U} T'$
whenever $T' \subseteq T$ and
$\post{w} \subseteq U \cup T'$ for all $w \in U$.
\end{lemma}
\begin{proof}
Choose $\sigma \in \strategy{\player}$
such that $v \forces{\strategyname}{U} T$, and let 
$\pathvar$ be a path for which $\ind{\pathvar}{0} = v \in U$ and
$\strategyname \allows \pathvar$. Then for some $j > 0$,
$\ind{\pathvar}{j} \in T$ and
$\ind{\pathvar}{k} \in U$ for all $k < j$. 
Since $\post{\ind{\pathvar}{j-1}} \subseteq
U \cup T'$, also $\ind{\pathvar}{j} \in T'$ and therefore
$v \forces{\strategyname}{U} T'$.  Thus
$v \forces{\player}{U} T'$.\qedhere
\end{proof}

\section{A Lattice of Parity Game Relations}
\label{sec:lattice}

In the rest of this paper, we study relations on parity games. 
We forego a formal treatment and present an overview of the studied relations and
how these are related in this section.
%
%

\subsection{Relations}
Let $\R$ be a relation over a set $V$, \ie ${\R} \subseteq {V \times V}$.
For $v, w \in V$ we write $v \R w$ to denote $(v, w) \in {\R}$. For a relation 
$\R$ and vertex $v \in V$ we define ${v \R} \isdef \{ w \in V \mid v \R w \}$, and 
likewise $\R v \isdef \{ w \in V \mid w \R v \}$. We also generalise this notation to sets of vertices
such that, for $U \subseteq V$, ${U \R} \isdef \bigcup_{u \in U} u \R$, and $\R U \isdef \bigcup_{u \in U} \R u$.

A relation $\R$ is a \emph{preorder} if it is reflexive and transitive. If, in addition, $\R$ is symmetric, then it is an \emph{equivalence} relation. Note that for an equivalence relation $\R$, and vertex $v \in V$, we have ${v \R} = {\R v}$. In this case we also write $\class{v}{\R} \isdef \{ v \in V \mid v \R w \}$, and call this the \emph{equivalence class} of $v$ under $\R$. By abuse of notation, for a subset $V' \subseteq V$, we write $\class{V'}{\R}$ for the set of equivalence classes with respect to $V$, \ie the set $\{\class{v}{\R} ~|~ v \in V' \}$. The set of equivalence classes of $V$ under $\R$ is denoted $\partition{V}{\R}$, and defined as $\{ \class{v}{\R} \mid v \in V \}$.

\subsection{Introducing a Lattice of Equivalences}
Preorders for parity games are particularly (and perhaps only)
interesting if they allow one to approximate the winning regions
of a game.  A preorder $\R$ approximates the winning region of a
game if, whenever $v \R w$, and player \even has a winning strategy
from $v$, then she also has a winning strategy from $w$.  For
equivalence relations, this requirement is stronger and often more
useful: we require that if $v \R w$, then \even has a winning
strategy from $v$ if and only if she has a winning strategy from
$w$.  The finest natural equivalence relation on $V$ is 
\emph{graph isomorphism}, denoted $\isomorphic$.
\begin{definition}[Isomorphism]
  Let $(V, \to, \priosym, \getplayername)$ be a parity game.
  Vertices $v,w \in V$ are \emph{isomorphic}, denoted $v \isomorphic w$
  iff $\phi(v) = w$ for some bijection $\phi \colon V \to V$ that satisfies, for all $\bar{v} \in V$:
  \begin{itemize}
    \item $\prio{\bar{v}} = \prio{\phi(\bar{v})}$,
    \item $\getplayer{\bar{v}} = \getplayer{\phi(\bar{v})}$, and
    \item $\bar{v} \to \bar{v}'$ if and only if $\phi(\bar{v}) \to \phi(\bar{v}')$.
  \end{itemize}
\end{definition}
The coarsest sensible equivalence on parity games is the
equivalence induced by the determinacy of parity games, \viz
the equivalence that exactly relates only and exactly those
vertices won by the same player.
\begin{definition}[Winner equivalence] 
  Let $(V, \to, \priosym, \getplayername)$ be a parity game. Vertices $v, w \in 
  V$ are \emph{winner equivalent}, denoted $v \winner w$ iff $v$ and 
  $w$ are won by the same player.
\end{definition}
Deciding winner equivalence of parity games is equivalent to paritioning the vertices
in the parity game into winning sets.\medskip

%
Winner equivalence and isomorphism are the extreme points in the lattice
of equivalence relations shown in Figure~\ref{fig:lattice}.
Between the extremal points in the lattice of Figure~\ref{fig:lattice} we list the other parity game equivalences that we study in more detail in the subsequent sections:
\begin{itemize}
  \item Strong bisimilarity ($\bisim$) \cite{Kei:13,KW:09}, see Section~\ref{sec:governed};
  \item Strong direct simulation equivalence ($\strongdirectsimeq$), see Section~\ref{sec:direct};
  \item Direct simulation equivalence ($\directsimeq$) \cite{FW:06,FW:02,GW:12}, see Section~\ref{sec:direct};
  \item Delayed simulation equivalence ($\delaysimeq$) \cite{FW:06}, see Section~\ref{sec:delayed};
  \item Delayed simulation equivalence, \even-biased ($\delaysimeeq$) \cite{FW:06}, see Section~\ref{sec:delayed_even};
  \item Delayed simulation equivalence, \odd-biased ($\delaysimoeq$) \cite{FW:06}, see Section~\ref{sec:delayed_odd};
  \item Governed bisimilarity ($\gov$) \cite{Kei:13,KW:09}, see Section~\ref{sec:governed};
  \item Stuttering bisimilarity ($\stut$) \cite{CKW:11}, see Section~\ref{sec:governed};
  \item Governed stuttering bisimilarity ($\gstut$) \cite{CKW:12}, see Section~\ref{sec:governed_stuttering};
\end{itemize}
In the lattice, an arrow from one equivalence to the other indicates that the first equivalence is finer than the latter. The number on an arrow refers to the theorem in this paper that claims this \emph{strictly finer-than} relation between the equivalences.
\medskip

\begin{figure}
\begin{center}
\parbox{.5\textwidth}
{
\begin{tikzpicture}[node distance=30pt]

  \node (isomorphism) {$\isomorphic$};
  \node [below of=isomorphism] (bisim) {$\bisim$};
  \node [below of=bisim] (fmib) {$\gov$};
  \node [below of=bisim, right of=bisim] (sdirsimeq) {$\strongdirectsimeq$};
  \node [below of=bisim, left of=bisim] (stut) {$\stut$};
  \node [below of=stut] (gstut) {$\gstut$};
  
  \node [below of=sdirsimeq] (dirsimeq) {$\directsimeq$};
  \node [below of=dirsimeq, left of=dirsimeq] (delaysimoeq) {$\delaysimoeq$};
  \node [below of=dirsimeq, right of=dirsimeq] (delaysimeeq) {$\delaysimeeq$};
  \node [below of=delaysimeeq, left of=delaysimeeq] (delaysimeq) {$\delaysimeq$};
  
  \node [below of=delaysimeq, left of=delaysimeq,xshift=-30pt] (winner) {$\winner$};
  
  \path[->] (isomorphism) edge node[left] {\scriptsize \ref{iso_lattice}} (bisim)
            (bisim) edge node[right] {\scriptsize \ref{strong_lattice}} (fmib)
            (bisim) edge node[above left] {\scriptsize \ref{strong_lattice}} (stut)
            (bisim) edge node[above right] {\scriptsize \ref{strong_lattice}} (sdirsimeq)
            (sdirsimeq) edge node[right] {\scriptsize \ref{strong_direct_lattice}} (dirsimeq)
            
            (fmib) edge node[below right] {\scriptsize \ref{governed_bisim_lattice}} (gstut)
            (fmib) edge node[above right] {\scriptsize \ref{governed_bisim_lattice}} (dirsimeq)
            (stut) edge node[left] {\scriptsize \ref{stut_lattice}} (gstut)
            (gstut) edge node[below left] {\scriptsize \ref{gstut_lattice}} (winner)
            (dirsimeq) edge node[above right] {\scriptsize \ref{direct_lattice}} (delaysimeeq)
            (dirsimeq) edge node[above left] {\scriptsize \ref{direct_lattice}} (delaysimoeq)
            (delaysimeeq) edge node[below right] {\scriptsize \ref{delay_lattice}} (delaysimeq)
            (delaysimoeq) edge node[above right] {\scriptsize \ref{delay_lattice}} (delaysimeq)
            (delaysimeq) edge node[below right] {\scriptsize \ref{delay_lattice}} (winner);

\end{tikzpicture}
}
\parbox{.45\textwidth}
{
\begin{enumerate}
  \item \label{iso_lattice} Theorem~\ref{thm:iso_refines_strong}, page~\pageref{thm:iso_refines_strong};
  \item \label{strong_lattice} Theorem~\ref{thm:strong_refines}, page~\pageref{thm:strong_refines};
  \item \label{strong_direct_lattice} Theorem~\ref{thm:strong_direct_refines}, page~\pageref{thm:strong_direct_refines};
  \item \label{governed_bisim_lattice} Theorem~\ref{thm:governed_bisim_refines}, page~\pageref{thm:governed_bisim_refines};
  \item \label{direct_lattice} Theorem~\ref{thm:direct_refines}, page~\pageref{thm:direct_refines}; 
  \item \label{delay_lattice} Theorem~\ref{thm:delay_refines}, page~\pageref{thm:delay_refines};
  \item \label{stut_lattice} Theorem~\ref{thm:stuttering_refines}, page~\pageref{thm:stuttering_refines};
  \item \label{gstut_lattice} Theorem~\ref{thm:gstut_refines}, page~\pageref{thm:gstut_refines}; 
\end{enumerate}
}
\caption{Lattice of equivalences for parity games. The numbers on the edges
refer to the legend shown to the right, which in turn refers to the theorems
that witness the existence of the edge.}
\label{fig:lattice}
\end{center}
\end{figure}

The original definitions of the equivalences listed above vary in nature. 
Strong-, governed-, stuttering-, and governed stuttering bisimilarity are defined coinductively, whereas direct simulation and all variations of delayed simulation are defined as simulation games. 
Furthermore, the direct- and delayed simulation games define a preorder, whereas the others define an equivalence relation; the preorders are lifted to equivalence relations in the standard way. 

\section{Direct and Delayed Simulation Equivalence}
\label{sec:preorders}

We introduce the direct simulation preorder and the induced direct simulation equivalence in Section~\ref{sec:direct}.
In Section~\ref{sec:delayed}, we recall the delayed simulation preorder, the induced delayed simulation equivalence and two \emph{biased} versions of the delayed simulation preorder and equivalence.
Throughout these sections, we assume that $\game = (V,\to,\priosym,\getplayername)$ is an arbitrary parity game.

\subsection{Direct Simulation and Direct Simulation Equivalence}
\label{sec:direct}

Direct simulation for parity games is one of the most basic preorders studied for parity games. 
It is difficult to trace the exact origins of the definition, but it was suggested (though not formally defined) in~\cite{FW:06} and appeared earlier in the setting of alternating B\"uchi automata~\cite{FW:02}. We here follow the game-based definition as given in~\cite{GW:12}.

\begin{definition}[Direct simulation game]
\label{def:direct-simulation-game}
The \emph{direct simulation game} is played on \emph{configurations} drawn from $V \times V$, and it is played in rounds.
A round of the game proceeds as follows:
\begin{enumerate}
  \item The players move from $(v, w)$ according to the rules in Table~\ref{tab:simulation_game_rules};
  \item Play continues in the next round from the newly reached position.
\end{enumerate}
An infinite play $(v_0, w_0), (v_1, w_1), \ldots$ is won by \duplicator if $\prio{v_j} = \prio{w_j}$ for all $j$,
\ie, \duplicator was able to mimic every move from \spoiler with a move to a vertex with equal priority. In all other cases \spoiler wins the play.
\end{definition}
We say that $v$ is \emph{directly simulated} by $w$, denoted $v \directsimg w$ whenever \duplicator has a winning strategy from $(v,w)$ in the direct simulation game.
\begin{table}[ht]
    \centering
    \begin{tabular}{cc||cc|cc}
        $\getplayer{v}$ & $\getplayer{w}$ & $1^\text{st}$ move & plays on & $2^\text{nd}$ move & 
        plays on \\
        \hline\hline
        $\even$ & $\even$ & $\spoiler$ & $v$ & $\duplicator$ & $w$\\
        $\even$ & $\odd$  & $\spoiler$ & $v$ & $\spoiler$ & $w$\\
        $\odd$  & $\even$ & $\duplicator$ & $w$ & $\duplicator$ & $v$\\
        $\odd$  & $\odd$  & $\spoiler$ & $w$ & $\duplicator$ & $v$
    \end{tabular} 
    \caption{Allowed moves in a (bi)simulation game.}
    \label{tab:simulation_game_rules}
\end{table}
\begin{example}
In the parity game in Figure~\ref{fig:directsim_example}, $v_0 \directsimg v_1$. Observe that from $(v_0, v_1)$, \duplicator can choose both successors in the direct simulation game. We do not have $v_1 \directsimg v_0$, since from configuration $(v_1, v_0)$, \spoiler's move $v_1 \to v_3$ cannot be matched from $v_0$. Note that additionally we have $v_0 \directsimg v_2$ and $v_2 \directsimg v_0$.
\begin{figure}
\begin{center}
\begin{tikzpicture}
  \node[odd,label=below:{$v_0$}] (v0) {1};
  \node[even, below of=v0,label=below:{$v_1$}] (v1) {1};
  \node[odd, right of=v0,label=below:{$v_2$}] (v2) {1};
  \node[even, right of=v1,label=below:{$v_3$}] (v3) {0};
  
  \path[->] (v0) edge (v2)
            (v1) edge (v2)
            (v1) edge (v3)
            (v2) edge[loop right] (v2)
            (v3) edge[loop right] (v3);
\end{tikzpicture}
\end{center}
\caption{Parity game with $v_0 \directsimg v_1$, $v_0 \directsimg v_2$, $v_2 \directsimg v_0$, and for all $v_i$, $v_i \directsimg v_i$.}
\label{fig:directsim_example}
\end{figure}
\end{example}
Direct simulation is a preorder: reflexivity is easily seen to hold (\duplicator can follow a copy-cat strategy), but transitivity is more involved. 
In the setting of alternating B\"uchi automata, direct simulation was shown to be transitive using strategy composition, see~\cite{FW:02,Fri:05}.
Following essentially the same technique one can show transitivity of direct simulation for parity games.
We use the direct simulation preorder to obtain direct simulation equivalence in the standard way.
\begin{definition}[Direct simulation equivalence {\cite{FW:06,GW:12}}]
    Vertices $v$ and $w$ are \emph{direct simulation equivalent}, denoted 
    $v \directsimeq w$, iff $v \directsimg w$ and $w \directsimg v$.
\end{definition}

The alternative coinductive definition of direct simulation which we present next, allows for a more straightforward proof of transitivity.
Our definition below was taken from~\cite{GW:16}, where it is also referred to as \emph{governed simulation}.
\begin{definition}[Direct simulation relation {\cite{GW:16}}]
\label{def:direct-simulation}
A relation ${\R} \subseteq V \times V$ is a \emph{direct simulation} if and only if $v \R w$ 
implies
\begin{itemize}
\item $\prio{v} = \prio{w}$;
\item if $v \in V_{\even}$, then for each $v' \in \post{v}$, $w \steps{\even} v' \R$;
\item if $v \in V_{\odd}$, then $w \steps{\even} \post{v} \R$.
\end{itemize}
We say that vertex $v$ is directly simulated by $w$, denoted $v \directsim w$, if and only if there is a direct simulation relation $\R$ such that $v \R w$.
\end{definition}

The theorem below states that the game-based and coinductive definitions of direct simulation coincide. 

\begin{theorem}
\label{thm:direct-sim-game-vs-coinductive}
For all $v, w \in V$, we have $v \directsim w$ if and only if $v \directsimg w$.
\end{theorem}
\begin{proof}
  We prove both implications separately.
  \begin{itemize}
    \item [$\Rightarrow$] We prove that from a pair of vertices $v \directsim w$ the game can always be played such that we again end up in related vertices after one round in the direct simulation game. Since $\prio{v} = \prio{w}$, it immediately follows that \duplicator has a winning strategy in the direct simulation game, and hence $v \directsimg w$.
    
    Let $v,w$ be such that $v \directsim w$. We distinguish two cases:
    \begin{itemize}
      \item $v \in V_{\even}$. Then \spoiler first chooses some $v \to v'$. Since $v \directsim w$, we know that $w \steps{\even} v' \directsim$, hence if $w \in V_{\even}$, \duplicator can choose a successor $w \to w'$ such that $v' \directsim w'$, and if $w \in V_{\odd}$, all successors $w \to w'$ that \spoiler may choose are such that $v' \directsim w'$.
      
      \item $v \in V_{\odd}$. Then $w \steps{\even} \post{v} \directsim$, so if $w \in V_{\even}$, there exist $w \to w'$ and $v \to v'$ such that $v' \directsim w'$, and \duplicator can play such that he picks those. In case $w \in V_{\odd}$, then for all $w \to w'$ there exists a $v \to v'$ such that $v' \directsim w'$. Since \spoiler plays first on $w$, \duplicator can match with the appropriate $v \to v'$.
    \end{itemize}
    
    \item [$\Leftarrow$] We prove that $\directsimg$ is a direct simulation relation. 
    
    Let $v,w$ be arbitrary such that $v \directsimg w$. Consider a winning strategy for \duplicator in the direct simulation game from $(v, w)$. Observe that $\prio{v} = \prio{w}$ follows trivially. We again distinguish two cases:
     \begin{itemize}
     \item $v \in V_{\even}$. Then \spoiler first chooses some $v \to v'$. If $w \in V_{\even}$, then \duplicator matches this with some $w \to w'$ according to his strategy. Since \duplicator's strategy is winning from $(v,w)$, it is also winning from $(v',w')$, hence $v' \directsimg w'$.
    If $w \in V_{\odd}$, any choice of successor $w \to w'$ that \spoiler makes is such that $v' \directsimg w'$. In both cases we have shown $w \steps{\even} v' \directsimg$.
    
    \item $v \in V_{\odd}$. If $w \in V_{\even}$, \duplicator plays both on $v$, say $v \to v'$, and on $w$, say $w \to w'$, such that $v' \directsimg w'$. If $w \in V_{\odd}$, \spoiler plays first on $w$, say $w \to w'$, and given this choice, \duplicator responds according to his strategy with $v \to v'$ such that $v' \directsimg w'$. Since \duplicator's strategy is winning, he is able to match any choice made by \spoiler. In both cases the strategy witnesses $w \steps{\even} \post{v} \directsimg$. 
\qedhere
   \end{itemize}
  
  \end{itemize}
\end{proof}

Given that both definitions of direct simulation coincide, we can now say that $v$ and $w$ are \emph{direct simulation equivalent}, \ie $v \directsimeq w$, if and only if $v \directsim w$ and $w \directsim v$. 
\begin{proposition}
\label{prop:preorder_directsim}
The relations $\directsim$ and $\directsimeq$ are a preorder and an equivalence relation, respectively. Moreover, $\directsim$ itself is a direct simulation relation.
\end{proposition}
\begin{proof} One can check that for direct simulation relations $R$ and $S$, the relation $R \circ S$, defined as $v \mathrel{(R \circ S)} w$ iff there is some $u$ such that $v \R u$ and $u \S w$, is again a direct simulation relation. \qedhere\end{proof}

\paragraph{Strong Direct Simulation.} If we impose an additional
constraint on direct simulation, \viz we do not allow to relate
vertices owned by different players, we obtain a notion that resembles
\emph{alternating refinement}~\cite{AHKV:98}. Clearly, this notion again
is a preorder. We write $v \strongdirectsim w$ iff there is some
\emph{strong direct simulation} relation that relates $v$ and $w$, and we
write $v \strongdirectsimeq w$ iff $v \strongdirectsim w$ and $w
\strongdirectsim v$. Note that in the parity game in Figure~\ref{fig:directsim_example}, 
we still have $v_0 \strongdirectsimeq v_2$, but $v_0 \not \strongdirectsim v_1$.

\subsection{Delayed simulation}
\label{sec:delayed}

Direct simulation equivalence is limited in its capability to relate
vertices.  The reason for this is that in each step of the simulation
game, \duplicator is required to match with a move to a vertex with
exactly the same priority.  Following Etessami \etal~\cite{EWS:05},
in~\cite{FW:06}, a more liberal notion of simulation called \emph{delayed simulation}
is considered. In this notion matching may be delayed.  The idea
is that in the winning condition of a play of a parity game, only
the priorities that occur infinitely often are of importance.
Therefore, intuitively it is allowed to delay matching a given
(dominating) priority for a finite number of rounds. 

The delayed simulation game is, like the direct simulation game,
played on an arena consisting of configurations that contain a pair
of vertices.  These configurations are now extended with a third
parameter which is used to keep track of the obligation that still
needs to be met by \duplicator.

An obligation in a delayed simulation game is either a natural number or the symbol $\checkmark$; the latter is used to indicate the absence of obligations.
We denote the set of obligations by $K$.  
Given two priorities and an existing obligation, a new obligation is obtained using the function $\update \colon \nat \times \nat \times K \to K$, where:
\begin{align*}
\update(n,m,\checkmark) & = \begin{cases}
  \checkmark & \text{ if } m \prioordereq n\\
  \min\{n,m\} & \text{otherwise}
\end{cases}\\
\update(n,m,k) & = \begin{cases}
  \checkmark & \text{if }m\prioordereq n \text{ and} \begin{cases}
    n \text{ odd and } n \leq k,\text {or}\\
    m \text{ even and } m \leq k
  \end{cases}\\
  \min\{ n,m,k \} & \text{otherwise}
\end{cases}
\end{align*}
By abuse of notation we will typically write 
$\update(v,w,k)$, for vertices $v,w \in V$ and obligation $k \in K$, to denote $\update(\prio{v}, \prio{w}, k)$.\medskip

The intuition behind this update is as follows. Either, with the new priorities, the pending obligation is fulfilled, and the new configuration does not give rise to a new obligation, in which case the result is $\checkmark$. Otherwise the new obligation is the minimum of the priorities passed to the function, or the current obligation. This signifies the most significant obligation that still needs to be met by \duplicator.
Note that there are two ways to fulfil the pending obligation. Either the first argument renders the pending obligation superfluous since it is smaller and odd, or the second argument is such that it matches the pending obligation.

%

\begin{definition}[Delayed simulation game {\cite{FW:06}}]
\label{def:delay_sim_game}
A \emph{delayed simulation game} is a game played by players \spoiler and \duplicator on an arena consisting of positions drawn from $V \times V$ and obligations taken from $K$.
The game is played in rounds. Assuming $(v, w)$ is the current position, and $k$ the current obligation, a round of the game proceeds as follows:
\begin{enumerate}
  \item \spoiler and \duplicator propose moves $v \to v'$ and $w \to w'$ according to the rules in Table~\ref{tab:simulation_game_rules}.
  \item The game continues from $(v', w')$ with obligation $\update(v', w', k)$.
\end{enumerate}
An infinite play $(v_0, w_0, k_0), (v_1, w_1, k_1), \ldots$ is won by \duplicator iff $k_j = \checkmark$ for infinitely many $j$. This means 
that \duplicator was always able to eventually fulfil all pending obligations. In all other cases \spoiler wins the game.
\end{definition}
We say that $v$ is \emph{delayed simulated} by $w$, denoted $v \delaysim w$ just whenever \duplicator has a winning strategy from $(v,w)$ with obligation $\update(v, w, \checkmark)$ in the delayed simulation game.
\begin{example}
In the parity game in Figure~\ref{fig:delaysim_example}, $v_i \delaysim v_j$ for all $i,j$. Observe that, with respect to direct simulation, vertices cannot be related to each other.
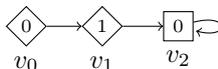
\begin{figure}
\begin{center}
\begin{tikzpicture}
  \node[even,label=below:{$v_0$}] (v0) {0};
  \node[even, right of=v0,label=below:{$v_1$}] (v1) {1};
  \node[odd, right of=v1,label=below:{$v_2$}] (v2) {0};
  
  \path[->] (v0) edge (v1)
            (v1) edge (v2)
            (v2) edge[loop right] (v2);
\end{tikzpicture}
\end{center}
\caption{Parity game in which all vertices are delay simulation equivalent.}
\label{fig:delaysim_example}
\end{figure}
\end{example}
Delayed simulation is, like direct simulation, a preorder.
The proof thereof is substantially more involved than the proof that direct simulation is a preorder, requiring an analysis of 24 cases, some of which are rather intricate.
For details, we refer to~\cite{Fri:05}; we here only repeat this result.
\begin{proposition}
The relation $\delaysim$ is a preorder.
\end{proposition}
We obtain delayed simulation equivalence in the standard way.
\begin{definition}[Delayed simulation equivalence {\cite{FW:06}}]
\label{def:delayed_simulation_preorder}
Vertices $v$ and $w$ are delayed simulation equivalent, denoted $v \delaysimeq w$, iff $v \delaysim w$ and $w \delaysim v$.
\end{definition}

Next, we give an alternative, coinductive definition for delayed simulation.
Since the moves in the game for delayed simulation and direct simulation match, one may expect that such a characterisation can be obtained by a more-or-less straightforward enhancement of the direct simulation relation. 
This is partly true: indeed, the moves of the game are captured in a way similar to how this is done for direct simulation. 
However, the winning condition of delayed simulation requires that infinitely often all obligations are met.
This requires `non-local' reasoning that must somehow be captured through a coinductive argument.
Meeting an obligation is typically a progress property, requiring an inductive argument rather than a coinductive argument.

To combine both aspects in a single, coinductive definition, we draw inspiration from Namjoshi's notion of well-founded bisimulation~\cite{Nam:97}.
Well-founded bisimulation is a relation which is equivalent to stuttering equivalence, but which permits local reasoning by introducing a well-foundedness criterion.
We use a similar well-foundedness requirement in our coinductive definition, ensuring progress is made towards fulfilling obligations.
This moreover requires, as can be expected, that our coinductive relation ranges not only over pairs of vertices but also over obligations. 
For a relation ${\R} \subseteq V \times K \times V$, we write $v \R[k] w$ if $v$ and $w$ are related under pending obligation $k$.
The well-foundedness restriction thus enables us to express that $v\R[k] w$ holds if we can build a coinductive argument that ultimately depends on pairs of vertices $v',w'$ related under obligation $\checkmark$; \viz $v'\R[\checkmark] w'$.

\begin{definition}[Well-founded delayed simulation]
  \label{def:wf_delay_dim}
A relation ${\R} \subseteq V \times \obligations \times V$ is a well-founded delayed simulation iff
  there is a well-founded order $\precdot$ on $V \times V \times \obligations$ such
that
for all $v,w \in V$ and $\obligation \in \obligations$ for which $v \R[\obligation] w$ holds, also:
  \begin{itemize}
    \item $v \in V_{\even}$ implies for all $v' \in \post{v}$,
    $w \steps{\even} \{ w' \in V \mid \ell =  
    \update(v',w',\obligation) \land v' \R[\ell] w' \land 
    (\obligation = \checkmark \lor (v', w', \ell) \precdot 
    (v,w,\obligation) ) \} $;
    
    \item $v \in V_{\odd}$ implies $w \steps{\even} \{ w' \in V \mid \exists{v' \in \post{v}}{\ell = \update(v',w',\obligation) \land v' \R[\ell] w' 
    \land (\obligation = \checkmark \lor (v',w',\ell) \precdot 
    (v,w,\obligation) ) } \} $.
  \end{itemize}
  Vertex $v$ is \emph{well-founded delayed simulated} by $w$, denoted $v 
  \delaysimc w$, iff there exists a well-founded delayed simulation $\R$ 
  such that $v \R[\update(v,w,\checkmark)] w$.
\end{definition}

In the remainder of this section we show that this definition is equivalent to the game-based definition.

\begin{lemma}\label{lem:wf_implies_delay}
  For $v,w \in V$, $v \delaysimc w$ implies $v \delaysim w$.
\end{lemma}
\begin{proof}
  We prove the stronger statement that if there is a well-founded delayed simulation $\R$ such that $v \R[\update(v, w, \obligation)] w$. for $k \in \obligations$, then \duplicator has a strategy to win the delayed simulation game from $(v, w, \update(v, w, \obligation))$. The result then follows immediately from the observation that $v \R[\update(v, w, \checkmark)] w$, and hence \duplicator wins the game from $(v, w, \update(v, w, \checkmark))$.
  
  We first show that \duplicator has a strategy to move between positions $(v,w)$ with obligation $\obligation$ for which $v \R[\obligation] w$ to positions $(v',w')$ and obligation $\obligation'$ for which $v' \R[\obligation'] w'$.
  Assume that $v \R[k] w$. 
  We distinguish four cases based on the owner of $v$ and $w$.
  \begin{itemize}
    \item $(v,w) \in V_{\even} \times V_{\even}$. In the delayed simulation 
    game, this corresponds to the vertex $(v, w, \obligation)$, in which \spoiler is to move first. \spoiler first plays an arbitrary move $v \to v'$. 
    By definition of the well-founded delayed simulation, there is $w \to w'$
    such that $v' \R[\update(v', w', \obligation)] w'$; \duplicator matches with this $w'$.
    
    \item $(v,w) \in V_{\even} \times V_{\odd}$. In the delayed simulation game,
    from position $(v, w, \obligation)$, \spoiler is to make both moves, so there is no
    \duplicator strategy to be defined. Observe that well-founded delayed simulation guarantees that for all $v \to v'$ and  $w \to w'$, $v' \R[\update(v', w', \obligation)] w'$.
    
    \item $(v,w) \in V_{\odd} \times V_{\even}$. \duplicator plays twice in the delayed simulation game.
    According to the well-founded delayed simulation, there exist $v \to v'$ and $w \to w'$ such that $v' \R[\update(v', w', \obligation)] w'$. \duplicator plays such moves.
    
    \item $(v,w) \in V_{\odd} \times V_{\odd}$. In the delayed simulation game, \spoiler is to move
    first, say $w \to w'$. From the well-founded 
    delayed simulation, we find that for all such moves, there exists some $v 
    \to v'$ such that $v' \R[\update(v', w', \obligation)] w'$. \duplicator plays to
    this $w'$.
  \end{itemize}
  
  It remains to be shown that for all configurations $(v, w, \obligation)$ such that $v \R[\obligation] w$, this \duplicator-strategy is, indeed, winning for \duplicator. Observe that it suffices to show that, if $\obligation \neq \checkmark$, eventually a configuration $(v', w', \checkmark)$ is reached. This follows, since in every round in the game above moves are made from $(v, w, \obligation)$ to $(v', w', \obligation')$ such that $(v', w', \obligation') \precdot (v, w, \obligation)$. Since $\precdot$ is a well-founded order, this can only be repeated finitely many times, and eventually all obligations are met.\qedhere
\end{proof}

Before we show the converse, we first show that a winning strategy for player \duplicator in the delayed simulation game induces a well-founded order on those configurations won by \duplicator.
\begin{lemma}\label{lem:wellfounded_exists}
  The winning strategy for \duplicator in the delayed simulation game induces a well-founded order on $V \times V \times K$ for those $(v,w,\obligation)$ for which \duplicator wins position $(v,w)$ with obligation $\obligation$.
\end{lemma}
\begin{proof}
  Observe that the delayed simulation game has a B\"uchi winning condition. 
  Hence those configurations in the game play that \duplicator can win, can be won using a memoryless strategy. For the remainder of the proof, fix such a winning memoryless strategy.
  
  For each position $(v,w)$ and obligation $\obligation$ that is won by \duplicator, we extract a finite tree from the solitaire game that is induced by \duplicator's strategy by taking the (infinite) unfolding of the game starting in $(v, w, \obligation)$, and pruning each branch at the first node with obligation $\checkmark$. 
Since the strategy is \duplicator-winning, this tree is finite. 
Furthermore, if $(v, w, \obligation)$ appears in the tree of a different configuration, the subtree rooted in $(v, w, \obligation)$ in that particular subtree is identical to the tree of $(v, w, \obligation)$.
  
  These trees determine a well-founded order $\precdot$ when we set $(v', w', \ell) \precdot (v, w, \obligation)$  iff the height of the tree rooted in $(v', w', \ell)$ is less than that of the tree rooted in $(v, w, \obligation)$. \qedhere
\end{proof}

The following corollary immediately follows from the existence of the well-founded order.
\begin{corollary} \label{cor:wf_decrease}
In the delayed simulation game, for every position $(v,w)$ with obligation $\obligation$ from which \duplicator has a winning strategy, if the game proceeds according to this strategy to some position $(v', w')$ with obligation $\ell = \update(v',w', \obligation)$, we have $\obligation = \checkmark$ or $(v', w', \ell) \precdot (v, w, \obligation)$.
\end{corollary}

Finally, we show that any delayed simulation is also a well-founded delayed
simulation.
\begin{lemma}\label{lem:delay_implies_wf}
  For $v, w \in V$, $v \delaysim w$ implies $v \delaysimc w$.
\end{lemma}
\begin{proof}
  Let ${\R} \subseteq V \times \obligations \times V$ be such that $v \R[\obligation] w$ if
  \duplicator wins the delay simulation game from $(v, w)$ with obligation $\obligation$.
  Observe that, in particular, since $v \delaysim w$, we have $v \R[\update(v,w,\checkmark)] w$.
  We show that $\R$ is a well-founded delayed simulation.
  
  For configurations $(v, w)$ with obligation $\obligation$, \duplicator has a winning strategy. Since the delay simulation game has a B\"uchi winning condition, \duplicator also has a memoryless winning strategy. For the remainder of the proof, fix this strategy. Observe that a well-founded order on configurations won by \duplicator exists, according to Lemma~\ref{lem:wellfounded_exists}.

  Next we show that this relation satisfies the transfer conditions. 
  Let $v, w, \obligation$ be arbitrary, such that $v \R[\obligation] w$.
  We distinguish four cases.
  \begin{itemize}
    \item $(v, w) \in V_{\even} \times V_{\even}$. According to Definition~\ref{def:delay_sim_game}, for all moves $v \to v'$ made by \spoiler, \duplicator matches with a move $w \to w'$, and the game continues from configuration $(v', w')$ with obligation $\ell = \update(v',w',\obligation)$. Since \duplicator's strategy is winning from $(v,w)$ with obligation $\obligation$, the strategy is also winning from $(v', w')$ with obligation $\ell$ hence $v' \R[\ell] w'$. Furthermore, if $\obligation \neq \checkmark$ we have $\rank{v'}{w'}{\ell} \precdot \rank{v}{w}{\obligation}$ according to Corollary~\ref{cor:wf_decrease}.

    \item $(v,w) \in V_{\even} \times V_{\odd}$. For all moves $v \to v'$ and $w \to w'$ made by \spoiler, the game continues from configuration $(v', w')$ with obligation $\ell = \update(v',w',\obligation)$ from which again \duplicator's strategy is winning, hence $v' \R[\ell] w'$, and again we have $\rank{v'}{w'}{\ell} \precdot \rank{v}{w}{\obligation}$ if $\obligation \neq \checkmark$ according to Corollary~\ref{cor:wf_decrease}.

    \item $(v,w) \in V_{\odd} \times V_{\even}$. \duplicator plays a move $w \to w'$ and $v \to v'$ and continues from $(v', w')$ with obligation $\ell = \update(v', w', \obligation)$, from which her strategy is again winning, hence $v' \R[\ell] w'$. Using the same argument as before, also $\rank{v'}{w'}{\ell} \precdot \rank{v}{w}{\obligation}$ if $\obligation \neq \checkmark$ according to Corollary~\ref{cor:wf_decrease}.

    \item $(v,w) \in V_{\odd} \times V_{\odd}$. For all moves $w \to w'$ made by \spoiler, \duplicator's strategy matches with some $v \to v'$ such that \duplicator wins from $(v',w')$ with obligation $\ell = \update(v',w',\obligation)$. Again $\rank{v'}{w'}{\ell} \precdot \rank{v}{w}{\obligation}$ if $\obligation \neq \checkmark$, hence $w \to_{\even} \{ w' \in V \mid \exists{v \to v'}{\ell = \update(v',w',\obligation) \land v' \R[\ell] w' 
    \land (\obligation = \checkmark \lor \rank{v'}{w'}{\ell} \precdot \rank{v}{w}{\obligation})} \}$.
    
  \end{itemize}
  In each of the cases above, the requirements for well-founded delayed 
  simulation are satisfied, hence $\R$ is a well-founded delayed simulation 
  relation, and $v \delaysimc w$. \qedhere
\end{proof}

The following theorem, stating that delayed simulation and well-founded delay simulation coincide, now follows directly.
\begin{theorem}\label{thm:delay_sim}
For all $v,w \in V$ we have $v \delaysim w$ if and only if $v \delaysimc w$.
\end{theorem}
\begin{proof}
Follows immediately from lemmata \ref{lem:wf_implies_delay} and \ref{lem:delay_implies_wf}.
\end{proof}

\subsubsection{Biased delayed simulations.}
\label{sec:biased_delayed}

As observed in~\cite{FW:06}, quotienting is problematic for delayed simulation: no sensible definition of quotienting appears to exist such that it guarantees that the quotient is again delayed simulation equivalent to the original game.
Fritz and Wilke `mend' this by introducing two variations (so called \emph{biased} delayed simulations) on delayed simulation which do permit some form of quotienting
although these are not unique.
We briefly describe these variations below. 

\paragraph{Even-biased delayed simulation.}
\label{sec:delayed_even}

The even-biased delayed simulation game, and its coinductive variant, are identical to their delayed simulation and well-founded delayed simulation counterparts. The only difference lies in the update function on obligations. Given two priorities and an existing obligation, a new obligation is obtained using the update function $\update[e] \colon \nat \times \nat \times \obligations \to \obligations$, where:
\begin{align*}
\update[e](n,m,\obligation) & = \begin{cases}
\obligation & \text{if } m \prioordereq n, n \text{ odd}, n \leq \obligation, \\
            & \text{ and } (m \text{ odd or } \obligation < m) \\
\update(n,m,\obligation) & \text{otherwise}
\end{cases}
\end{align*}
We again abbreviate $\update[e](\prio{v},\prio{w},k)$ by $\update[e](v,w,k)$.

Using the new update function in the delayed simulation game ensures that a pending obligation is only changed back to $\checkmark$ by a small even priority; a small odd priority does not change the obligation.
We say that $v$ is \emph{even-biased delayed simulated} by $w$, denoted $v \delaysime w$ iff \duplicator has a winning strategy from $(v, w)$ with obligation $\update[e](v, w, \checkmark)$ in the even-biased delayed simulation game.

Likewise, we obtain \emph{well-founded, even-biased delayed simulation} by replacing all occurrences of $\update$ by $\update[e]$ in Definition~\ref{def:wf_delay_dim}. Vertex $v$ is \emph{well-founded, even-biased delayed simulated} by $w$, denoted $v \delaysimce w$, iff there exists a well-founded, even-biased delayed simulation preorder $\R$ such that $v \R[{\update[e](v, w, \checkmark)}] w$.

\paragraph{Odd-biased delayed simulation.}
\label{sec:delayed_odd}
Odd-biased delayed simulation is defined in a similar way as the even-biased delayed simulation.
Instead of small even priorities leading to an update of a pending obligation, small odd priorities lead to a change in the obligation. 
Given two priorities and an existing obligation, a new obligation is obtained using the update function $\update[o] \colon \priosym \times \priosym \times \obligations \to \obligations$, where:
\begin{align*}
\update[o](n,m,\obligation) & = \begin{cases}
\obligation & \text{if } m \prioordereq n, m \text{ even}, m \leq \obligation, \\
            & \text{ and } (n \text{ even or } \obligation < n) \\
\update(n,m,\obligation) & \text{otherwise}
\end{cases}
\end{align*}
The game-based and coinductive definitions are analogous to the even-biased version.

\section{Governed Bisimulation and Governed Stuttering Bisimulation}
\label{sec:bisimulations}

In this section we consider essentially two notions of bisimulation for parity games, and some derived notions.
First, in Section~\ref{sec:governed}, we introduce governed bisimulation, which was studied under various guises in \eg~\cite{GW:12,Kei:13,KW:09}.
Governed bisimulation is, as we demonstrate in that section, closely related to direct simulation.
Next, in Section~\ref{sec:governed_stuttering}, governed stuttering bisimulation~\cite{Cra:15,CKW:12,Kei:13} is introduced.

\subsection{Governed bisimulation}
\label{sec:governed}

Our definition of governed bisimulation, as presented below, is based on the one from~\cite{KW:09} where it is defined in the closely related setting of \emph{Boolean equation systems}; because of its capabilities to relate \emph{conjunctive} and \emph{disjunctive} equations, it was dubbed \emph{idempotence identifying bisimulation}.
It was rephrased for parity games in~\cite{Kei:13} and there named \emph{governed bisimulation}.

\begin{definition}[Governed bisimulation]
\label{def:govbisim}
A symmetric relation ${\R} \subseteq V\times V$ is a \emph{governed bisimulation} iff $v \R w$ 
implies
\begin{itemize}
\item $\prio{v} = \prio{w}$;
\item if $\getplayer{v} \ne \getplayer{w}$, then $v' \R w'$ for all $v' \in \post{v}$ and $w' \in \post{w}$;
\item for all $v' \in \post{v}$ there is some $w' \in \post{w}$ such that $v' \R w'$.
\end{itemize}
Vertices $v$ and $w$ are said to be \emph{governed bisimilar}, denoted $v 
\gov w$, if and only if there is a governed bisimulation $\R$ such that $v \R 
w$.
\end{definition}
\begin{example}
In the parity game in Figure~\ref{fig:govbisim_example}, we have for all $i$, $v_i \gov v_i$, and furthermore, $v_2 \gov v_3$ and $v_0 \gov v_6$. Observe that we have $v_0 \directsimeq v_1$, where $v_0 \directsim v_1$ is witnessed by relation $R_1 = \{ (v_i, v_i) \mid 0 \leq i \leq 5 \} \cup \{ (v_0, v_1) \}$, and $v_1 \directsim v_0$ is witnessed by $R_2 = \{ (v_i, v_i) \mid 0 \leq i \leq 5 \} \cup (v_1, v_0), (v_4, v_2), (v_4, v_3) \}$. We do, however, not have $v_0 \gov v_1$, since the latter would require $v_4$ to be related to $v_2$, but from $v_4$ the step $v_2 \to v_5$ cannot be mimicked.
\begin{figure}
\begin{center}
\begin{tikzpicture}
  \node[even, label=below left: $v_0$] (v0) {0};
  \node[odd, left of=v0, label=below left: $v_6$] (v6) {0};
  \node[even, right of=v0, label=right: $v_1$] (v1) {0};
  \node[even, below of=v0, label=below left: $v_2$] (v2) {1};
  \node[even, below of=v1, label=below right: $v_3$] (v3) {1};
  \node[even, below of=v1, right of=v1, label=below right: $v_4$] (v4) {1};
  \node[odd, below of=v2, label=below left: $v_5$] (v5) {2};
  
  \path[->] (v0) edge (v2)
            (v1) edge (v2)
            (v6) edge (v2)
            (v1) edge (v4)
            (v2) edge (v5)
            (v2) edge[bend left] (v3)
            (v3) edge[bend left] (v2)
            (v3) edge (v5)
            (v4) edge[loop right] (v4)
            (v5) edge[loop left] (v5);
\end{tikzpicture}
\end{center}
\caption{Parity game in which $v_2$ and $v_3$ are governed bisimilar. Vertices $v_0$, $v_1$ and $v_6$ are direct simulation equivalent. Vertices $v_0$ and $v_6$ are governed bisimilar but not strong direct simulation equivalent. Vertices $v_0$ and $v_1$ are strong direct simulation equivalent, but not governed bisimilar.}
\label{fig:govbisim_example}
\end{figure}
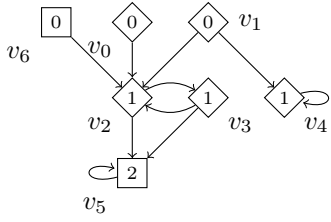
\end{example}
Governed bisimulation is such that vertices owned by different players can only be related whenever all their successors are. 
It turns out that this is exactly what is obtained when imposing a symmetry requirement on direct simulation.
As a result, we have the following theorem.
\begin{theorem}
\label{thm:alternative_govbisim} We have $v \gov w$ iff there is a symmetric direct simulation relation $\R$ such that
$v \R w$.
\end{theorem}
\begin{proof}  
We prove both implications separately.
\begin{itemize}
   \item[$\Rightarrow$] Let $\R$ be a governed bisimulation, and let $v, w$ be arbitrary such that $v \R w$. 
        Since $\R$ is symmetric, it suffices to show that $\R$ is a direct simulation.
   \begin{itemize}
     \item $\prio{v} = \prio{w}$ follows immediately.
     \item Suppose $v \in V_{\even}$. Let $v' \in \post{v}$. 
           In case $w \in V_{\odd}$ we have $v' \R w'$ for all $w' \in \post{w}$
           If $w \in V_{\even}$, there is some $w' \in \post{w}$ such that $v' \R w'$.
           Both cases lead to the desired $w \steps{\even} v'\R$. 

     \item Suppose $v \in V_{\odd}$. We again distinguish two cases.
     In case $w \in V_{\even}$ we have for all $v' \in \post{v}$ and $w' \in \post{w}$, $v' \R w'$, \ie $w \steps{\even} \post{v}\R$. 
     Suppose $w \in V_{\odd}$. Pick an arbitrary $w'\in \post{w}$. Since $\R$ is symmetric, also $w \R v$.
     Hence,  there exists a $v' \in \post{v}$ such that $w' \R v'$, which implies $v' \R w'$. 
     Thus, both cases lead to $w \steps{\even} \post{v}\R$.
   \end{itemize}

   \item[$\Leftarrow$] Let $\R$ be a symmetric direct simulation relation. Pick arbitrary $v, w$ for which $v \R w$.
   \begin{itemize}
     \item $\prio{v} = \prio{w}$ follows immediately.
     \item Suppose $\getplayer{v} \neq \getplayer{w}$. 
     Because $\R$ is symmetric, we may assume, without loss of generality, that $v \in V_{\even}$ and $w \in V_{\odd}$. 
     Pick an arbitrary $v' \in \post{v}$.
     Since $\R$ is a direct simulation, we find $w \steps{\even} v'\R$. 
     As $w \in V_{\odd}$, we thus find that $v' \R w'$ for all $w' \in \post{w}$.

     \item Let $v' \in \post{v}$. By the previous case, it suffices to consider only the case that $\getplayer{v} = \getplayer{w}$. 
     Suppose $v,w \in V_{\even}$. Then $w \steps{\even} v'\R$; \ie there is some $w' \in \post{w}$ such that $v' \R w'$.
     Now assume $v,w \in V_{\odd}$. By symmetry, we have $w \R v$. Then we have $v \steps{\even} \post{w}\R$. 
     Thus, for every $v'' \in \post{w}$ there is some $w' \in \post{w}$ such that $w' \R v''$.
     In particular, we have $w' \R v'$ for some $w' \in \post{w}$. By symmetry, we then also have $v' \R w'$ for some
     $w' \in \post{w}$.\qedhere

   \end{itemize}
 \end{itemize}
\end{proof}
As a consequence, we immediately find that governed bisimilarity is an equivalence relation.
\begin{theorem}\label{thm:pg_fmib_equivalence}
$\gov$ is an equivalence relation on parity games.
\end{theorem}
\begin{proof}
Follows from combining Theorem~\ref{thm:alternative_govbisim} and Proposition~\ref{prop:preorder_directsim}.\qedhere\end{proof}
Additionally, we immediately obtain a game-based definition for governed bisimulation:
we only need to require that \spoiler can switch to a symmetric position in the game play.
\begin{definition}[Governed bisimulation game]
\label{def:governed-bisim-game}
The \emph{governed bisimulation game} is played on \emph{configurations} drawn from $V \times V$, and it is played in rounds.
A round of the game proceeds as follows:
\begin{enumerate}
  \item \spoiler chooses $(u_0, u_1) \in \{ (v,w), (w,v) \}$;
  \item The players move from $(u_0, u_1)$ according to the rules in Table~\ref{tab:simulation_game_rules}
  \item Play continues in the next round from the newly reached position.
\end{enumerate}
An infinite play $(v_0, w_0), (v_1, w_1), \ldots$ is won by \duplicator if $\prio{v_j} = \prio{w_j}$ for all $j$,
\ie, \duplicator was able to mimic every move from \spoiler with a move to a vertex with equal priority. In all other cases \spoiler wins the play.
\end{definition}
We write $v \govg w$ whenever \duplicator has a winning strategy from $(v,w)$ in the governed bisimulation game.

\begin{theorem}
\label{thm:governed-bisim-game-vs-coinductive}
For all $v, w \in V$, we have $v \gov w$ if and only if $v \govg w$.
\end{theorem}
\begin{proof}
Along the lines of the proof of Theorem~\ref{thm:direct-sim-game-vs-coinductive}.\qedhere\end{proof}

\paragraph{Strong Bisimulation.} If we again impose the additional
constraint on governed bisimulation that we do not allow to relate
vertices owned by different players, we obtain a notion called
\emph{strong bisimulation}~\cite{Kei:13}. The derived notion of
\emph{strong bisimilarity}, denoted $v \bisim w$ and defined as $v
\bisim w$ iff there is some strong bisimulation relation that relates
$v$ and $w$, is an equivalence relation.

\subsection{Governed stuttering bisimulation}
\label{sec:governed_stuttering}

The (bi)simulation games discussed so far all have in common that the game-play proceeds in `lock-step': \duplicator must match every move proposed by \spoiler with a proper countermove.
In a sense, this ignores the fact that the parity condition is not sensitive to finite repetitions of priorities but only cares about infinite repetitions.
The insensitivity of the parity condition to finite repetitions is reminiscent to the notion of \emph{stuttering} in process theory.
Indeed, as we demonstrate in what follows, governed bisimulation can be weakened such that it becomes insensitive to finite stuttering, but remains sensitive to infinite stuttering.
The resulting relation is called \emph{governed stuttering bisimulation}.
\medskip

Essentially, governed stuttering bisimulation is obtained by porting
stuttering equivalence for Kripke structures to the setting of
parity games.  Intuitively, governed stuttering bisimulation requires
that a move from a vertex $v$ to $v'$ is matched by a finite (and
potentially empty) sequence of moves from $w$, through vertices that remain
related to $v$, to arrive at some $w'$ that is
related to $v'$. In addition,
every \emph{divergent play} from a vertex $v$ (\ie a play that remains
confined to a single equivalence class) should be matched with a
divergent play from a related vertex $w$.

Details, however, are subtle. In stuttering equivalence it suffices
to have the \emph{ability} to move or diverge and match such moves
or divergences with some move or a divergence.  In contrast, in the
parity game setting we are concerned with player's capabilities.
Only moves and divergences that can be \emph{forced} by a player
count, and matching of such moves and divergences must be done
through moves or divergences that the same player can force. 
Figure~\ref{fig:fake-divergence} illustrates some of these concepts. While in the
depicted parity game there is an infinite play that
passes through the two left-most vertices with priority~$0$, neither \even nor \odd
can force such an infinite play. As a result, we may ignore such
infinite plays, and in this sense, the abilities (for both players) from those two
vertices are no different from the abilities both players have from
the two right-most vertices with priority~$0$.

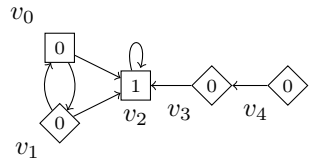
\begin{figure}[h!]
\begin{center}
\begin{tikzpicture}
  \node[odd, label=above left:$v_0$] (n1) {0};
  \node[even, below of=n1, label=below left:$v_1$] (n2) {0};
  \node[odd, right of=n1, yshift=-0.5cm, label=below:$v_2$] (n3) {1};
  \node[even, right of=n3, label=below left:$v_3$] (n4) {0};
  \node[even, right of=n4, label=below left:$v_4$] (n5) {0};
  
  
  \path[->] (n1) edge (n3)  
                 edge[bend left] (n2)
            (n2) edge (n3)
                 edge[bend left] (n1)
            (n3) edge[loop above] (n3)
            (n4) edge (n3)
            (n5) edge (n4);
\end{tikzpicture}
\end{center}
\sbox1{$\gstut$}
\caption{Equal priorities are related by \usebox1. Neither player can force 
play to visit only vertices with priority $0$.}
\label{fig:fake-divergence}
\end{figure}

The definition of governed stuttering bisimulation presented below is based on~\cite{CKW:12,Cra:15}.
For our definition, we strongly rely on our notation to denote that a player is able to `force play'.
\begin{definition}[Governed stuttering bisimulation]
\label{def:fmist}
Let ${\R} \subseteq 
V \times V$ be an equivalence relation. Then $\R$ is a \emph{governed 
stuttering bisimulation} if and only if $v \R w$ implies
\begin{enumerate}
\item[a)] $\prio{v} = \prio{w}$;
\item[b)] $v \to \C$ implies $w \forces{\getplayer{v}}{\R} \C$, for all $\C \in 
\partition{V}{\R} \setminus \{\class{v}{\R}\}$. \label{def:fmist-transfer-cond}
\item[c)] $v \diverges{\player}{\R}$ implies $w\diverges{\player}{\R}$ for 
$\player \in \{\even,\odd\}$. \label{def:fmist-div-cond}
\end{enumerate}
Vertices $v$ and $w$ are \emph{governed stuttering bisimilar}, denoted $v 
\gstut w$, iff a governed stuttering bisimulation $\R$ exists such that $v \R 
w$.
\end{definition}
\begin{example}
The parity game in Figure~\ref{fig:fake-divergence} nicely illustrates the key properties of governed stuttering bisimulation: for $v_j$ in $\{ v_0, v_1, v_3, v_4 \}$ we have neither $v_j \diverges{\even}{\gstut}$ nor $v_j \diverges{\odd}{\gstut}$. Furthermore, for all these vertices, both players can force the game to reach vertex $v_2$. Therefore, all vertices with the same priorities are related by $\gstut$. Also note that the vertices with priority $0$ are not related by, \eg, governed bisimulation since the latter is sensitive to counting, and $v_0$ and $v_1$ can reach multiple equivalence classes.
\end{example}
Proving that governed stuttering bisimilarity is an equivalence
relation that is again a governed stuttering bisimulation relation 
is technically involved. In particular, all standard
proof techniques for doing so break down or become too complex to
manage. Instead of a large monolithic proof of the result, we proceed
in small steps by gradually rephrasing the above definition to one
that is ultimately more easily seen to be an equivalence. Our first step
in this direction is to remove the asymmetry in clause \emph{b)} of 
the definition of governed stuttering bisimulation. Before we do so,
we state a useful lemma that allows us to strengthen the conclusion
of Lemma~\ref{lem:exits}.
\begin{lemma}\label{lem:exits_gen}
Let $R$ be a governed stuttering bisimulation. Let $\U \subseteq \partition{V}{R}
\setminus \{\class{v}{R}\}$.
If $\post{v} \subseteq \bigcup \U$,
then  $\post{u} \setminus \class{v}{R} \subseteq \bigcup\U$
for all $u \in \class{v}{R}$.
\end{lemma}
\begin{proof} Let $v$ be such that $v \step \U$ for some $\U \subseteq
\partition{V}{R} \setminus \{\class{v}{R}\}$.
Suppose $u \step \C$ for some $\C \notin \U \cup \{\class{v}{R}\}$.
Since $v \gstut u$, by Definition~\ref{def:fmist}, we have
$v \forces{\getplayer{u}}{\R} \C$. But $\post{v} \subseteq \bigcup \U$ and
$\C \notin \U$ so $v \nforces{\getplayer{u}}{\R} \C$. Contradiction.\qedhere
\end{proof}

\begin{theorem}\label{th:fmist2}
  Let ${\R} \subseteq V \times V$ and $v,w\in V$. Then $\R$ is a \gstutname 
  iff $\R$ is an equivalence relation and $v \R w$ implies:
  \begin{enumerate}
    \item[a)] $\Omega(v) = \Omega(w)$;
    \item[b)] $v \forces{\player}{\R} \C$ iff $w \forces{\player}{\R} \C$ for 
    all $\player
    \in 
    \{\even,\odd\}, \C\in \partition{V}{\R} \setminus \{\class{v}{\R}\}$;
    \item[c)] $v \diverges{\player}{\R}$ iff $w\diverges{\player}{\R}$ for all 
    $\player \in 
    \{\even,\odd\}$.
  \end{enumerate}
\end{theorem}
\begin{proof}
The proof for the implication from right to left follows immediately. We focus on
the implication from left to right. Assume that $\R$ is a governed stuttering
bisimulation. We prove the second condition only;
the other two conditions follow immediately from Definition~\ref{def:fmist} and
symmetry of $R$. Let
$\player$ be an arbitrary player and assume that $v \forces{\player}{\R} \C$ for
given $v \in V$ and $\C \in \partition{V}{\R} \setminus \{\class{v}{\R}\}$.
Let $S = \{u \in \class{v}{R} ~|~ \post{u} \cap \C \neq \emptyset\}$.
We distinguish two cases.
\begin{itemize}
\item Case $V_{\player} \cap S \neq \emptyset$. Let $u
\in S \cap V_{\player}$. Since $u \step \C$,
$w \forces{\player}{R} \C$ follows from Definition~\ref{def:fmist}. 

\item Case $S \subseteq V_{\opponent{\player}}$. By Lemma~\ref{lem:exits},
there is a $u \in S$ for which $\post{u} \subseteq \C$ and 
by Lemma~\ref{lem:exits_gen} (for $\U = \{\C\}$), 
$\post{u} \setminus \class{v}{R} \subseteq \C$ 
for \emph{all} $u \in \class{v}{R}$.
Furthermore, by
Lemma~\ref{lem:force-vs-div}, $v \forces{\player}{R} \C$ implies
$v \nforces{\opponent{\player}}{R}$. Then, by Definition~\ref{def:fmist},
$w \nforces{\opponent{\player}}{R}$ and by Lemma~\ref{lem:force-vs-div},
$w \forces{\player}{R} V \setminus \class{v}{R}$. But since
$\post{u} \subseteq \C \cup \class{v}{R}$ for all $u$, the desired
$w \forces{\player}{R} \C$ follows from
Lemma~\ref{lem:exits_restrict}.
\qedhere\end{itemize}
\end{proof}
While the above alternative characterisation of \gstutname is now
fully symmetric, the restriction on the class $\C$ that is considered
in clause \emph{b)} turns out to be too strong to facilitate an insightful proof
that $\gstut$ is an equivalence relation. We therefore further generalise this clause
such that it is phrased in terms of \emph{sets} of classes. \medskip

A perhaps surprising side-result of this
generalisation is that the
divergence requirement of clause \emph{c)} becomes superfluous.
Note that this generalisation is
not trivial, as $v \forces{\player}{\R} \{\C_1, \C_2\}$ is in general
neither equivalent to saying that $v \forces{\player}{\R} \C_1$ and $v 
\forces{\player}{\R}
\C_2$, nor $v \forces{\player}{\R} \C_1$ or $v \forces{\player}{\R}
\C_2$.

\begin{theorem} \label{th:fmist3}
  Let ${\R} \subseteq V \times V$ and $v,w\in V$. Then $\R$ is a \gstutname 
  iff $\R$
  is an equivalence relation and $v \R w$ implies:
  \begin{enumerate}
    \item[a)] $\prio{v} = \prio{w}$; \label{th:fmist3-prio}
    \item[b)] $v \forces{\player}{\R} \U$ iff $w\forces{\player}{\R} \U$ for 
    all $\player \in 
    \{\even,\odd\}, \U \subseteq \partition{V}{\R} \setminus 
    \{\class{v}{\R}\}$. 
    \label{th:fmist3-transfer}
  \end{enumerate}
\end{theorem}
\begin{proof}
We show that the second condition is equivalent to the conjunction of the last
two conditions in Theorem~\ref{th:fmist2}. We split the proof into an \emph{if}-part
and an \emph{only-if}-part.

\begin{itemize}
\item[$\Leftarrow$] The second condition from Theorem~\ref{th:fmist2} is equivalent to
the second condition above if we let $\U$ range only over singleton sets (if
$v \forces{\player}{\R} \C$, take $\U = \{\C\}$). The third condition is equivalent
to the second condition above, where $\U = \partition{V}{\R} \setminus \{\class{v}{\R}\}$.
This can be seen by appealing to Lemma~\ref{lem:force-vs-div}.

\item[$\Rightarrow$] Let $\R$ be a governed stuttering bisimulation and let $v,w \in V$
such that $v \R w$. Assume that $v \forces{\player}{\R} \U$ for some $\U \subseteq
\partition{V}{\R} \setminus \{\class{v}{\R}\}$. Let $S = \{u \in \class{v}{\R} ~|~ u \step \U\}$.
By Lemma~\ref{lem:exits}, either $S \cap V_{\player} \neq \emptyset$ or
there is some $u \in S$ such that $\post{u} \subseteq \U$. We consider both cases
separately.
\begin{itemize}
\item Case $S \cap V_{\player} \neq \emptyset$. Pick some $u$ from this set.
There is a class
$\C \in \U$ such that $u \step \C$ (in particular, $u \forces{\player}{\R} \C$ since $u \in V_{\player}$). By Theorem~\ref{th:fmist2}
then also $w \forces{\player}{\R} \C$, from which
$w \forces{\player}{\R} \U$ follows immediately.

\item Case $S \cap V_{\player} = \emptyset$. Then $\post{u} \subseteq
\U$ for some $u \in S$, but then, by Lemma~\ref{lem:exits_gen},
$\post{t} \setminus \class{v}{R} \subseteq \U$ for all $t \in \class{v}{R}$.
From $v \forces{\player}{\R} \U$ we derive, using
Lemma~\ref{lem:force-vs-div}, that not $v
\diverges{\opponent{\player}}{\R}$. By Theorem~\ref{th:fmist2} it
follows that not $w \diverges{\opponent{\player}}{\R}$, and by
Lemma~\ref{lem:force-vs-div} again $w \forces{\player}{\R} V \setminus
\class{v}{\R}$. Since for all $t \in \class{v}{R}$,
$\post{t} \subseteq \class{v}{\R} \cup \bigcup\U$ and $\bigcup\U \subseteq
V \setminus \class{v}{\R}$, by Lemma~\ref{lem:exits_restrict},
also $w \forces{\player}{\R} \U$.\qedhere
\end{itemize}

\end{itemize}
\end{proof}
In the previous theorem, we lifted the notion
of forcing play via the current equivalence class towards a target class, to the notion of forcing
a play via the current equivalence class towards a set of target classes. This is still not sufficient
for easily proving transitivity of governed stuttering bisimulation. Therefore, in the theorem below,
we introduce a final generalisation; rather than forcing play towards a set of target classes
via the current equivalence class, we now allow the play to be forced to that set via a set
of equivalence classes.
\begin{theorem} \label{th:fmist4}
  Let ${\R} \subseteq V \times V$ and $v,w\in V$. Then $\R$ is a \gstutname 
  iff $\R$ is an equivalence relation and $v \R w$ implies:
  \begin{enumerate}
    \item[a)] $\prio{v} = \prio{w}$; \label{th:fmist4-prio}
    \item[b)] $v \forces{\player}{\U} \T$ iff $w\forces{\player}{\U} \T$ for 
    all $\player \in 
    \{\even,\odd\}, \U,\T \subseteq \partition{V}{\R}$ such that
    $\class{v}{\R} \in \U$ and $\class{v}{\R} \notin \T$. 
    \label{th:fmist4-transfer}
  \end{enumerate}
\end{theorem}
\begin{proof}
We show that the second condition is equivalent to the second condition in
Theorem~\ref{th:fmist3}. We split the proof into an \emph{if}-case and an
\emph{only-if}-part.
\begin{itemize}
\item[$\Leftarrow$] The second condition from Theorem~\ref{th:fmist3} is equivalent
to the second condition above if we fix $\U = \{\class{v}{\R}\}$.

\item[$\Rightarrow$] Let $\R$ be a governed stuttering bisimulation and let
$\player, v, w, \U$ and $\T$ be as described. Assume that $v \forces{\player}{\U} \T$;
under this assumption we will prove that $w \forces{\player}{\U} \T$. The proof for
the implication in the other direction is completely symmetric. Let $\strategyname \in \strategy{\player}$ be
such that $v \forces{\strategyname}{\U} \T$ and consider the set of paths originating
in $v$ that are allowed by $\strategyname$. All these paths must have a prefix
$v \dots v', u$ such that $v, \dots, v' \notin \bigcup \T$ but $u \in \bigcup \T$. Call these
prefixes the $\strategyname$-prefixes of $v$. 

We proceed by induction on the length of the longest such prefix. 
If the longest
prefix has length 2, then all prefixes have length 2, implying that $v \steps{\player} \T$.
In particular, $v \forces{\player}{\R} \T$ and by Theorem~\ref{th:fmist3} also
$w \forces{\player}{\R} \T$, which proves $w \forces{\player}{\U} \T$.

As the induction hypothesis, assume that if $u \R u'$, $u \forces{\strategyname}{\U} \T$
and the longest $\strategyname$-prefix of $u$ is shorter than the longest $\strategyname$-prefix
of $v$, then $u' \forces{\player}{\U} \T$. Note that every $\strategyname$-prefix
$\pathvar$ of $v$ must have a first position $n$ such that $\ind{\pathvar}{n} \notin
\class{v}{\R}$. Collect all these $\ind{\pathvar}{n}$ in a set $U$, and notice that
for all $u \in U$, also $u \forces{\strategyname}{\U} \T$.
Furthermore, $v \forces{\strategyname}{\R} U$. 

By Theorem~\ref{th:fmist3},
$w \forces{\player}{\R} \class{U}{\R}$. Now consider an arbitrary $u' \in \bigcup
\class{U}{\R}$. Because there is some $u \in U$ such that $u \R u'$, its longest $\strategyname$-prefix is shorter than the longest $\strategyname$-prefix of $v$, and
because $u \forces{\strategyname}{\U} \T$ for such $u$, we can use the induction
hypothesis to derive that $u' \forces{\player}{\U} \T$.

The above in particular implies two facts: $w \forces{\player}{\U} \class{U}{\R}$, and
$u' \forces{\player}{\U} \T$ for all $u' \in \bigcup \class{U}{\R}$. Using these,
we can now apply Lemma~\ref{lem:glue-strategies} to conclude 
$w \forces{\player}{\U} \T$.\qedhere

\end{itemize}
\end{proof}

With this last characterisation, it is now straightforward to prove that
governed stuttering bisimilarity is an equivalence relation. We do so by
showing that the transitive closure of the union of two governed stuttering
bisimulations $R$ and $S$ is again a governed stuttering bisimulation.
The generalisation from classes to sets of classes allows us to view equivalence
classes in $(R \cup S)^*$ as the union of sets of equivalence classes of
$R$ (or $S$), giving us an easy way to compare the effect of the second
requirement of Theorem~\ref{th:fmist4} on $(R \cup S)^*$ with its effect on
$R$ and $S$.
\begin{theorem}\label{th:fmist-equiv}
  $\gstut$ is an equivalence relation.
\end{theorem}
\begin{proof}
  We show that $(R \cup S)^*$ is a governed stuttering bisimulation if $R$ and
$S$ are, by showing that $(R \cup S)^*$ satisfies the conditions of Theorem~\ref{th:fmist4}
if $R$ and $S$ do. If $v,w \in V$ are related under $(R \cup S)^*$, then there exists
a sequence of vertices $u_0, \dots, u_n$ such that $v \R u_0 \S \dots \R u_n \S w$
(the strict alternation between the two relations can always be achieved because
$R$ and $S$ are reflexive). By transitivity of $=$ we then have $\prio{v} = \prio{w}$,
so the first property is satisfied.

For the second property, assume that $v \forces{\player}{\U} \T$ for some $\player
\in \{\even,\odd\}$ and some $\U,\T \subseteq \partition{V}{(R \cup S)^*}$ such that
$\class{v}{(R \cup S)^*} \in \U$ and $\class{v}{(R \cup S)^*} \notin \T$. We need to
prove that $w \forces{\player}{\U} \T$. Note that $R$ and $S$ both refine $(R \cup S)^*$,
so we can find sets $\U_R \subseteq \partition{V}{\R}$ and $\U_S \subseteq \partition{V}{\S}$
such that $\bigcup \U_R = \bigcup \U_S = \bigcup \U$. Because
$v \forces{\player}{\U} \T$, also $v \forces{\player}{\U_R} \T$, and by
Theorem~\ref{th:fmist4} then $u_0 \forces{\player}{\U_R} \T$, which is equivalent
to $u_0 \forces{\player}{\U_S} \T$. By a simple inductive argument we now arrive
at $w \forces{\player}{\U_S} \T$, which is equivalent to $w \forces{\player}{\U} \T$.
\qedhere 
\end{proof}
As a side-result of the proof of Theorem~\ref{th:fmist-equiv}, we
find that the union of \emph{all} governed stuttering bisimulations
is again a governed stuttering bisimulation, which coincides with
governed stuttering bisimilarity.\medskip

In order to better understand the differences between governed
stuttering bisimulation and, \eg delayed simulation equivalence,
we next provide a game-based characterisation of the relation. While
in this new game, \spoiler and \duplicator still move according to the
same rules as in the delayed simulation game, \duplicator now has
more freedom to choose a new configuration: she can now also choose
to `roll-back' one of the proposed moves. This allows her to postpone
matching a move. Of course, such moves may not be postponed
indefinitely, so some additional mechanism is needed to keep track of
\duplicator's
progress so as to prevent \duplicator from becoming too powerful.
For this, we use a system of challenges and rewards: a $\dagger$-challenge
indicates \duplicator decided to match a move by \spoiler by not
moving; a $\checkmark$-reward indicates \duplicator matched a move by
\spoiler by making a countermove, and a challenge $(k,u)$ taken from $\{0,1\} \times V$
indicates that \duplicator is in the process of matching a move to
vertex $u$. We let $C$ denote the set of challenges $(\{0,1\} \times V) \cup
\{\dagger,\checkmark\}$.

\begin{definition}[Governed Stuttering bisimulation game]
\label{def:gstut_game}
The governed stuttering bisimulation game is played on an arena of
configurations drawn from $(V \times V) \times C$, and it is
played in rounds. A round of the game starting in a configuration
$((v,w),c)$ proceeds as follows:

\begin{enumerate}
\item \spoiler chooses to play from $(u_0,u_1) \in \{(v,w), (w, v)\}$;

\item the players move from $(u_0,u_1)$ to $(t_0,t_1)$ according to
the rules in Table~\ref{tab:simulation_game_rules};

\item \duplicator selects a new configuration drawn from the following set:
\[\{ ((t_0,t_1),\checkmark),\
((u_0,t_1),\update(c,(0,t_0),v,u_0)),\
((t_0,u_1),\update(c,(1,t_1),w,u_1)) \}
\] 
where update $\update$ is defined as follows:
\[\update(c,c',u,t) = \left \{ \begin{array}{ll}
 c' & \text{if \spoiler played on $t$, $u = t$ and $c \in \{\dagger,\checkmark,c'\}$} \\
 \checkmark &\text{if $u \neq t$, or \spoiler played on $t$ and 
$c \notin \{\dagger,\checkmark,c'\}$} \\
 \dagger &\text{otherwise}
\end{array}
\right .
\] \\

\end{enumerate}

An infinite play $((v_0,w_0),c_0), ((v_1,w_1),c_1), \dots$ is won by
\duplicator if and only if $\prio{v_j} = \prio{w_j}$ for all $j$ and $c_k = \checkmark$
for infinitely many $k$.
\duplicator wins the governed stuttering bisimulation game for a
position $(v,w)$ iff she
has a strategy that wins all plays starting in configuration $((v,w),\checkmark)$.  
\end{definition}
We write $v \gstutg w$ whenever \duplicator wins the governed stuttering bisimulation game
for position $(v,w)$.

Observe that in the governed stuttering game, \duplicator earns, as explained
before, a $\checkmark$ reward whenever she continues playing in the position
determined at the end of step~2. However, she also earns a $\checkmark$ whenever
\spoiler decides to drop a pending challenge or, in step~1 of a round, switch
positions. The example below illustrates some of the intricacies in the game
play.
\begin{example}
Consider the parity game depicted in Figure~\ref{fig:fake-divergence}. In this parity game, all vertices with priority $0$ are related by $\gstut$. The game illustrates why \duplicator gains a $\checkmark$ reward whenever \spoiler does not respect a pending challenge. This can be seen as follows: consider the game starting in $((v_1, v_3), \checkmark)$ and suppose \spoiler decides to play $v_1 \to v_2$. The only suitable response by \duplicator is to play $v_3 \to v_4$. New configurations $((v_2, v_4), \checkmark)$ and $((v_2, v_3, \checkmark))$ are not an option for \duplicator since he immediately loses due to the different priorities of $v_2$ and $v_4$ or $v_3$ respectively. The new configuration chosen by \duplicator will hence be $((v_1, v_3), (0, v_2))$, challenging \spoiler to play $v_1 \to v_2$ again in the next round. From this configuration, if \spoiler indeed plays $v_1 \to v_2$, \duplicator can match with $v_4 \to v_2$, and play stays in $((v_2, v_2), \checkmark)$ indefinitely, leading to a win from duplicator. Now, let us consider what happens if \spoiler plays $v_1 \to v_0$ instead. \spoiler did not respect the challenge, and \duplicator matches with $v_4 \to v_3$, and we end up in $((v_1, v_3), \checkmark)$ again. If \duplicator would not have earned a \checkmark reward in this case, play would have ended up in $((v_1, v_3), (0, v_0))$ instead, and, if in the next round \spoiler again ignores the challenge, play can alternate indefinitely between $((v_1, v_3), (0, v_0))$ and $((v_1, v_4), (0, v_2))$, which would result in a win for \spoiler. This is undesirable since we already observed that $v_1, v_3$ and $v_4$ are governed stuttering bisimilar.
\end{example}

For the remainder of this section we turn our attention to proving that the
governed stuttering bisimulation game and governed stuttering bisimulation
coincide. Our next result states that whenever vertices $v,w$ are
governed stuttering bisimilar, \duplicator wins all plays starting in
configuration $((v,w),\checkmark)$. We sketch the main ideas behind the 
proof; details can be found in the Appendix.

\begin{proposition}\label{prop:soundness_gstut}
  For all $v,w \in V$ if $v \gstut w$ then $v \gstutg w$.
\end{proposition}
\begin{proof} The proof proceeds by showing that \duplicator has
  a strategy that ensures 1) that plays allowed by this strategy move
  along configurations of the form $((u_0,u_1),c)$ for which $u_0 \gstut u_1$
  and 2) \duplicator never gets stuck playing according to this strategy
  and 3) there is a strictly decreasing measure between two consecutive
  non-$\checkmark$ configurations on any play allowed by this strategy.
  Together, this implies that \duplicator has a winning strategy for
  configurations $((v,w),\checkmark)$.\qedhere
\end{proof}
We next establish that vertices related through the
governed stuttering bisimulation game are related by governed stuttering
bisimulation. A straightforward proof thereof is hampered by the fact that
any purported governed stuttering bisimulation relation is, by definition,
required to be an equivalence relation. However, proving that the governed
stuttering bisimulation game induces an equivalence relation is rather
difficult. The strategy employed to prove the stated result is to use contraposition;
this requires showing that for any given pair of non-governed stuttering
bisimilar vertices we can construct a strategy that is winning for \spoiler.
Note that we can do so because the governed stuttering bisimulation game
has a B\"uchi winning condition, which implies the game is determined.
This strategy is based on a fixpoint characterisation of governed
stuttering bisimilarity, given below.

\begin{definition}
Let ${\R} \subseteq V \times V$ be an equivalence relation on $V$. The
predicate transformer 
$\mathcal{F} : V \times V \to V \times V$ is defined as follows:
\[
\mathcal{F}(\R) =
\begin{array}[t]{ll}
\{
(v,w) \in R \mid &
\prio{v} = \prio{w} \wedge
\forall{\player \in \{\even,\odd\}, \mathcal{U},\mathcal{T} \subseteq \partition{V}{\R}}\\
& \class{v}{\R} \in \mathcal{U} \wedge \class{v}{\R} \notin \mathcal{T} 
\implies
v \forces{\player}{\mathcal{U}}{\mathcal{T}} 
\Leftrightarrow
w \forces{\player}{\mathcal{U}}{\mathcal{T}} \}
\end{array}
\]
\end{definition}
The predicate transformer $\mathcal{F}$ has the following properties.
\begin{lemma}
$\mathcal{F}(R)$ is an equivalence relation for any equivalence relation $R$ on
$V$.
\end{lemma}
\begin{proof}
Let $R$ be an equivalence relation over $V$. Reflexivity of
$\mathcal{F}(R)$ follows from the fact that $R$ is reflexive.
Symmetry of $\mathcal{F}(R)$ follows from symmetry of $R$ and
from the fact that for all $(v,w) \in R$ we have $\class{v}{R} = \class{w}{R}$. 
For transitivity, we observe that for all pairs $(u,v), (v,w) \in \mathcal{F}(R)$ and
all $\mathcal{U},\mathcal{T}$ for which $\class{u}{R} \in \mathcal{U}$ and
$\class{u}{R} \notin \mathcal{T}$, if
$u \forces{\player}{\mathcal{U}}{\mathcal{T}}$ then also 
$v \forces{\player}{\mathcal{U}}{\mathcal{T}}$. Since $\class{u}{R} = \class{v}{R}$,
we immediately conclude $w \forces{\player}{\mathcal{U}}{\mathcal{T}}$.
The implication from right to left follows from symmetric arguments.
Thus, $\mathcal{F}(\R)$ is an equivalence relation.\qedhere
\end{proof}

\begin{lemma} $\mathcal{F}$ is a monotone operator on the complete lattice
of equivalence relations on $V$.
\end{lemma}
\begin{proof} 
By the previous lemma, it follows that $\mathcal{F}$ is an operator on the
lattice of equivalence relations on $V$. We next show that the operator
is monotone. Let $R,S$ be arbitrary equivalences on $V$. Suppose
$R \subseteq S$, and consider some pair $(v,w) \in \mathcal{F}(R)$.
From this, it follows that $\prio{v} = \prio{w}$. Let
$\player \in \{\even,\odd\}$, $\mathcal{U},\mathcal{T} \subseteq \partition{V}{S}$,
such that $\class{v}{S} \in \mathcal{U}$ and $\class{v}{S} \notin \mathcal{T}$.
Define $\bar{\mathcal{U}}$ as the set $\{ \class{u}{R} \mid \class{u}{S} \in \mathcal{U}\}$
and define $\bar{\mathcal{T}}$ as the set $\{ \class{u}{R} \mid \class{u}{S} \in \mathcal{T}\}$.
Since $(v,w) \in \mathcal{F}(R)$, we have:
\[v \forces{\player}{\bar{\mathcal{U}}}{\bar{\mathcal{T}}} 
\Leftrightarrow
w \forces{\player}{\bar{\mathcal{U}}}{\bar{\mathcal{T}}}
\]
Since $\bigcup\bar{\mathcal{U}} =\bigcup{\mathcal{U}}$ and
$\bigcup\bar{\mathcal{T}} =\bigcup{\mathcal{U}}$, we immediately have:
\[v \forces{\player}{\mathcal{U}}{\mathcal{T}}
\Leftrightarrow
w \forces{\player}{\mathcal{U}}{\mathcal{T}}
\]
This proves that $(v,w) \in \mathcal{F}(S)$. Thus $\mathcal{F}$ is a monotone
operator.\qedhere\end{proof}
\begin{corollary}\label{cor:fmist_fixpoint}
We have $\gstut{} = \nu \mathcal{F}$.
\end{corollary}
\begin{proof}
Follows from the fact that for $R = \nu \mathcal{F}$ and 
$\nu \mathcal{F} = \mathcal{F}(\nu \mathcal{F})$ the definition
of $\mathcal{F}$ reduces to the definition of governed stuttering
bisimulation.\qedhere
\end{proof}
We finally state our completeness result. Again, we only outline
the main steps of the proof; details can be found in the Appendix.

\begin{proposition}
\label{prop:completeness_gstut}
For all $v,w \in V$ if $v \gstutg w$ then $v \gstut w$.
\end{proposition}

\begin{proof}
We essentially prove the contrapositive of the statement, \ie for all $v,w \in
V$, if $v \not\gstut w$, then also $v \not\gstutg w$. Let $v
\not\gstut w$. By Corollary~\ref{cor:fmist_fixpoint}, then also
$(v,w) \notin \nu \mathcal{F}$. By the Tarski-Kleene fixpoint
approximation theorem, we thus have $(v,w) \notin \bigcap\limits_{k
\ge 1} \mathcal{F}^k(V \times V)$. Using induction, one can prove,
for $R^k = \bigcap_{l \le k} R^l$, that for all $k \ge 1$:
\[
\tag{IH}
\begin{array}{l}
\text{\spoiler wins the governed stuttering bisimulation game} \\
\text{for all configurations 
$((u_0,u_1),c)$ for which $(u_0,u_1) \notin R^k$}
\\
\end{array}
\]
For the inductive case, one can construct a strategy for
\spoiler that guarantees he never gets stuck and for which
every play allowed by the strategy either 1)
visits some configuration $((t_0,t_1),c')$
for which the induction hypothesis applies, or 2) is such that 
there are only a finite number of $\checkmark$ rewards along the play.
\qedhere
\end{proof}
Propositions~\ref{prop:soundness_gstut} and~\ref{prop:completeness_gstut}
lead to the following theorem.
\begin{theorem}\label{thm:gstut_game_is_gstut}
For all $v,w \in V$ we have $v \gstut w$ iff $v \gstutg w$.
\end{theorem}

\paragraph{Stuttering Bisimulation.} When we impose the additional
constraint on governed stuttering bisimulation that we do not allow to relate
vertices owned by different players, we obtain a notion called
\emph{stuttering bisimulation}~\cite{CKW:11}. The derived notion of
\emph{stuttering bisimilarity}, denoted $v \stut w$ and defined as $v
\stut w$ iff there is some stuttering bisimulation relation that relates
$v$ and $w$, is an equivalence relation.

\section{Quotienting}
\label{sec:quotients}

Simulation and bisimulation equivalences are often used to reduce
the size of graphs by factoring out vertices that are equivalent,
\ie by computing \emph{quotient structures}. This can be particularly
interesting if computationally expensive algorithms must be run on
the graph: whenever the analysis such algorithms perform on the
graphs are insensitive to (bi)simulation equivalence, they can be
run on the smaller quotient structures instead. In our setting, the
same reasoning applies: typically, parity game solving is expensive
and it may therefore pay off to first compute a quotient structure
and only then solve the resulting quotient structure.

In this section, we show that most of the (bi)simulation relations
we studied in the previous two sections have unique quotient
structures. A fundamental property of quotienting is that the
resulting quotient structure of a game should again be equivalent
to the original game. This requires that we lift the equivalences
defined on game graphs to equivalences between two different
game graphs. We do so in the standard way.
\begin{definition}
Let $\game_j = (V_j,\to_j,\priosym_j,\getplayername_j)$, for $j = 1, 2$,
be arbitrary parity games. We say that $\game_1 \sim \game_2$, for an
equivalence relation $\sim$ defined on the vertices of a parity game, whenever
in the disjoint union of $\game_1$ and $\game_2$, 
for all $v_1 \in V_1$ there is some $v_2 \in V_2$ such that
$v_1 \sim v_2$ and for all $\bar{v}_2 \in V_2$ there is some $\bar{v}_1 \in V_1$ such
that $\bar{v}_1 \sim \bar{v}_2$.
\end{definition}

\subsection{Simulation Equivalence Quotients}
\label{sec:direct_sim_quotient}

Quotienting for delayed simulation equivalence is, as observed
in~\cite{Fri:05,FW:06}, problematic, and only the biased versions
admit some form of quotienting.  However, the quotients for biased
delayed simulation equivalences are not unique, see also Lemma 3.5
in~\cite{Fri:05}.  We therefore only consider quotienting for direct
simulation equivalence.

The equivalence classes of direct simulation equivalence determine
the set of vertices of the quotient structure. Defining the transition
relation of the quotient structure is a bit more subtle. As observed
in~\cite{BG:03}, a unique quotient structure of simulation
equivalence for Kripke structures exists, but requires that vertices
have no transitions to a pair of vertices, one of which is sometimes
referred to as a \emph{`little brother'} of the other one (a vertex that is simulated
by, but not equivalent to the other vertex). 

While in the setting
of Kripke structures, only transitions to \emph{maximal} successor
vertices must be retained, depending on the owner of the source
vertex, we need to consider maximal or \emph{minimal} successor
vertices. 

\begin{definition} Let $V' \subseteq V$ be an arbitrary non-empty
set of vertices.  An element $v$ is:
\begin{itemize}
\item \emph{minimal among $V'$} iff for
all $u \in V'$ for which $u \directsim v$, also $v \directsim u$;
\item \emph{maximal among $V'$} iff for all $u \in V'$ for which $v
\directsim u$, also $u \directsim v$.
\end{itemize}
For a given vertex $v$, a successor $v' \in \post{v}$ is in the set
$\minsucc{v}$ iff $v'$ is minimal among $\post{v}$; likewise, $v' \in \post{v}$
is in the set $\maxsucc{v}$ iff $v'$ is maximal among $\post{v}$.

\end{definition}
Since $\directsim$ is a preorder, 
$\minsucc{v}$ and $\maxsucc{v}$ are non-empty sets.

An additional complication in defining a unique quotient
structure is that a single equivalence class may contain vertices
owned by \even and vertices owned by \odd. It turns out that the
owner of such equivalence classes can be chosen arbitrarily: we
prove that such classes have a unique successor equivalence class.
For equivalence classes with exactly one successor, we can assign a
unique owner; we choose to assign such classes to player \even.

\begin{definition}[Direct simulation equivalence quotient]
\label{def:pg_dsim_quotient}
The direct simulation equivalence quotient of $(V,\step,\priosym,\getplayername)$
is the structure
$(\partition{V}{\directsimeq}, \to', \priosym', \getplayername')$, where, for
$\C,\C' \in \partition{V}{\directsimeq}$:

\begin{itemize}
  \item $\priosym'(\C) = \min\{\prio{v} ~|~ v \in \C\}$,
  \item $\getplayername'(\C) = \begin{cases}
   \ensuremath{\odd} & \text{if $\C \subseteq V_{\odd}$ and for all $u \in \C$,
 $|\class{\minsucc{u}}{\directsimeq}| > 1$} \\
   \ensuremath{\even} & \text{otherwise}
   \end{cases} $
  \item $\C \to' \C'$ iff
   $\begin{cases}
    \forall{v \in \C} \exists{v' \in \minsucc{v}} v' \in \C' & \text{ if $\C \subseteq V_{\odd}$} \\
    \forall{v \in \C \cap V_{\even}} \exists{v' \in \maxsucc{v}} v' \in \C' & \text{ otherwise}
 \end{cases} $
    
\end{itemize}
\end{definition}
Observe that it is not obvious that $\to'$ is a total edge relation. The lemma 
below allows us to establish that this is the case.
\begin{lemma}\label{lem:to_minimal_maximal}
Let $\C \in \partition{V}{\directsimeq}$. Then:
\begin{itemize}
\item If $\C \subseteq V_{\odd}$
then 
$\class{\minsucc{v}}{\directsimeq} =
\class{\minsucc{w}}{\directsimeq}$
for all $v,w \in \C$, 

\item If $\C \subseteq V_{\even}$
then 
$\class{\maxsucc{v}}{\directsimeq} =
\class{\maxsucc{w}}{\directsimeq}$
for all $v,w \in \C$, 

\item If $\C \not\subseteq V_{\odd}$ and $\C \not\subseteq V_{\even}$ then for all
$v \in \C \cap V_{\even}$ and $w \in \C \cap V_{\odd}$  we have
$\class{\maxsucc{v}}{\directsimeq} = \class{\minsucc{w}}{\directsimeq}$.

\end{itemize}
\end{lemma}
\begin{proof}
We prove the first and the third statement; the proof for the second statement is analogous to
that of the first.
\begin{itemize}
\item Suppose $\C \subseteq V_{\odd}$.  Pick $v,w \in \C$. Let
$v' \in \minsucc{v}$. Since $v' \in \post{v}$
and $w \directsim v$, we have $w' \directsim v'$ for some $w' \in \post{w}$. 
This implies that there is some $w'' \in \minsucc{w}$ such that $w'' \directsim v'$;
for, if $w' \notin \minsucc{w}$, then there must be some $w'' \in \minsucc{w}$ such that
$w'' \directsim w'$. But then also $w'' \directsim v'$.

We next show that also $v' \directsim w''$. Since $v \directsim w$ and $w'' \in
\post{w}$ we have $v'' \directsim w''$ for some $v'' \in \post{v}$. Since
$w'' \directsim v'$ and $v'' \directsim w''$, we have $v'' \directsim v'$. But since $v' \in \minsucc{v}$, this implies $v' \directsimeq v''$.
But from $v' \directsim v''$ and $v'' \directsim w''$ we obtain $v' \directsim w''$.

Hence, $v' \directsimeq w''$ for some $w'' \in \minsucc{w}$.

\item Suppose $\getplayer{v} \neq \getplayer{w}$ for some $v,w \in \C$. Pick $v,w \in \C$
such that $v \in V_{\even}$ and $w \in V_{\odd}$. Since $w \directsim v$,
there must be $w' \in \post{w}$ and $v' \in \post{v}$ such that $w' \directsim v'$.
Fix such $v'$ and $w'$. Since $v \directsim w$ we find that for all $v'' \in \post{v}$
and $w'' \in \post{w}$ we have $v'' \directsim w''$. In particular, $v' \directsim w'$.
So $v' \directsimeq w'$.

Next, since for all $w'' \in \post{w}$ we have $v' \directsim w''$ and $v' \directsimeq w'$,
we also have $w' \directsim w''$ for all $w'' \in \post{w}$. But this implies
$w' \in \minsucc{w}$, and, in particular, $|\class{\minsucc{w}}{\directsimeq}| = 1$.
Likewise, we can deduce $v' \in \maxsucc{v}$ and 
$|\class{\maxsucc{v}}{\directsimeq}| = 1$.

We thus find $\class{\maxsucc{v}}{\directsimeq} = 
\{ \class{v'}{\directsimeq} \} = \{ \class{w'}{\directsimeq}
\} = \class{\minsucc{w}}{\directsimeq}$. \qedhere
\end{itemize}

\end{proof}
As a consequence of the above lemma, we obtain the following two results:
\begin{corollary}
\label{cor:exists_implies_all}
Let $(\partition{V}{\directsimeq},\to',\priosym', \getplayername')$ be a
direct simulation equivalence quotient of some parity game $(V,\to,\priosym,\getplayername)$.
Then for all $\C,\C' \in \partition{V}{\directsimeq}$:
\begin{itemize}
\item if $\C \subseteq V_{\odd}$ and for some $v \in \C$, $v' \in \C'$ 
also $v' \in \minsucc{v}$, then $\C \to' \C'$.

\item if $\C \cap V_{\even} \neq \emptyset$ and for some $v \in
\C \cap V_{\even}$, $v' \in \C'$ also $v' \in \maxsucc{v}$,
then $\C \to' \C'$.

\end{itemize}
\end{corollary}
\begin{corollary}
The direct simulation reduced quotient structure associated to a 
parity game $(V, \to \priosym, \getplayername)$ is again a parity game.
\end{corollary}
We next establish that the direct simulation quotient of a parity game is
equivalent to the original parity game. 

\begin{theorem}\label{thm:direct_sim_quotient}
Let $\game = (V, \to, \priosym, \getplayername)$ be a parity game
and $\game_q = (\partition{V}{\directsimeq}, \to', \priosym',
\getplayername')$ its direct simulation quotient. Then $\game_q
\directsimeq \game$.

\end{theorem}
\begin{proof}
  We first prove $\game_q \directsim \game$. 
  Let $H \subseteq \partition{V}{\directsimeq} \times V$ be
  the relation $H = \{ (\C, v) \mid \exists{w \in \C} w
  \directsim v \}$. We prove $H$ is a direct simulation relation.
  Let $\C, v$ be arbitrary such that $(\C,v) \in H$. By definition,
  there is some $w \in \C$ such that $w \directsim v$. For the
  remainder of this proof, we fix such a $w$. Clearly,
  $\priosym'(\C) = \prio{v}$ follows directly from $w \directsim
  v$ and the fact that for all $u,u' \in \C$ we have
  $\prio{u} = \prio{u'}$. We next prove the transfer condition. For this,
  we distinguish four cases.

  \begin{itemize}
   \item Case $\getplayer{\C} = \odd$ and $\getplayer{v} = \odd$.
   Let $v' \in \post{v}$. Since $v' \in \post{v}$,
   also $w' \directsim v'$ for some $w' \in \post{w}$. Let $w^-
   \in \minsucc{w}$ such that $w^- \directsim w'$. 
   Then by Corollary~\ref{cor:exists_implies_all},
   $\C \to \class{w^-}{\directsimeq}$.
   Moreover, by transitivity, $w^- \directsim v'$, so 
   we have $(\class{w^-}{\directsimeq}, v') \in H$, as required.
   
   \item Case $\getplayer{\C} = \odd$ and $\getplayer{v} = \even$.
   Note that $\getplayer{\C} = \odd$ implies $\C \subseteq V_{\odd}$; 
   hence $\getplayer{w} = \odd$.
   But then for some $w' \in \post{w}, v' \in \post{v}$, we have
   $w' \directsim v'$. Let $v',w'$ be such. Then again for some
   $w^- \in \minsucc{w}$ satisfying $w^- \directsim w'$ we have
   $\C \to \class{w^-}{\directsimeq}$ and by transitivity, we 
   have $(\class{w^-}{\directsimeq}, v') \in H$.

   \item Case $\getplayer{\C} = \even$ and $\getplayer{v} =\even$.
   Pick $\C' \in \post{\C}$. We must show that $(\C',v') \in H$ for some
   $v' \in \post{v}$. Since $\getplayer{\C} = \even$, there must be some
   $u \in \C \cap V_{\even}$; pick such a $u$. Since
   $u \directsimeq w \directsim v$, for all $u' \in \post{u}\cap \C'$  there
   is some $v' \in \post{v}$ such that $u' \directsim v'$. Hence,
   there is some $v' \in \post{v}$ such that $(\C',v') \in H$.

   \item Case $\getplayer{\C} = \even$ and $\getplayer{v} = \odd$.
   Pick $\C' \in \post{\C}$. We must show that $(\C',v') \in H$ for all
   $v' \in \post{v}$. Then the argument is similar to the previous case.
   \end{itemize}
  Next, to prove $\game \directsim \game_q$,
   we show that $H \subseteq V \times \partition{V}{\directsimeq}$, given by
  $H = \{ (v,\C) \mid \exists{w \in \C} v
  \directsim w \}$, is a direct simulation relation.
  Let $\C, v$ be arbitrary such that $(v,\C) \in H$. By definition,
  there is some $w \in \C$ such that $v \directsim w$. We again fix such
  a $w$.  Following similar arguments as above,
  $\priosym'(\C) = \prio{v}$.
  We again prove the transfer condition by distinguishing four
  cases.

  \begin{itemize}
  \item Case $\getplayer{\C} = \even$ and $\getplayer{v} = \even$. Pick
  $v' \in \post{v}$. We must show that $(v',\C') \in H$ for some $\C' \in \post{\C}$.
  We distinguish two cases.
  \begin{itemize}
  \item Case $\C \subseteq V_{\even}$.  
  Since $v \directsim w$, also $v' \directsim w'$ for some
  $w' \in \post{w}$.  Consider $w^+ \in \maxsucc{w}$ such that $w' \directsim w^+$. 
  Then by
  Corollary~\ref{cor:exists_implies_all}, $\C \to
  \class{w^+}{\directsimeq}$. By transitivity $v' \directsim w'\directsim
  w^+$; hence  $(v', \class{w^+}{\directsimeq}) \in H$.

  \item Case $\C \cap V_{\odd} \neq \emptyset$ and $\C \cap V_{\even}
  \neq \emptyset$.  Then, by Lemma~\ref{lem:to_minimal_maximal},
  $\C \to \C'$ for some unique $\C'$.  Fix this $\C'$. Let $u \in \C \cap
  V_{\odd}$. Then $u' \in \C'$ for some $u' \in \post{u}$. Pick such a $u'$.
  Since $u \directsimeq w$, also $v \directsim u$. As a result, $v' \directsim u'$. 
  Therefore $(v', \C') \in H$.

  \end{itemize} 

  \item Case $\getplayer{\C} = \even$ and $\getplayer{v} = \odd$.
  We must show that $(v',\C') \in H$ for some $v' \in \post{v}, \C' \in \post{C}$.
  Then the argument is similar to the previous case.


  %

  \item Case $\getplayer{\C} = \odd$ and $\getplayer{v} = \even$.
  Let $v' \in \post{v}$ and $\C' \in \post{\C}$. 
  Since $v \directsim w$ and $\getplayer{w} =
  \odd$, we have $v' \directsim w'$ for all $w' \in \post{w} \cap \C'$.  
  Therefore also $(v',\C') \in H$.

  \item Case $\getplayer{\C} = \odd$ and $\getplayer{v} = \odd$. Let $\C' \in \post{\C}$
  and let $w' \in \post{w} \cap \C'$. Since $v \directsim w$, there must be some
  $v' \in \post{v}$ such that $v' \directsim w'$. Fix this $v'$. Then also
  $(v', \C') \in H$.\qedhere\end{itemize}

\end{proof}
We finally prove that the quotient is unique.
\begin{theorem}\label{thm:unique_quotient_dsim}
Let $\game,\game'$ be two parity games and let $\game_q$ and $\game_q'$ be
their direct simulation equivalence quotients, respectively. Then 
$\game \directsimeq \game'$ iff
the two structures $\game_q = (\partition{V}{\directsimeq}, \step_q,
\priosym_q, \getplayername_q)$ and $\game_q' = (\partition{V'}{\directsimeq}, 
\step'_q, \priosym'_q, \getplayername'_q)$ are isomorphic.
\end{theorem}
\begin{proof}
The proof that isomorphism of $\game_q$ and $\game_q'$ implies
$\game \directsimeq \game'$ follows essentially from
Theorem~\ref{thm:direct_sim_quotient} and that isomorphic structures
are also direct simulation equivalent.

The proof that $\game \directsimeq \game'$ implies that $\game_q$
and $\game_q'$ are isomorphic structures follows the following
steps.  Assume that $\game \directsimeq \game'$.  Let $f \subseteq
\partition{V'}{\directsimeq} \times \partition{V}{\directsimeq}$
be defined as $(\C', \C) \in f$ iff $\C' \directsimeq \C$.  Note
that for all $(\C',\C) \in f$ we have $\priosym'_q(\C') =
\priosym_q(\C)$.

We first show that $f$ is a total bijective function
from $\partition{V'}{\directsimeq}$ to $\partition{V}{\directsimeq}$.
For injectivity and functionality of $f$ we reason as follows. Suppose $f$
is not functional. Then there is some $v' \in V'$ and two $v,\bar{v}
\in V$ such that $\class{v}{\directsimeq} \neq
\class{\bar{v}}{\directsimeq}$,
$(\class{v'}{\directsimeq},\class{v}{\directsimeq}) \in f$ and
$(\class{v'}{\directsimeq},\class{\bar{v}}{\directsimeq}) \in f$. Then
by definition, $v' \directsimeq v$ and $v' \directsimeq \bar{v}$.
But then also $v \directsimeq \bar{v}$, contradicting that
$\class{v}{\directsimeq} \neq \class{\bar{v}}{\directsimeq}$.  So
$f$ is functional. The proof that $f^{-1}$ is a function from
$\partition{V}{\directsimeq}$ to $\partition{V'}{\directsimeq}$ is
similar. We may therefore conclude that $f$ is an injective function.

For surjectivity of $f$, we observe that by definition of $\game
\directsimeq \game'$, for each $v \in V$ there is some $v' \in V'$
such that $v \directsimeq v'$.
Hence, for each $\class{v}{\directsimeq} \in \partition{V}{\directsimeq}$
there is some $\class{v}{\directsimeq} \in \partition{V'}{\directsimeq}$
such that $(\class{v'}{\directsimeq}, \class{v}{\directsimeq}) \in
f$. Similarly, we can show that $f^{-1}$ is surjective and therefore
$f$ is total bijection.

We next prove that $\getplayername_q'(\C) = \getplayername_q(f(\C))$.
Towards the contrary, assume that $\getplayername_q'(\C) = \odd$ whereas
$\getplayername_q(f(\C)) = \even$ for some $\C$. Then $\C \subseteq V'_{\odd}$
and for all $v \in \C$ we have $|\minsucc{v}| > 1$, and
there is some $w \in f(\C)$ satisfying either $w \in  V_{\even}$, or 
$|\minsucc{w}| = 1$. Let $w \in f(\C)$ be such and pick an
arbitrary $v \in \C$. We distinguish two
cases. 
\begin{itemize}
\item Case $w \in V_{\even}$. Since $w \directsimeq f(\C) \directsimeq \C
\directsimeq v$ we have $w \directsim v$ in particular. Pick an arbitrary
$w' \in \post{w}$. Then, since $\getplayer{v} = \odd$, we have $w' \directsim v'$ for all $v' \in \post{v}$;
more specifically, we have $w' \directsim v_1'$ and $w' \directsim v_2'$
for $v_1',v_2' \in \minsucc{v}$ such that $v_1' \not\directsimeq v_2'$.
Since $v_1',v_2'$ are minimal elements, we thus also have $v_1' \directsim w'$
and $v_2' \directsim w'$ and hence $v_1' \directsimeq w'$ and
$v_2' \directsimeq w'$. But from this we obtain $v_1' \directsimeq v_2'$.
Contradiction.

\item Case $|\minsucc{w}| = 1$. Without loss of generality we may assume
that $w \in V_{\odd}$. Since $w \directsimeq f(\C) \directsimeq \C \directsimeq v$
we also have $w \directsim v$. Let $v_1',v_2' \in \minsucc{v}$ be such that
$v_1' \not\directsimeq v_2'$. Then there must be some $w_1', w_2' \in
\post{w}$ such that $w_1' \directsim v_1'$ and $w_2' \directsim v_2'$.
Let $w_1',w_2'$ be such; without loss of generality, we may assume that 
$w_1'$ and $w_2'$ are minimal.
Since $v_1',v_2'$ are minimal, we find that $v_1' \directsim w_1'$ and
$v_2' \directsim w_2'$ and hence $v_1' \directsimeq w_1'$ and
$v_2' \directsimeq w_2'$. But because $v_1' \not\directsimeq v_2'$ we
have $w_1' \not\directsimeq w_2'$. Since $w_1'$ and $w_2'$ are minimal
we have $|\class{\minsucc{w}}{\directsimeq}| \ge 2$. Contradiction.

\end{itemize}
Hence, $\getplayername'_q(\C) = \getplayername_q(f(\C))$.

Finally, we prove that $\C \step'_q \C'$ iff  $f(\C) \step_q f(\C')$.
Suppose $\C \step'_q \C'$ but not $f(\C) \step_q f(\C')$.  Assume
$\getplayername_q'(\C) = \getplayername_q(f(\C)) = \even$. The case
where $\getplayername_q'(\C) = \getplayername_q(f(\C)) = \odd$ is
similar.  Since $\C \directsimeq f(\C)$, there must be some $\mathcal{D}$
such that $f(\C) \step_q \mathcal{D}$ and $\C' \directsim \mathcal{D}$.
But then also $\C \step_q' \C''$ and $\mathcal{D} \directsim \C''$
for some $\C''$. Then $\C' \directsim \C''$. Distinguish two cases:
\begin{itemize}
\item Case $\C' = \C''$. Then $f(\C') = f(\C'') = \mathcal{D}$, contradicting
our assumption that $f(\C) \nstep_q f(\C')$.

\item Case $\C' \neq \C''$. Then we have $\C \step_q' \C'$ and
$\C \step_q' \C''$ and $\C' \directsim \C''$. But this means that
vertices in $\C'$ are not maximal. Hence, 
$\game_q'$ does not have a transition $\C \step_q' \C'$. \qedhere
\end{itemize}

\end{proof}

\begin{corollary} The direct simulation equivalence quotient of $\game$
is a unique (up-to isomorphism) parity game direct 
simulation equivalent to $\game$.
\end{corollary}

\subsection{Governed Bisimulation and Governed Stuttering Bisimulation Quotients}

We first define the governed bisimilarity quotients and claim some elementary
results of these. We then define the governed stuttering bisimilarity quotient.
\begin{definition}[Governed bisimulation quotient]
\label{def:pg_gbisim_quotient}
The governed bisimulation quotient of $(V,\step,\priosym,\getplayername)$
is the structure
$(\partition{V}{\gov}, \to', \priosym', \getplayername')$, where, for
$\C,\C' \in \partition{V}{\gov}$:

\begin{itemize}
\item $\priosym'(\C) = \min\{\prio{v} ~|~ v \in \C \}$,
\item $\getplayername'(\C) = \begin{cases}
   \ensuremath{\odd} & \text{if $\C \subseteq V_{\odd}$ and for all $u \in \C$,
 $|\class{\post{u}}{\gov}| > 1$} \\
   \ensuremath{\even} & \text{otherwise}
   \end{cases} $
  \item $\C \to' \C'$ if and only if 
    $\forall{v \in \C} \exists{v' \in \post{v}}{v' \in \C'}$ 

\end{itemize}
\end{definition}
\begin{theorem} Let $\game= (V, \step, \priosym,\getplayername)$ be a parity game
and $\game_q = (\partition{V}{\gov}, \step',\priosym',\getplayername')$ be
its governed bisimulation quotient. Then $\game \gov \game'$.
\end{theorem}
\begin{proof}
Follows from the fact that the relation 
$R  = \{(v,\C), (\C,v)~|~ v \in \C\}$, is a governed bisimulation relation.\qedhere\end{proof}
\begin{theorem}
Let $\game,\game'$ be two parity games and let $\game_q$ and $\game_q'$ be
their direct simulation equivalence quotients, respectively. Then 
$\game \gov \game'$ iff
the two structures $\game_q = (\partition{V}{\gov}, \step_q,
\priosym_q, \getplayername_q)$ and $\game_q' = (\partition{V'}{\gov}, 
\step'_q, \priosym'_q, \getplayername'_q)$ are isomorphic.
\end{theorem}
\begin{proof}
Similar to the proof of Theorem~\ref{thm:unique_quotient_dsim}.\qedhere\end{proof}
\begin{corollary} The governed bisimulation quotient of $\game$
is a unique (up-to isomorphism) parity game that is governed 
bisimilar to $\game$.
\end{corollary}
We next define the governed stuttering bisimulation quotient. It requires some 
subtlety to properly deal with divergences and ensure that
a unique player is assigned to an equivalence class.

%
%
%

\begin{definition}[Governed stuttering bisimulation quotient]
\label{def:pg_gstut_quotient}
The governed stuttering bisimulation quotient of $(V,\step,\priosym,\getplayername)$
is the structure
$(\partition{V}{\gstut}, \to', \priosym', \getplayername')$, where, for
$\C,\C' \in \partition{V}{\gstut}$:

\begin{itemize}
\item $\priosym'(\C) = \min\{\prio{v} ~|~ v \in \C \}$,
\item $\getplayername'(\C) = \begin{cases}
  {\even} & \text{if for all $v\in \C$, $v \diverges{\even}{\gstut}$, or
   for some $v \in \C, \C' \neq \C$, $v \steps{\even} \C'$ } \\
   {\odd} & \text{otherwise}
  \end{cases}$
 
\item $\C \to' \C'$ if and only if $\begin{cases}
    \exists{\player \in \{\even,\odd\}} \forall{v \in \C} v \diverges{\player}{\gstut} & \text{ if  $\C = \C'$} \\
    \exists{\player \in \{\even,\odd\}} \forall{v \in \C} v \forces{\player}{\gstut} \C' & \text{ if $\C \neq \C'$} \\
    \end{cases} $

\end{itemize}
\end{definition}

\begin{theorem} Let $\game= (V, \step, \priosym,\getplayername)$ be a parity game
and $\game_q = (\partition{V}{\gstut}, \step',\priosym',\getplayername')$ be
its governed stuttering bisimulation quotient. Then $\game \gstut \game_q$.
\end{theorem}
\begin{proof}
Consider the 
relation $R \subseteq (V \cup \partition{V}{\gstut})
\times (V \cup \partition{V}{\gstut})$, defined as follows:
\[R  = \{(v,\C), (\C,v), (v,w), (\C,\C)~|~ v,w \in \C\}
\]
Then $R$ is a governed stuttering bisimulation 
relation. Note that $R$ is an equivalence relation. It thus suffices
to prove that $R$ meets the remaining conditions of governed stuttering 
bisimulation. We do so by proving that $R$ is a governed stuttering bisimulation
for the following cases: $v \R w$, $\C \R \C'$, $v \R \C$ and $\C \R v$.
Observe that for $v,w \in V$ we have $v \R w$ iff
in $\game$ we have $v \gstut w$, and for $\C,\C' \in \partition{V}{\gstut}$
we have $\C \R \C'$ iff $\C = \C'$. As a result, for these cases $R$ is
immediately a governed stuttering bisimulation relation. We therefore focus
on the cases $v \R \C$ and $\C \R v$. Both cases are addressed separately.\\

\noindent
Suppose that $v \R \C$. We reason as follows:
\begin{itemize}
\item By definition, $\prio{v} = \priosym'(\C)$ as all $w \in \C$ are
such that $\prio{w} = \prio{v}$.

\item Suppose $v \step \C'$ for some $\C' \in \partition{V}{R}
\setminus \{\class{v}{R}\}$. Then by definition of $\gstut$ we 
have $w \forces{\getplayer{v}}{\gstut} \C'$ for all $w \in \C$ and therefore
$\C \step' \C'$.
\begin{itemize}
\item  Case $\getplayer{v} =\even$. Since for all
$w \in \C$ we have $w \forces{\even}{\gstut} \C'$, there is some $w \in \C$ such that
$w \steps{\even} \C'$ and therefore 
$\getplayername'(\C) = \even$. Because $\C \step' \C'$, also $\C \forces{\even}{R} \C'$.

\item
Case $\getplayer{v} = \odd$.  Suppose $\getplayername'(\C) = \odd$. Since $\C \step' \C'$, also
$\C \forces{\odd}{R} \C'$. 
Next, suppose $\getplayername'(\C) = \even$. 
Then there
must be some $w \in \C$ and some $\C'' \neq \C$ such that $w \steps{\even} \C''$, as
$w \diverges{\even}{\gstut}$ would conflict with $w \forces{\odd}{\gstut} \C'$.
Let $w$ be
such.  We also have $w \forces{\odd}{\gstut} \C'$. This can only be the case if
$\getplayer{w} = \odd$ and $\post{w} \subseteq \C'$.
But then $\C'$ is the only successor of $\C$, \ie 
$\C \step' \C''$ implies
$\C'' = \C'$, and therefore $\C \forces{\odd}{R} \C'$.

\end{itemize}

\item Suppose $v \diverges{\player}{R}$. Then $w \diverges{\player}{R}$
for all $w \in \C$. But then also $\C \step'
\C$. 
\begin{itemize}
\item Case $\player = \even$. Then also $\getplayername'(\C) = \even$ and
since $\C \step' \C$, also $\C \diverges{\player}{R}$.

\item Case $\player = \odd$. If $\getplayername'(\C) = \odd$, then, since
$\C \step' \C$, also $\C \diverges{\player}{R}$.
If $\getplayername'(\C) = \even$, then also $v \diverges{\even}{\gstut}$, since
$u \steps{\even} \C'$ for some $u \in \C$ and $\C' \neq \C$ contradicts
$u \diverges{\odd}{\gstut}$. But $v \diverges{\even}{\gstut}$ and $v \diverges{\odd}{\gstut}$
implies that for all $\C'$ such that $\C \step' \C'$, we have $\C' = \C$. Hence,
also $\C \diverges{\player}{R}$.

\end{itemize}

\end{itemize}

\noindent
Assume that $\C \R v$. We now reason as follows:
\begin{itemize}
\item $\priosym'(\C) = \prio{v}$ follows from the same arguments as before.

\item Suppose $\C \step' \C'$ for some $\C' \neq \C$. Then there is some $\player \in \{ \even, \odd \}$ such that for all $w \in \C$, we have $w \forces{\player}{\gstut} \C'$. We distinguish two cases.
\begin{itemize}
  \item Assume $\getplayername'(\C) = \even$. 
  \begin{itemize}
    \item Case $\player = \even$. Then $v \forces{\even}{\gstut}$ follows immediately since $v \in \C$.
    \item Case $\player = \odd$. Hence, for all $w \in \C$, we have $w \forces{\odd}{\gstut} \C'$. Since $\getplayername'(\C) = \even$, there must be some $w \in \C$, $\C'' \neq \C$, such that $w \steps{\even} \C''$, since $w \diverges{\even}{\gstut}$ contradicts $w \forces{\odd}{\gstut} \C'$.

Let $w$ be such, and observe that $w \forces{\even}{\gstut} \C''$. This can only be the case if $\C'' = \C'$, and from $w \forces{\even}{\gstut} \C''$ and $\C'' = \C'$ we obtain $v \forces{\even}{\gstut} \C'$
  \end{itemize}

  \item Assume $\getplayername'(\C) = \odd$. Observe that this implies that $\player = \odd$, since, if $\player = \even$, then for all $w \in \C$, we have $w \forces{\even}{\gstut} \C'$, which means there is some $w \in \C$ for which $w \steps{\even} \C'$. This contradicts $\getplayername'(\C) = \odd$. So, $\player = \odd$. It then immediately follows that for all $w \in \C$, we have $w \forces{\odd}{\gstut} \C'$, and in particular $v \forces{\odd}{\gstut} \C'$.
  
  \end{itemize}

\item Assume that $\C \diverges{\player}{R}$. We distinguish two cases.
\begin{itemize}
\item Case $\player = \even$.  Observe that $\getplayername'(\C) = \even$. Suppose that $\getplayername'(\C) = \odd$,
then $\C \diverges{\even}{R}$ implies that for all $\C'$ such that $\C \step' \C'$ we have $\C' = \C$ and therefore,
for all $w \in \C$ we have both $w \diverges{\odd}{\gstut}$ and $w \diverges{\even}{\gstut}$. This contradicts $\getplayername'(\C) = \odd$, hence $\getplayername'(\C) = \even$.

Now, towards a contradiction, assume that $w \ndiverges{\even}{\gstut}$ for all $w \in \C$.
Because $\getplayername'(\C) = \even$, there must be some $w \in \C$ such
that $w \steps{\even} \C'$ for some $\C' \neq \C$. Let $w$ be such. 
Since $w \ndiverges{\even}{\gstut}$ but $\C \step' \C$, we must have
$w \diverges{\odd}{\gstut}$. This contradicts $w \steps{\even} \C'$. So there must be some $w \in \C$ such that
$w \diverges{\even}{\gstut}$. But then also $v \diverges{\even}{\gstut}$ and therefore
$v \diverges{\even}{R}$.

\item Case $\player = \odd$.  Suppose $\getplayername'(\C') = \even$. 
Because $\C \diverges{\odd}{R}$ we find that $\C \not\step' \C'$ for
$\C' \not= \C$. Hence, for all $w \in \C$ we have $w \diverges{\even}{\gstut}$ and 
$w \diverges{\odd}{\gstut}$; in particular,
$v \diverges{\odd}{R}$.

Next, assume that $\getplayername'(\C') = \odd$. Then there must be
some $w \in \C$ such that $w \ndiverges{\even}{\gstut}$. Consequently $w \ndiverges{\even}{\gstut}$
for all $w \in \C$. Then $\C \step' \C$ can only be because for all $w \in \C$ we have
$w \diverges{\odd}{\gstut}$. In particular, $v \diverges{\odd}{\gstut}$ and therefore
$v \diverges{\odd}{R}$. \qedhere
\end{itemize}

\end{itemize}

\end{proof}
\begin{theorem}
  Let $\game,\game'$ be two parity games and let $\game_q$ and $\game_q'$ be
  their direct simulation equivalence quotients, respectively. Then 
  $\game \gstut \game'$ iff
  the two structures $\game_q = (\partition{V}{\gstut}, \step_q,
  \priosym_q, \getplayername_q)$ and $\game_q' = (\partition{V'}{\gstut}, 
  \step'_q, \priosym'_q, \getplayername'_q)$ are isomorphic.
\end{theorem}
\begin{proof}
Again similar to the proof of Theorem~\ref{thm:unique_quotient_dsim}.\qedhere\end{proof}
\begin{corollary} The governed stuttering bisimulation quotient of $\game$
is a unique (up-to isomorphism) parity game that is governed stuttering
bisimilar to $\game$.
\end{corollary}

\section{A Comparison of Discriminating Power}
\label{sec:comparison}

In this section, we compare the discriminative power of each of the
equivalences discussed in the preceding sections, essentially
justifying the lattice we illustrated in Section~\ref{sec:lattice}.
This permits us to assess the reductive power of
each of the studied equivalences that admit (unique) quotienting.
For each of the equivalences we show which other equivalences it strictly refines.
Incomparability results are described separately.

We first focus on proving the right-hand side of the lattice we
presented in Section~\ref{sec:lattice}. That is, we first compare isomorphism, strong bisimilarity and governed bisimilarity and then
focus on the various simulation equivalences.
\begin{theorem}\label{thm:iso_refines_strong}
Isomorphism is strictly finer than strong bisimilarity.
\end{theorem}
\begin{proof}
Clearly, every pair of isomorphic parity games is a pair of strong bisimilar
parity games. Strictness follows from a standard example:
\begin{center}
\begin{tikzpicture}
  \node[even] (n1) {0};
  \node[even, right of=n1] (n2) {0};
  \node[even, right of=n2, xshift=30pt] (n3) {0};
  
  \path[->] (n1) edge[bend left] (n2)
            (n2) edge[bend left] (n1)
            (n3) edge[loop right] (n3);
\end{tikzpicture}
\end{center}
Clearly, both vertices in the left parity game are strongly bisimilar to
the vertex in the right parity game, and \emph{vice versa}. However,
these vertices are not isomorphic.\qedhere
\end{proof}

\begin{figure}
\begin{center}
\begin{tikzpicture}
  
  \node[odd,label=below:{\scriptsize $v_1$}] (v1) {\scriptsize 0};
  \node[odd, left of=v1,label=below:{\scriptsize $v_2$}] (v2) {\scriptsize 1};
  \node[even, right of=v1,label=below:{\scriptsize $v_3$}] (v3) {\scriptsize 1};
  
  \path[->] (v1) edge (v2)
            (v1) edge (v3)
            (v2) edge[loop left]  (v2)
            (v3) edge[loop right] (v3)
;
\end{tikzpicture}
\end{center}
\caption{Parity game which is minimal with respect to strong bisimilarity. Vertices $v_2$ and $v_3$ are governed bisimilar.}  
\label{fig:strictness_strong_governed}
\end{figure}
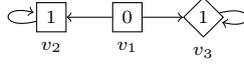

The following theorem relates strong bisimilarity to stuttering equivalence, governed bisimulation and strong direct simulation equivalence, and except for the comparison to governed bisimulation it is essentially the counterpart of the classical theorems in the setting of Kripke structures. 
\begin{theorem}
\label{thm:strong_refines}
Strong bisimilarity is strictly finer than strong direct simulation equivalence,  stuttering bisimulation equivalence and governed bisimulation equivalence.
\end{theorem}
\begin{proof}
We sketch each of the refinements:
\begin{itemize}
  \item Every strong bisimulation relation is a direct simulation relation.
Since such a relation is symmetric, every pair of parity games related
via strong bisimilarity is also related via strong direct simulation equivalence.
Strictness follows from the parity game in Figure~\ref{fig:govbisim_example}, in which $v_0$ and $v_1$ are strong direct simulation equivalent but not strongly bisimilar.
  \item Every strong bisimulation relation is a stuttering bisimulation relation, this follows directly from the definitiions. Strictness follows from the parity game in Figure~\ref{fig:fake-divergence}, in which $v_3$ and $v_4$ are stuttering bisimilar, but not strongly bisimilar.
  \item Every strong bisimulation relation is a governed bisimulation relation, this follows directly from the definitions. Strictness follows from the parity game in Figure~\ref{fig:strictness_strong_governed}, which is minimal modulo strong bisimilarity, but vertices $v_2$ and $v_3$ are governed bisimilar.
  \qedhere\end{itemize}
\end{proof}
We next state, without proof, a result that essentially follows by definition.
\begin{theorem}
\label{thm:strong_direct_refines}
Strong direct simulation equivalence strictly refines direct simulation equivalence.
\end{theorem}
The following refinement results follow a line of reasoning similar to the ones seen before.
\begin{theorem}\label{thm:governed_bisim_refines}
Governed bisimulation equivalence strictly refines direct simulation equivalence and
governed stuttering bisimulation equivalence.
\end{theorem}
\begin{proof}
We again sketch both refinements.
\begin{itemize}
  \item Refinement follows directly from the observation that a governed bisimulation is a symmetric direct simulation. Strictness follows from examples similar to those discriminating strong bisimilarity and simulation equivalence.
  \item It follows from the definitions that every governed bisimulation is also a governed stuttering bisimulation. The strictness of the refinement follows from the 
example in Figure~\ref{fig:fake-divergence} in which all vertices with priority $0$ are governed stuttering bisimilar, but none are governed bisimilar.
\qedhere\end{itemize}

\end{proof}

\begin{theorem}\label{thm:governed_bisim_incomparable_strong_direct}
Governed bisimulation equivalence and strong direct simulation equivalence are incomparable.
\end{theorem}
\begin{proof}
This follows from the parity game in Figure~\ref{fig:govbisim_example} in which vertices $v_0$ and $v_6$ are governed bisimilar, but not strong direct simulation equivalent. Furthermore, $v_0$ and $v_1$ are strong direct simulation equivalent, but not governed bisimilar.\qedhere
\end{proof}

To the best of our knowledge, the (elementary) result that direct simulation is strictly finer than (biased) delayed simulation has not been formally established. We therefore give a brief proof here.
\begin{theorem}
  \label{thm:direct_refines}
Direct simulation equivalences is strictly finer than even- and odd-biased delayed simulation equivalence.
\end{theorem}
\begin{proof}
  Let $\game = (V, \to, \priosym, \getplayername)$ be a parity game, with $v, w \in V$ such that $v \directsimeq w$. Observe that $v \directsim w$ and $w \directsim v$. We prove that also $v \delaysime w$; $w \delaysime v$ and both cases for $\delaysimo$ follow the exact same line of reasoning. As a result, we find that $v \delaysimeeq w$ and $v \delaysimoeq$.
  
  Observe that $\duplicator$ has a winning strategy in the direct simulation game from $(v, w)$. Therefore, in each round of the game, \duplicator was able to mimic \spoiler's move by a move to a vertex with equal priority. If \duplicator plays the same strategy in the delayed simulation game, this gives rise to plays with obligation $\checkmark$ in every configuration that is reached.
  
  Strictness follows straightforwardly from the observation that the (biased) delayed simulation relations can relate vertices with different priorities, whereas direct simulation cannot. \qedhere
  
\end{proof}
The results relating (biased) delay simulations to each other and to winner equivalence were proven by Fritz and Wilke \cite{FW:06}.
\begin{theorem}\label{thm:delay_refines}
Even- and odd-biased delay simulation are incomparable, and both are strictly finer than delay simulation equivalence. Delay simulation equivalence in turn is strictly finer than winner equivalence.
\end{theorem}
This completes the results underlying the right-hand side of the lattice we
presented in Section~\ref{sec:lattice}.\medskip

We next focus on the left-hand side of the lattice. 

\begin{theorem}
  \label{thm:stuttering_incomparability}
Stuttering bisimilarity is incomparable to governed bisimilarity, strong direct simulation, direct simulation and all delayed simulation variations.
\end{theorem}
\begin{proof}
In the parity game in Figure~\ref{fig:fake-divergence}, $v_3$ and $v_4$ are stuttering bisimilar but they cannot be related under governed bisimilarity, (strong) direct simulation equivalence, nor any of the delayed simulation equivalences. For the other direction, consider vertices $v_0$, $v_1$ and $v_6$ from the parity game in Figure~\ref{fig:govbisim_example}. None of these vertices are stuttering bisimilar, whereas $v_0$ and $v_6$ are governed bisimilar, $v_0$ and $v_1$ are strong direct simulation equivalent, and all three are direct simulation equivalent, and therefore also delay simulation equivalent.
\qedhere
\end{proof}

\begin{theorem}
Governed stuttering bisimilarity is incomparable to strong direct simulation, direct simulation and all delayed simulation variations.
\end{theorem}
\begin{proof}
Along the same lines as the proof of Theorem~\ref{thm:stuttering_incomparability}. \qedhere
\end{proof}

\begin{theorem}\label{thm:stuttering_refines}
Stuttering bisimilarity strictly refines governed stuttering bisimilarity.
\end{theorem}
\begin{proof}
It follows from the definitions that every stuttering bisimulation is also a governed stuttering bisimulation. Strictness follows from the parity game in Figure~\ref{fig:fake-divergence} in which $v_0$ and $v_1$ are governed stuttering bisimilar, but not stuttering bisimilar. \qedhere
\end{proof}


To complete the lattice, we next move to showing that governed stuttering bisimilarity is strictly finer than winner equivalence.
%
In order to prove this result, we must first lift the concept of governed stuttering bisimilarity to paths.

Paths of length 1 are equivalent if the vertices they consist of are equivalent. If paths $p$ and $q$ are equivalent, then $p \pathconcat \path{v} \gstut q$ iff $v$ is equivalent to the last vertex in $q$, and $p \pathconcat \path{v} \gstut q \pathconcat \path{w}$ iff $v \gstut w$. An infinite path $p$ is equivalent to a path $q$ if for all finite prefixes of $p$ there is an equivalent prefix of $q$ and \emph{vice versa}.

\begin{lemma}
\label{lem:quotient-preserves-winner}
Let $(V, {\to}, \getplayername, \priosym)$ be a parity game, and let $(\partition{V}{\gstut}, {\to'}, \getplayername', \priosym')$ be its quotient. Let $v \in V$, and $\C \in \partition{V}{\gstut}$ such that $v \in \C$. For all players $\player$, and all $\strategyname \in \strategy{\player}$ there is some $\strategynamealt \in \memstrategy{\player}$ such that for all $q \in \paths[\omega]{\strategynamealt}{\C}$ there is a $p \in \paths[\omega]{\strategyname}{v}$ such that $p \gstut q$.
\end{lemma}
\begin{proof}
\def\allowsdiv#1{\mathsf{div}(#1)}
\def\next#1{\mathsf{next}(#1)}
Define an arbitrary complete ordering $\precdot$ on vertices, and define the following for finite paths $q$ starting in $\C$, where $\min_{\precdot} \emptyset$ is defined to be $\bot$:
\begin{align*}
\next{q} &= \min_\precdot \{ v' \in V \mid \exists{p \in \paths{\strategyname}{v}}{p \gstut q \land p \steps{\strategyname} v' \land p \pathconcat \path{v'} \not\gstut q} \} \\
\allowsdiv{q} &= \exists{p \in \paths[\omega]{\strategyname}{v}}{p \gstut q}
\end{align*}

We next show that it is possible to define a strategy 
$\strategynamealt \in \memstrategy{\player}$ for finite plays $q = \path{\C \ldots \C'}$ such that if $q \gstut p$ for some $p \in \paths{\strategyname}{v}$, then:
\[
\begin{cases}
\strategynamealt(q) = \C' & \text{ if } \allowsdiv{q} \text{ and } \C' \to' \C' \\
\strategynamealt(q) = \class{\next{q}}{\gstut} & \text{ otherwise. }
\end{cases}
\]

Let $p \in \paths{\strategyname}{v}$ be such that $q \gstut p$ for
$q = \path{\C \ldots \C'}$ and assume  $\getplayername'(\C') = \player$.
In case $\allowsdiv{q}$ and $\C' \to' \C'$, then obviously $\strategynamealt(q)$ 
can be defined as $\C'$. We proceed to show that if $\lnot \allowsdiv{q}$ or $\C' \not\to' \C'$, 
then 1) $\next{q} \neq \bot$, and 2) we can set $\strategynamealt(q) = \class{\next{q}}{\gstut}$.
We show the first by distinguishing two cases:
\begin{description}

\item[Case $\neg \allowsdiv{q}$,] then it follows straightforwardly that $\next{q} \ne \bot$.

\item[Case $\C' \not\to' \C'$.] Because $\C'$ is a vertex is a
quotient graph, $\C' \ndiverges{\player}{\gstut}$.  Consider the
path $p \in \paths{\strategyname}{v}$ for which $p \gstut q$, and
assume that $p$ is of the form $\bar{p} \pathconcat \path{u}$.
Since $\bar{p} \pathconcat \path{u} \gstut q$, also $u \gstut \C'$
and hence $u \ndiverges{\player}{\gstut}$.  Then by
Lemma~\ref{lem:force-vs-div}, $u \forces{\opponent{\player}}{\gstut}
V \setminus \class{u}{\gstut}$.  Let $\sigma' \in
\strategy{\opponent{\player}}$ be such that $u \forces{\sigma'}{\gstut}
V \setminus \class{u}{\gstut}$ and consider the unique path $r
\pathconcat \path{v'} \in \paths{\sigma'}{u}$ such that $\sigma
\allows r \pathconcat \path{v'}$, $\sigma' \allows r \pathconcat
\path{v'}$, $r \gstut u$ and $v' \in V \setminus \class{u}{\gstut}$.
Then $p \pathconcat r \in \paths{\strategyname}{v}$ is such that
$p \pathconcat r \gstut q$, $p\pathconcat r \steps{\sigma} v'$ and
$p \pathconcat r \pathconcat \path{v'} \not\gstut q$. Hence, $\next{q}
\neq \bot$.

\end{description}
Next, we show that we can set $\strategynamealt(q) = \class{\next{q}}{\gstut}$.
Since $\next{q} \ne \bot$, there must be some $\path{v
\ldots v' v''} \in \paths{\strategyname}{v}$ such that $v'' =
\next{q}$, $\path{v \ldots v'} \gstut q$ and $v' \not\gstut v''$.
Since $v' \gstut \C'$ and $v' \step \class{v''}{\gstut}$, also
$\C' \forces{\getplayer{v'}}{\gstut} \class{v''}{\gstut}$. As
$\C'$ is a vertex in a quotient graph, this implies
$\C' \steps{\getplayer{v'}} \class{v''}{\gstut}$. Hence, we can
set $\strategynamealt(q)  = \class{\next{q}}{\gstut}$.
\medskip

Now we have shown that it is always possible to define a strategy adhering to the restrictions above, let $\strategynamealt$ be such a strategy. We show using induction on $n$ that for all $n$
$$
\forall{q \in \paths[n]{\strategynamealt}{\C}}{\exists{p \in \paths{\strategyname}{v}}{p \gstut q}}.
$$
For $n=0$, this is trivial, because $v \gstut \C$. For $n = m + 1$, assume as the induction hypothesis that $\forall{\bar{q} \in \paths[m]{\strategynamealt}{\C}}{\exists{\bar{p} \in \paths{\strategyname}{v}}{\bar{p} \gstut \bar{q}}}$. Let $q \in \paths[n]{\strategynamealt}{\C}$ and let $\C',\C'' \in \partition{V}{\gstut}$ and $\bar{q} \in \paths[m]{\strategynamealt}{\C}$ such that $\bar{q} = \path{\C \ldots \C'}$ and $q = \bar{q} \pathconcat \path{\C''}$. Distinguish cases on the player who owns $\C'$.
\begin{description}
\item[Case $\getplayername'(\C') = \player.$]
Then $\C'' = \strategynamealt(\bar{q})$. The induction hypothesis yields some $\bar{p} \in \paths{\strategyname}{v}$ such that $\bar{p} \gstut \bar{q}$, therefore $\C'' = \C'$ if $\allowsdiv{\bar{q}}$ and $\C' \step' \C'$, and otherwise $\C'' = \class{\next{\bar{q}}}{\gstut}$.

If $\C' = \C''$, then $\allowsdiv{\bar{q}}$, so there must be some $p \in \paths[\omega]{\strategyname}{v}$ such that $p \gstut q$ and therefore also some $p \in \paths{\strategyname}{v}$ such that $p \gstut q$. 

If $\C'' = \class{\next{\bar{q}}}{\gstut}$, there must be some $p \in \paths{\strategyname}{v}$ such that $p = p' \pathconcat \path{v'}$ and $p' \gstut \bar{q}$ and $v' \gstut \C''$. By definition, $p \gstut q$ for such $p$.

\item[Case $\getplayername'(\C') \ne \player.$] From the induction
hypothesis, obtain a $\bar{p} \in \paths{\strategyname}{v}$ such
that $\bar{p} \gstut \bar{q}$. Without loss of generality we may
assume that $\bar{p}$ is finite.    Note that $\C' \step
\class{\C''}{\gstut}$. We distinguish two cases.
\begin{itemize}
\item Case $\C' = \C''$. Then we have $\bar{p} \gstut \bar{q} \C''$.

\item Case $\C' \neq \C''$. Let $v'$ be the last vertex in
$\bar{p}$.  Because $\bar{p} \gstut \bar{q}$, also
$v' \forces{\opponent{\player}}{\gstut} \class{\C''}{\gstut}$.  So
let $\strategyname' \in \strategy{\opponent{\player}}$ be such that
$v' \forces{\strategyname'}{\gstut} \class{\C''}{\gstut}$. Now
consider an infinite path $\bar{p} \pathconcat p$ such that
$\strategyname \allows \bar{p} \pathconcat p$ and $\strategyname'
\allows \bar{p} \pathconcat p$. For some index $k \ge 0$, it must
be the case that $p_k \gstut \C''$ and $p_l \gstut \C'$ for all $l
< k$. So $\bar{p} \pathconcat \path{p_0 \ldots p_k} \gstut q$.
\end{itemize}

\end{description}
Finally, we prove that for all $q \in \paths[\omega]{\strategynamealt}{\C}$
there is a $p \in \paths[\omega]{\strategyname}{v}$ such that $p
\gstut q$. Let $q \in \paths[\omega]{\strategynamealt}{\C}$. Then by
the above, we find that there is some $p \in \paths{\strategyname}{v}$.
Suppose $p$ is finite and $p = \bar{p} v'$ for some vertex $v'$.
Since $q$ is a path through the quotient graph,
$q$ must be of the form $\bar{q} \C^\omega$ for some $\C \gstut v'$. 
\begin{itemize}
 \item Case $\getplayername'(\C) = \player$. Then $\strategynamealt(\bar{q}
 \C) = \C$, and thus $\allowsdiv{\bar{q} \C}$ by definition of
 $\strategynamealt$.  But then there must be some $p' \in
 \paths[\omega]{\strategyname}{v}$ such that $p' \gstut \bar{q} \C
 \gstut q$.

 \item Case $\getplayername'(\C) \neq \player$. Since $\C \step' \C$
we have $\C \diverges{\opponent{\player}}{\gstut}$ and since $v'
\gstut \C$, also $v' \diverges{\opponent{\player}}{\gstut}$. Let
$\strategyname' \in \strategy{\opponent{\player}}$ be such that $v'
\diverges{\opponent{\player}}{\gstut}$. Then there is an infinite
path $p' \in \paths[\omega]{\strategyname'}{v'}$ such that $\strategyname
\allows p'$ and $p' \gstut v'$.  But then $\strategyname \bar{p}
p'$ and $q \gstut \bar{p} p'$.\qedhere
\end{itemize}

\end{proof}

\begin{theorem}\label{thm:gstut_refines}
Governed stuttering bisimularity strictly refines winner equivalence.
\end{theorem}
\begin{proof}
Let $\game = (V, {\to}, \getplayername, \priosym)$ be a parity game, and let $v, w \in V$ such that $v \gstut w$. Let $(\partition{V}{\gstut}, {\to'}, \getplayername', \priosym')$ be the governed stuttering bisimulation quotient of $\game$, and let $\C \in \partition{V}{\gstut}$ be such that $w \gstut \C\mkern1mu$. By transitivity of $\gstut\mkern2mu$, also $v \gstut \C\mkern1mu$. Now suppose that player $\player$ has a winning strategy $\strategyname$ from $v$. Then by Lemma~\ref{lem:quotient-preserves-winner}, $\player$ has a strategy $\strategynamealt$ from $\C$ such that for every play $q \in \paths{\strategynamealt}{\C}$ there is a play $p \in \paths{\strategyname}{v}$ such that $p \gstut q$. Because the priorities occurring infinitely often on such $p$ and $q$ are the same, $\strategynamealt$ is also winning for $\player$. If $\opponent{\player}$ had a winning strategy $\strategyname'$ from $w$, then we could repeat this argument to construct a winning strategy for $\opponent{\player}$ from $\C$, but this would be contrary to the fact that parity games are determined. Therefore, $w$ must also be won by player $\player$. \qedhere\end{proof}

\section{Conclusion}
\label{sec:conclusions}

Preorders and equivalences for parity games have been studied on a
number of occasions, see~\cite{Jan:05,CKW:11,CKW:12,Cra:15,Kei:13, FW:06, Fri:05}. A major motivation for some
of these is that they provide the prospect of simplifying games
prior to solving them. In this paper, we reconsidered several of
the parity game relations previously defined by us, \viz (governed)
bisimulation and (governed) stuttering bisimulation.  More specifically,
we gave detailed proofs showing that our relations are equivalences,
they have unique parity game quotients and they approximate the
winning regions of parity games. Furthermore, we showed that our
coinductively defined equivalence relations admit game-based
definitions; the latter facilitated the comparison of our equivalences
to the game-based definitions of relations for parity games found
in the literature.  For the latter relations, we additionally gave
coinductive definitions. Finally, we showed that, unlike \eg delayed
simulation or any of its biased versions, our equivalence relations
give rise to unique quotients.

There are several natural continuations of this research. First,
the experiments that were conducted in~\cite{CKW:12,Kei:13} showed
that  parity games that could not be solved become solvable by
preprocessing the games using an $\mathcal{O}(m  n)$  stuttering
bisimulation minimisation algorithm or an $\mathcal{O}(m 
n^2)$ governed stuttering bisimulation minimisation algorithm; the
overall gain in speed otherwise was not significant. It would be
worthwhile to establish whether this is still true when using the
$\mathcal{O}(m \log{n})$ stuttering equivalence minimisation algorithm
of~\cite{GrW:16}.  Moreover, it would be interesting to see whether
the $\mathcal{O}(m n^2)$ time complexity of governed stuttering
bisimulation can be reduced using ideas from~\cite{GrW:16}. Similarly,
we believe that our coinductive rephrasing of delayed simulation
will help to devise a more efficient algorithm for computing it,
using a partition refinement approach.

Finally, an interesting line of investigation is to see whether the
incomparable notions of governed stuttering bisimulation and delayed
simulation equivalence can be married. Given that we have established
game-based and coinductive definitions for both relations, defining such a
relation now seems within reach. The resulting relation would be
closer to winning equivalence and perhaps even shed light on ways
to efficiently solve parity games in general.

\bibliographystyle{plain}
\bibliography{literature}

\appendix
\section{Detailed proofs of Propositions~\ref{prop:soundness_gstut} and~\ref{prop:completeness_gstut}}

Before we address Propositions~\ref{prop:soundness_gstut}
and~\ref{prop:completeness_gstut}, we first repeat the definition
of the variant function we will use in the proof of
Proposition~\ref{prop:soundness_gstut} and we state three lemmata
that characterise properties of this variant function.

\begin{definition}[Governed stuttering bisimulation game measure]
We define a measure with respect to $\gstut$ for a configuration
$((u_0,u_1),c)$ in the governed stuttering bisimulation game as follows:
\begin{equation*}
\begin{array}{l}
\measure{u_0, u_1, c}  \isdef\\
\begin{cases}
  (0,0) & \text{if } c = \checkmark \\
  (\infty,0) & \text{if } c = \dagger \land {\exists{v_0 \in \post{u_0}, v_1 \in
\post{u_1}}{v_0 \gstut u_0 \wedge v_1 \gstut u_1}} \\

  (\front{u_j,u_{1-j}}, 0) & \text{if } 
    c = \dagger \land \forall{v \in \post{u_j}}{u_j \not \gstut v} \\

  (0, \exit{u_j,u_{1-j},t} ) & \text{if } 
    c = (j,t) \\
\end{cases}
\end{array}
\end{equation*}
where $\front{u_0,u_1}$ denotes the number of steps before $\getplayer{u_0}$'s opponent
is forced from $\class{u_0}{\gstut}$ and $\exit{u_0,u_1,u_2}$ denotes the
number of steps it takes for $\getplayer{u_0}$ to force play to $\class{u_2}{\gstut}$
from $u_1$.  Formally, we have:
\begin{equation*}
\begin{array}{l}
\front{u_0,u_1} \isdef \dist{\opponent{\getplayer{u_0}}}{\class{u_0}{\gstut}}{u_{1}}{V \setminus \class{u_0}{\gstut}}\\
\exit{u_0,u_1,u_2} \isdef
\dist{\getplayer{u_0}}{\class{u_0}{\gstut}}{u_{1}}{\class{u_2}{\gstut}}
\end{array}
\end{equation*}
where for $U, T \subseteq V$ and $v \in U$:
\begin{equation*}
  \dist{\player}{U}{v}{T} \isdef
  \begin{cases}
  \min\{ n \mid v \in \battr[n]{U}{\player}{T} \} & \text{if } v \forces{\player}{U} T\\
  \infty & \text{otherwise}
  \end{cases}
\end{equation*}
Measures are ordered lexicographically, \ie $(m_0, m_1) < (n_0, n_1)$ iff
$m_0 < n_0 \lor (m_0 = n_0 \land m_1 < n_1)$.
\end{definition}
We first prove some basic properties for the function $\measuresym$. 
\begin{lemma}\label{lem:star_positive}
For $u,v \in V$, $(u \gstut v \land \measure{u,v,\dagger} = (m_0,m_1)) \implies m_0 > 0$.
\end{lemma}
\begin{proof}
First, observe that (apart from $c = \dagger$), the conditions in the second and third clause of the definition of $\measuresym$ are complementary.
Furthermore observe that, for all $u_0, u_1$ such that $u_0 \gstut u_1$, we have $\front{u_0, u_1} > 0$ since $u_1 \not \in V \setminus \class{u_0}{\gstut}$. The result then immediately follows. \qedhere\end{proof}

\begin{lemma}\label{lem:decrease_dist}
  Let $U,T \subseteq V$, such that $U \cap T = \emptyset$ and let $u \in U$. 
  For all players $\player$, if $u \forces{\player}{U} T$ then
  $\dist{\player}{U}{u}{T} >
  \min\{ \dist{\player}{U}{v}{T} \mid u \to v \land v \in U \cup T\}$.
\end{lemma}
\begin{proof}
  Assume $u \forces{\player}{U} T$ and let 
  $n = \dist{\player}{U}{u}{T}$.
  Hence $u \in \battr[n]{U}{\player}{T}$
  and $u \not \in \battr[n-1]{U}{\player}{T}$. Observe that $n > 0$ since
  $u \in U$ and $U \cap T = \emptyset$.
  We proceed by a case distinction on $\getplayer{u}$.
  \begin{itemize}
    \item $\getplayer{u} = \player$. Since $n$ is such that $u \notin \battr[n-1]{U}{\player}{T}$, we have $\exists{v \in \post{u}}{v \in \battr[n-1]{U}{\player}{T}}$. Let
    $v$ be such, and observe that $\dist{\player}{U}{v}{T} \leq n - 1 < n$.
    The result then follows immediately.
    \item $\getplayer{u} \neq \player$. As $n$ is such that $u \notin \battr[n-1]{U}{\player}{T}$, we have $\forall{v \in \post{u}}{v \in \battr[n-1]{U}{\player}{T}}$. 
    Hence
    $\forall{v \in V}{u \to v \implies \dist{\player}{U}{v}{T} \leq n - 1 < n}$.
    Again the result follows immediately. \qedhere
  \end{itemize}
\end{proof}
We also prove the following stronger result in case $u$ is
owned by the opponent.
\begin{lemma}\label{lem:decrease_dist_opponent}
  Let $U,T \subseteq V$, such that $U \cap T = \emptyset$ and $u \in U \cap V_{\opponent{\player}}$. 
  Then $u \forces{\player}{U} T$
  implies
  $\dist{\player}{U}{u}{T} >
  \max\{ \dist{\player}{U}{v}{T} \mid u \to v \land v \in U \cup T \}$.
\end{lemma}
\begin{proof}
  Let $u \in U\cap V_{\opponent{\player}}$ such that $u \forces{\player}{U} T$.
  Suppose $n = \dist{\player}{U}{u}{T}$. Then
  $u \in \battr[n]{U}{\player}{T}$ and $u \not \in \battr[n-1]{U}{\player}{T}$.
  Since  $u \notin V_{\player}$,
  $\forall{v \in \post{u}}{v \in \battr[n-1]{U}{\player}{T}}$, hence
  $\forall{v \in \post{u}}{\dist{\player}{U}{v}{T} \leq n - 1 < n}$;
  furthermore, such $v$ are in $U \cup T$,
  and the result follows immediately. \qedhere
\end{proof}

\begin{proposition}
   For all $v,w \in V$ if $v \gstut w$ then $v \gstutg w$.
\end{proposition}
\begin{proof}
  We prove for all governed stuttering bisimilar vertices $v \gstut w$ that 
  there is a \duplicator winning strategy in the governed stuttering 
  bisimulation game from configuration $((v,w),\checkmark)$.
  
  We show this by constructing a \duplicator-strategy that moves
  between governed stuttering bisimilar vertices, and that makes
  sure that from every configuration $((v,w),c)$, within a finite
  number of steps another configuration $((v',w'),\checkmark)$ is reached.
  As a consequence, the \duplicator-strategy is such that it passes
  through configurations with reward $\checkmark$ infinitely often,
  hence the \duplicator strategy is winning.

  Formally, we preserve the invariant $\Phi$ which is the conjunction of the following
  for configurations $((u_0,u_1),c)$:
  \begin{itemize}
    \item $u_0 \gstut u_1$,
    \item $c = (j,t)$ implies $u_j \not\gstut t$,
    \item $c = (0,u)$ implies $(u_0,u_1) \in V_{\even} \times V$,
    \item $c = (1,u)$ implies $(u_0,u_1) \in V \times V_{\odd}$.
  \end{itemize}

  In addition, we prove that from every configuration $((u_0, u_1),c)$, a
  configuration $((u_0',u_1'),\checkmark)$ is reached within a
  finite number of steps by showing that $\measuresym$
  is a variant function. That is, if, in a round, we move from configuration
  $((u_0,u_1),c)$ to configuration $((u_0',u_1'),c')$ with $c \neq
  \checkmark$ and $c' \neq \checkmark$, then $\measure{u_0,u_1,c}
  > \measure{u_0',u_1',c'}$. 

  From these two observations, the result immediately follows.  Note
  that initially we are in a configuration $((v,w),\checkmark)$;
  hence $\Phi$ is satisfied trivially.  Suppose the game
  has reached a configuration $((u_0, u_1),c)$ satisfying $\Phi$.
  In step~1 of the round, \spoiler chooses to play from $(t_0,t_1)$,
  taken from $(u_0, u_1)$ or $(u_1,u_0)$.
  We remark that if \spoiler decides to play from $(u_1,u_0)$, then,
  regardless of step~2, any pending challenge or $\dagger$ will be
  replaced by a $\checkmark$ at the end of step~3. For this case, we
  therefore do not need to argue that $\measuresym$ decreases.

  We distinguish cases based on which player can force a divergence
  in the coinductive definition and consider \duplicator's options
  in step~2 and~3 of the round, and prove that \duplicator can
  always arrive in a new configuration $((t_0',t_1'),c')$ that satisfies $\Phi$ and
  for which, if $c' \neq \checkmark$ and $c \neq \checkmark$, 
  $\measure{t_0,t_1,c} > \measure{t_0',t_1',c'}$.

    \begin{itemize}
      \item $t_0 \diverges{\even}{\gstut}$ and $t_0 \diverges{\odd}{\gstut}$. This
      case is trivial, as in that case exactly one (reachable) equivalence class exists.

      \item $t_0 \diverges{\even}{\gstut}$ and $t_0 \ndiverges{\odd}{\gstut}$. 
      Since $t_0 \gstut t_1$ also $t_1 \diverges{\even}{\gstut}$. We distinguish cases based on the owners of the vertices.
      \begin{itemize}
        \item $\getplayer{t_0} = \getplayer{t_1} = \even$. \spoiler
	plays $t_0 \to w_0$.
        \begin{itemize}
        \item Case there is some $w_1 \in \post{t_1}$ such that
        $w_0 \gstut w_1$. Then \duplicator plays to
        such a $w_1$. The new configuration is $((w_0,w_1),\checkmark)$.

        \item Case there is no $w_1 \in \post{t_1}$ such that
        $w_0 \gstut w_1$. 
        Then \duplicator plays to a $w_1 \in \post{t_1}$ for which
        $w_1 \gstut t_1$ with minimal $\measure{t_0,w_1,(0,w_0)}$; 
        the existence of a $w_1 \gstut t_1$ follows from  $t_1 \diverges{\even}{\gstut}$.

        \emph{New configuration:} if
        $c \in \{\checkmark,\dagger,(0,w_0)\}$ and $u_0 = t_0$ then the new
        configuration is $((t_0,w_1),(0,w_0))$, and else
        $((t_0,w_1),\checkmark)$.

	\emph{Progress:} we demonstrate $\measure{t_0,w_1,(0,w_0)} < \measure{t_0,t_1,c}$
        for $c \in \{\dagger,(0,w_0)\}$.
	In case $c = \dagger$ this follows from 
        Lemma~\ref{lem:star_positive}.
        In case $c = (0,w_0)$, this follows from
	Lemmata~\ref{lem:decrease_dist} and~\ref{lem:decrease_dist_opponent}.

        \end{itemize}

	\item $\getplayer{t_0} = \getplayer{t_1} = \odd$. \spoiler
	plays $t_1 \to w_1$. Since $t_1
	\diverges{\even}{\gstut}$, all $w_1 \in \post{t_1}$
	satisfy $t_1 \gstut w_1$. The same holds for all $w_0 \in \post{t_0}$.
        \duplicator can thus play arbitrary $t_0 \to w_0$. 

        \emph{New configuration:} $((w_0, w_1),\checkmark)$

	\item $\getplayer{t_0} = \even$, $\getplayer{t_1} = \odd$.
	\spoiler plays $t_0 \to w_0$ and $t_1 \to w_1$.
	Since $t_1 \diverges{\even}{\gstut}$, also $w_1 \gstut t_1$.
        We distinguish two further cases.

	\begin{itemize}
	  \item Case $w_0 \gstut w_1$. 

          \emph{New configuration:} $((w_0,w_1),\checkmark)$.

	  \item Case $w_0 \not \gstut t_0$. 

          \emph{New configuration:} if $c \in
	  \{\checkmark,\dagger,(0,w_0)\}$ and $u_0 = t_0$ then
	  the new configuration is $((t_0, w_1),(0,w_0))$; else the new
	  configuration is $((t_0,w_1),\checkmark)$.  Observe
	  that $t_0 \gstut w_1$.

	\emph{Progress:} we demonstrate $\measure{t_0,w_1,(0,w_0)} < \measure{t_0,t_1,c}$
        for $c \in \{\dagger,(0,w_0)\}$.
	In case $c = \dagger$ this follows from 
        Lemma~\ref{lem:star_positive}.
        In case $c = (0,w_0)$, this follows from
	Lemmata~\ref{lem:decrease_dist} and~\ref{lem:decrease_dist_opponent}.

	\end{itemize}

	\item $\getplayer{t_0} = \odd$, $\getplayer{t_1} = \even$.
	\duplicator plays $t_1 \to w_1$ such that $w_1 \gstut t_1$ 
        and $t_0 \to w_0$. Such a $w_1$ exists because $t_1 \diverges{\even}{\gstut}$.

        \emph{New configuration:} $((w_0, w_1), \checkmark)$.

      \end{itemize}

    \item $t_0 \diverges{\odd}{\gstut}$ and $t_0 \ndiverges{\even}{\gstut}$. So, as before, $t_1
    \diverges{\odd}{\gstut}$.  This case is dual to the previous one.

    \item $t_0 \ndiverges{\even}{\gstut}$ and $t_0
    \ndiverges{\odd}{\gstut}$.  We consider the owners of $t_0$ and $t_1$.
    \begin{itemize}
    \item $\getplayer{t_0} = \getplayer{t_1} = \even$. \spoiler plays
    $t_0 \to w_0$. We distinguish two cases.

   \begin{itemize}
    \item Case there is some $w_1 \in \post{t_1}$ for which $w_0 \gstut w_1$.
     \duplicator plays $t_1 \to w_1$ such that $w_0 \gstut w_1$.

     \emph{New configuration:} $((w_0,w_1),\checkmark)$.

    \item Case there is no $w_1 \in \post{t_1}$ for which $w_0 \gstut w_1$. 
     \begin{itemize}
      \item Case $t_0 \gstut w_0$. Then for all $w_1 \in \post{t_1}$,
        $w_1 \not\gstut t_1$. 

        \emph{New configuration:} 
        $((w_0,t_1),\dagger)$ if $c \in \{\checkmark,\dagger\}$ and $u_0 = t_0$; 
        otherwise $((w_0,t_1),\checkmark)$.

	\emph{Progress:} we must show $\measure{t_0,w_1,\dagger} <
	\measure{t_0,t_1,\dagger}$. This follows from
	Lemma~\ref{lem:decrease_dist_opponent}.

      \item Case $t_0 \not\gstut w_0$. \duplicator plays $t_1 \to w_1$ such
        that $t_1 \gstut w_1$ and $\measure{t_0,w_0,(0,w_0)}$ is minimal.

        \emph{New configuration:} 
        if $c \in \{\checkmark,\dagger,(0,w_0)\}$ and $u_0 = t_0$ then 
        configuration $((t_0,w_1),(0,w_0))$ and else $((t_0,w_1),\checkmark)$.

        \emph{Progress:}  we must show $\measure{t_0,w_1,(0,w_0)} <
	\measure{t_0,t_1,\dagger}$ for $c \in \{\dagger,(0,w_0)\}$. 
	In case $c = \dagger$ this follows from 
        Lemma~\ref{lem:star_positive}.
        In case $c = (0,w_0)$, this follows from
	Lemmata~\ref{lem:decrease_dist} and~\ref{lem:decrease_dist_opponent}.

    \end{itemize}

   \end{itemize}

    \item $\getplayer{t_0} = \getplayer{t_1} = \odd$. \spoiler plays
    $t_1 \to w_1$.

   \begin{itemize}
    \item Case there is some $w_0 \in \post{t_0}$ for which $w_0 \gstut w_1$.
     \duplicator plays $t_0 \to w_0$ such that $w_0 \gstut w_1$.
     
     \emph{New configuration:} $((w_0,w_1),\checkmark)$.

    \item Case there is no $w_0 \in \post{t_0}$ for which $w_0 \gstut w_1$. 
     \begin{itemize}
      \item Case $t_1 \gstut w_1$. Then for all $w_0 \in \post{t_0}$,
        $w_0 \not\gstut t_0$. \duplicator plays some arbitrary $t_0 \to w_0$.

        \emph{New configuration:} 
        $((t_0,w_1),\dagger)$ if $c \in \{\checkmark,\dagger\}$ and $u_1 = t_1$; 
        otherwise $((t_0,w_1),\checkmark)$.

	\emph{Progress:} we must show $\measure{w_0,t_1,\dagger} <
	\measure{t_0,t_1,\dagger}$. This follows from
	Lemma~\ref{lem:decrease_dist_opponent}.

      \item Case $t_1 \not\gstut w_1$. \duplicator plays $t_0 \to w_0$ such
        that $t_0 \gstut w_0$ and $\measure{w_0,t_0,(1,w_1)}$ is minimal.

        \emph{New configuration:} 
        if $c \in \{\checkmark,\dagger,(1,w_1)\}$ and $u_1 = t_1$ then 
        configuration $((w_0,t_1),(1,w_1))$ and else $((w_0,t_1),\checkmark)$.

        \emph{Progress:}  we must show $\measure{w_0,t_1,(1,w_1)} <
	\measure{t_0,t_1,c}$ for $c \in \{\dagger,(1,w_1)\}$. 
	In case $c = \dagger$ this follows from 
        Lemma~\ref{lem:star_positive}.
        In case $c = (1,w_1)$, this follows from
	Lemmata~\ref{lem:decrease_dist} and~\ref{lem:decrease_dist_opponent}.

    \end{itemize}

   \end{itemize}

    \item $\getplayer{t_0} = \even, \getplayer{t_1} = \odd$.
          \spoiler plays $t_0 \to w_0$ and $t_1 \to w_1$. In case
          $w_0 \not\gstut w_1$ then either $t_0 \gstut w_0$ or
          $t_1 \gstut w_1$. We distinguish three cases:

          \begin{itemize}
          \item Case $w_0 \gstut w_1$. 

          \emph{New configuration:} $((w_0, w_1),\checkmark)$.

	  \item Case $w_0 \not \gstut w_1$ and $t_0 \gstut w_0$.

	  \emph{New configuration:} $((w_0, t_1), (1,w_1))$ if $c
	  \in \{\checkmark,\dagger,(1,w_1)\}$ and $u_1 = t_1$;
	  otherwise $((w_0,t_1),\checkmark)$.

          \emph{Progress:} we must show $\measure{w_0,t_1,(1,w_1)} 
          < \measure{t_0,t_1,c}$ for $c \in \{\dagger,(1,w_1)\}$.
          In case $c = \dagger$ this follows from Lemma~\ref{lem:star_positive}.
          In case $c = (1,w_1)$ this follows from 
          Lemmata~\ref{lem:decrease_dist} and~\ref{lem:decrease_dist_opponent}.

          \item Case $w_0 \not \gstut w_1$ and $t_1 \gstut w_1$.

          \emph{New configuration:} $((t_0,w_1),(0,w_0))$ if 
          $c \in \{\checkmark,\dagger,(0,w_0)\}$ and $u_0 = t_0$;  otherwise
          $((t_0,w_1),\checkmark)$.

          \emph{Progress:} we must show $\measure{t_0,w_1,(0,w_0)} 
          < \measure{t_0,t_1,c}$ for $c \in \{\dagger,(0,w_0)\}$.
          In case $c = \dagger$ this follows from Lemma~\ref{lem:star_positive}.
          In case $c = (0,w_0)$ this follows from 
          Lemmata~\ref{lem:decrease_dist} and~\ref{lem:decrease_dist_opponent}.

          \end{itemize}

    \item $\getplayer{t_0} = \odd, \getplayer{t_1} = \even$.
          \begin{itemize}
          \item Case there are $w_0 \in t_0^\bullet$ and $w_1 \in t_1^\bullet$
                such that $w_0 \gstut w_1$. Then \duplicator plays to such $w_0$ and
                $w_1$. 

                \emph{New configuration:} $((w_0,w_1),\checkmark)$.

          \item Case there are no $w_0 \in t_0^\bullet$ and $w_1 \in t_1^\bullet$
                such that $w_0 \gstut w_1$. 
              \begin{itemize}
              \item case there is some $w_0 \in t_0^\bullet$
                such that $w_0 \gstut t_0$. Then \duplicator plays to $w_0$ that
                is such while minimising $\measure{w_0,t_1,\dagger}$. 

                \emph{New configuration:} $((w_0,t_1),\dagger)$ if $u_1 = t_1$;
                else $((w_0,t_1),\checkmark)$.

	  \emph{Progress:} we first show that $u_1 = t_1$ implies
	  $c \in \{\dagger,\checkmark\}$.  Towards a contradiction,
	  assume $c = (0,t)$ for some $t$. By our invariant, this
	  implies $(t_0,t_1) \in V_{\even} \times V$. Since $(u_0,u_1)
	  \in \{(t_0,t_1),(t_1,t_0)\}$ and $u_1 = t_1$ we have $u_0
	  = t_0$. But then both $u_0 \in V_{\odd}$ and $u_0 \in
	  V_{\even}$. Contradiction.  Towards another contradiction,
	  assume $c = (1,t)$ for some $t$. By our invariant, this
	  implies $(t_0,t_1) \in V \times V_{\odd}$. This contradicts
	  $u_1 = t_1$ since $\getplayer{u_1} = \even$.

	  It therefore suffices to show $\measure{w_0,t_1,\dagger} <
	  \measure{t_0,t_1,\dagger}$. This follows from the fact
	  that we minimised $\measure{w_0,t_1,\dagger}$ and
	  Lemmata~\ref{lem:decrease_dist}
	  and~\ref{lem:decrease_dist_opponent}.

          \item case there is some $w_1 \in t_1^\bullet$
                such that $w_1 \gstut t_1$. Then \duplicator plays to $w_1$ that
                is such while minimising $\measure{t_0,w_1,\dagger}$. 

                \emph{New configuration:} $((t_0,w_1),\dagger)$ if $u_0 = t_0$;
                else $((t_0,w_1),\checkmark)$.

	  \emph{Progress:} using arguments, similar to those in the previous
          case, it follows that $c \in \{\dagger,\checkmark\}$.  

	  It therefore suffices to show $\measure{t_0,w_1,\dagger} <
	  \measure{t_0,t_1,\dagger}$. This follows from the fact
	  that we minimised $\measure{t_0,w_1,\dagger}$ and
	  Lemmata~\ref{lem:decrease_dist}
	  and~\ref{lem:decrease_dist_opponent}. \qedhere

          \end{itemize}
       \end{itemize}
    \end{itemize}

  \end{itemize}


\end{proof}
We next focus on proving that every pair of vertices related through the
governed stuttering bisimulation game are in fact governed stuttering bisimilar.
\begin{proposition}
For all $v,w \in V$ if $v \gstutg w$ then $v \gstut w$.
\end{proposition}

\begin{proof} We prove the contrapositive of the statement, \ie for all
$v,w \in V$, if $v \not\gstut w$, then also $v \not\gstutg w$. Let
$v \not\gstut w$. By Corollary~\ref{cor:fmist_fixpoint}, then also
$(v,w) \notin \nu \mathcal{F}$. By the Tarski-Kleene fixpoint approximation theorem,
we thus have $(v,w) \notin \bigcap\limits_{k \ge 1} \mathcal{F}^k(V \times V)$.
Let $R^k$ denote the relation $\mathcal{F}^k(V \times V)$; \ie, $R^k$ is the
relation obtained by applying the operator $\mathcal{F}$ $k$-times.
Note that
because of monotonicity, 
$R^k = \bigcap_{l \le k} R^l$.
We next prove, using induction, that for all $k \ge 1$:
\[
\tag{IH}
\label{eq:IH}
\begin{array}{l}
\text{\spoiler wins the governed stuttering bisimulation game} \\
\text{for all configurations 
$((u_0,u_1),c)$ for which
$(u_0,u_1) \notin R^k$}
\\
\end{array}
\]

\begin{itemize}

\item Base case $k = 1$. Observe that $R^1 = \{ (v,w) \in V \times V
\mid \prio{v} = \prio{w} \}$.  \spoiler wins the governed stuttering
bisimulation game for all configurations $((u_0,u_1),c)$ satisfying $(u_0,u_1)
\notin R^1$: all plays starting in such a configuration trivially
violate \duplicator's winning condition.

\item Inductive step. Assume that the statement holds for some $k \ge 1$.  
Pick an arbitrary position $(u_0,u_1)$ for which $(u_0,u_1) \notin R^{k+1}$ 
and let $c$ be an arbitrary challenge/reward. We must show that
\spoiler wins the governed stuttering bisimulation game for these.
Recall that we have $R^{k+1} \subseteq R^k$.

If $(u_0,u_1) \notin R^{k}$, then by the induction hypothesis, \spoiler
wins the governed stuttering bisimulation game from configuration
$((u_0,u_1), c)$.

Observe that by definition of
$\mathcal{F}$, we have for all $(v,w) \in R^k \setminus R^{k+1}$ that
there are $\player \in \{\even,\odd\}$ and
$\mathcal{U},\mathcal{T} \subseteq V_{/R^{k}}$ for which
\[
\tag{*}
\label{eq:star}
\text{
$\class{v}{R^{k}} \in \mathcal{U}\setminus \mathcal{T}$
but not
$v \forces{\player}{\mathcal{U}} \mathcal{T} \Leftrightarrow
w \forces{\player}{\mathcal{U}} \mathcal{T}$.
}
\]
Let $\player, \mathcal{U}, \mathcal{T}$ be such that~\eqref{eq:star}. We
focus on the case $\player = \even$; the case that $\player = \odd$ is fully
dual.
Assume that $v \forces{i}{\mathcal{U}} \mathcal{T}$ and
not $w \forces{\player}{\mathcal{U}} \mathcal{T}$; the case
in which not $v \forces{\player}{\mathcal{U}} \mathcal{T}$ but
$w \forces{\player}{\mathcal{U}} \mathcal{T}$ is symmetric.
Note that we can assume that
$\mathcal{T} \cap \mathcal{U} = \emptyset$, as
$v \forces{\player}{\mathcal{U}} \mathcal{T}$ iff
$v \forces{\player}{\mathcal{U} \setminus \mathcal{T}} \mathcal{T}$ for
any $\mathcal{U},\mathcal{T}$. We may therefore also simplify
$\mathcal{U} \setminus \mathcal{T}$ to $\mathcal{U}$.

Let $\sigma \in \strategy{\player}$ be the (memoryless) strategy underlying
$v \forces{\player}{\mathcal{U}} \mathcal{T}$.
Using $\sigma$, we construct a winning strategy for 
\spoiler for configuration $((v,w),c)$.
We first show that \spoiler can
invariantly move between configurations $((t_0,t_1),c)$ that satisfy the following
property $\Phi$:
\[
\begin{array}[t]{l}
\text{If } \class{t_0}{R^{k}} = \class{t_1}{R^{k}} \text{ then}\\

\begin{cases}
t_0 \forces{\sigma}{\mathcal{U}} \mathcal{T} \text{ but not }
t_1 \forces{\player}{\mathcal{U}} \mathcal{T}  & \\

c = (0,t) \text{ implies } t_0 \in V_{\even} \text{ and }
t_0 \steps{\sigma} t \\

c = (1,t) \text{ implies } 
 t_1 \in V_{\odd}, t \in \post{t_1} \text{ and
not } t \forces{\player}{\mathcal{U}} \mathcal{T}
\end{cases}
\end{array}
\]
%
%
%
%
%
%
Let $((t_0,t_1),c)$ be a configuration for which $\Phi$ holds. For all
such configurations \spoiler's move in step~1 of a round is to play from
$(t_0,t_1)$; \ie \spoiler does not switch positions.
We distinguish three main cases, showing that \duplicator has no other
option than to choose a new configuration that satisfies $\Phi$.

\begin{enumerate}
\item Case $c \in \{ \dagger,\checkmark\}$.
We furthermore distinguish cases based on the players of $t_0$ and $t_1$.

  \begin{itemize}
  \item Case $\getplayer{t_0} = \getplayer{t_1} = \even$. Since
        $t_0 \forces{\sigma}{\mathcal{U}} \mathcal{T}$, \spoiler proposes to move from
        $t_0$ to $\sigma(t_0)$. \duplicator proposes $u_1 \in t_1^\bullet$.
        Observe that not $u_1  \forces{\player}{\mathcal{U}} \mathcal{T}$. \duplicator
        then can propose to continue in:
        $((\sigma(t_0),u_1),\checkmark)$,
        $((t_0,u_1), (0,\sigma(t_0)))$, or
        $((\sigma(t_0),t_1), \dagger)$. Clearly, all new configurations
        satisfy $\Phi$.
          
  \item Case $\getplayer{t_0} = \getplayer{t_1} = \odd$.
        \spoiler proposes to move from
        $t_1$ to $u_1$ such that not $u_1 \forces{\player}{\mathcal{U}} \mathcal{T}$.
        Such $u_1$ exists. \duplicator proposes $u_0 \in t_0^\bullet$.
        Observe that $t_0 \steps{\sigma} u_0$. \duplicator then proposes to 
        continue in:
        $((u_0,u_1),\checkmark)$,
        $((t_0,u_1),\dagger)$, or
        $((u_0,t_1),(1,u_1))$. All new configurations satisfy $\Phi$.

  \item Case $\getplayer{t_0} = \even$, $\getplayer{t_1} = \odd$.
        Since $t_0 \forces{\sigma}{\mathcal{U}} \mathcal{T}$
        \spoiler proposes to move from
        $t_0$ to $\sigma(t_0)$ and from $t_1$ to $u_1$ such that not 
        $u_1 \forces{\player}{\mathcal{U}} \mathcal{T}$.
        Note that such $u_1$ exists.  \duplicator then proposes to 
        continue in:
        $( (\sigma(t_0),u_1),\checkmark)$,
        $( (t_0,u_1), (0,\sigma(t_0)))$, or
        $( (\sigma(t_0),t_1),(1,u_1))$.
        Again, all new configurations satisfy $\Phi$.

  \item Case $\getplayer{t_0} = \odd$, $\getplayer{t_1} = \even$.
        \duplicator proposes to move from
        $t_1$ to $u_1$ and from $t_0$ to $u_0$. Since $\player = \even$, we have 
        $t_0 \steps{\sigma} u_0$ and because of $\Phi$, we have not 
        $u_1 \forces{\player}{\mathcal{U}} \mathcal{T}$.  \duplicator then 
        proposes to continue in:
        $((u_0,u_1), \checkmark)$,
        $((t_0,u_1), \dagger)$, or
        $((u_0,t_1), \dagger)$. 
        All new configurations satisfy $\Phi$.

  \end{itemize}

\item Case $c = (0,t)$. Because of $\Phi$, we have
      $t_0 \forces{\sigma}{\mathcal{U}} \mathcal{T}$ and $\getplayer{t_0} = \even$. Then \spoiler plays from configuration $(t_0,t_1)$.
      We furthermore distinguish cases based on the owner of $t_1$.
  \begin{itemize}
  \item Case $\getplayer{t_1} = \even$. \spoiler proposes to 
        move from $t_0$ to $t$.
        \duplicator proposes $u_1 \in t_1^\bullet$. Observe that not
        $u_1 \forces{\player}{\mathcal{U}} \mathcal{T}$. \duplicator then proposes
        to continue in:
        $((t,u_1),\checkmark)$,
        $((t_0,u_1),(0,t))$, or
        $((t,t_1),\dagger)$

  \item Case $\getplayer{t_1} = \odd$. \spoiler proposes to move from 
        $t_0$ to $t$ and from $t_1$ to $u_1$ such that not 
        $u_1 \forces{\player}{\mathcal{U}} \mathcal{T}$.
        Such $u_1$ exists.  \duplicator then proposes to continue in:
        $((t,u_1),\checkmark)$,
        $((t_0,u_1),(0,t))$, or
        $((t,t_1), \checkmark)$

  \end{itemize}
  In both cases, the new rounds satisfy $\Phi$

\item Case $c = (1,t)$. Because of $\Phi$, we have not
      $t_1 \forces{\sigma}{\mathcal{U}} \mathcal{T}$ and $\getplayer{t_1} = \odd$.
      Then \spoiler plays from configuration $(t_0,t_1)$.
      We furthermore distinguish cases based on the owner of $t_0$.
  \begin{itemize}
  \item Case $\getplayer{t_0} = \even$. \spoiler proposes to move from
          $t_0$ to $\sigma(t_0)$ and from $t_1$ to $t$.
          \duplicator then proposes to continue in:
          $((\sigma(t_0),t),\checkmark)$,
          $((t_0,t),\checkmark)$, or
          $((\sigma(t_0),t_1),(1,t))$.

  \item Case $\getplayer{t_0} = \odd$. \spoiler proposes to move from
          $t_1$ to $t$. \duplicator proposes $u_0 \in t_0^\bullet$.
          Observe that $t_0 \steps{\sigma} u_0$. \duplicator then proposes to 
          continue in:
          $((u_0,t),\checkmark)$,
          $((t_0,t),\dagger)$, or
          $((u_0,t_1),(1,t))$.

  \end{itemize}

\end{enumerate}
We next observe that for any $(t_0,t_1)$ for which $t_0, t_1$ meet the premiss of
$\Phi$, but not the conclusion,
\spoiler can, in a single round, move to a configuration that either does not
meet $\Phi$'s premiss or to one that meets $\Phi$'s conclusion.
More specifically, suppose that $\class{t_0}{R^k} = \class{t_1}{R^k}$
but one of the following holds:
\begin{enumerate}
\item $t_1 \forces{\sigma}{\mathcal{U}} \mathcal{T}$ but not
$t_0 \forces{\player}{\mathcal{U}} \mathcal{T}$;

\item $c = (0,t)$ implies $t_0 \notin V_{\even}$ or not $t_0 \steps{\sigma} t$;

\item $c = (1,t)$ implies $t_1 \notin V_{\odd}$, $t \notin \post{t_1}$, or 
$t \forces{\player}{\mathcal{U}} \mathcal{T}$.
\end{enumerate}
Whenever we are in case 1, \spoiler switches positions in step~1 of a round
and follows the strategy outlined above. Whenever we are in case 2 or 3,
\spoiler drops challenge $c$ in step~1 and plays as if $c \in \{\dagger,\checkmark\}$.
In all three cases, \duplicator is rewarded a $\checkmark$ as the 
new challenge at the end of the round  and $\Phi$ holds trivially.\medskip

Summarising, we find that for configurations $((t_0,t_1),c)$ for 
which both~\eqref{eq:star} and $\Phi$ hold, \spoiler can move to another
configuration that either meets $\Phi$ or is such that $\Phi$'s premiss
is violated. For configurations $((t_0,t_1),c)$ for which~\eqref{eq:star} but
not $\Phi$ holds, \spoiler can move in a single round to a configuration 
for which she can
henceforth maintain $\Phi$ as an invariant or for which $\Phi$'s premiss
is violated.  

We finally argue that when \spoiler plays according to the above
strategy, she wins all plays.  Observe that we only need to show
this for all infinite plays that pass only through positions $(u,t)$
for which $(u,t) \in R^k$; for those plays that at some point pass
along a position $(u,t) \notin R^k$, our induction hypothesis yields
a winning strategy for \spoiler.\medskip

  Let $(v_0,w_0)\ (v_1,w_1)\ (v_2,w_2) \dots$ be an infinite sequence
  of positions on an infinite play $\pi$ that is allowed by \spoiler's
  strategy, such that for all $l$, $\class{v_l}{R^k} = \class{w_l}{R^k}$.
  Towards a contradiction, assume that \duplicator wins $\pi$.
  Observe that $\class{v_l}{R^k} = \class{w_l}{R^k}$ implies that
  $\prio{v_l} = \prio{w_l}$ for all $(v_l,w_l) \in R^k \setminus
  R^{k+1}$; therefore we can only arrive at a contradiction by
  showing that \duplicator earns a finite number of $\checkmark$
  rewards along $\pi$.

  By invariant $\Phi$, for all positions $(v_l,w_l)$, for $l \ge 1$,
  we have $v_l \forces{\sigma}{\mathcal{U}} \mathcal{T}$. 
  Let $\delta(v_l,w_l)$ denote the length of the longest path 
  from $v_l$ to reach $\mathcal{T}$ when playing according to $\sigma$. Note that
  $\delta$ is finite and decreases along the positions in $\pi$, but never
  reaches $0$, as all vertices remain in $\mathcal{U}$. This means
  that for some $m$, we have $\delta(v_m,w_m) = \delta(v_n,w_n)$ for
  all $n \ge m$. Fix
  this $m$. Moreover, there must be some $u$ such that:
  \[
  u \forces{\sigma}{\mathcal{U}}{\mathcal{T}} \wedge
  \forall{n \ge m} \forall{(v_n,w_n) \in \pi} v_n = u 
  \]
  But this means that, once \spoiler's strategy reaches the position
  containing $u$, all remaining $\checkmark$'s earned by \duplicator
  must be due to \spoiler switching positions or discarding a challenge
  in step~1 of each new round. 
  As we explained, \spoiler switches positions and/or drops a challenge only
  in the first round when starting in a configuration that does not satisfy
  $\Phi$; she never does so afterwards. Therefore, \duplicator earns no
  $\checkmark$ rewards when $\pi$ reaches a configuration containing a position
  with $u$. But then \duplicator earns only a finite number of $\checkmark$ rewards
  along $\pi$, contradicting the assumption that \duplicator wins $\pi$.

  Therefore, \spoiler has a strategy to win any configuration $((u_0,u_1),c)$
  for which $(u_0,u_1) \notin R^k$.
  \qedhere

\end{itemize}

\end{proof}

\end{document}